\documentclass[ser, preprint]{imsart}
\usepackage{amsthm}
\usepackage{amsmath}
\usepackage{amssymb}
\usepackage{graphicx}
\usepackage{upgreek}
\usepackage{mathrsfs}
\usepackage{framed}
\usepackage{bm}
\usepackage{bbm}
\usepackage{float}
\usepackage{dsfont}
\usepackage[T1]{fontenc}
\usepackage[latin9]{inputenc}
\floatstyle{ruled}
\startlocaldefs

\theoremstyle{plain}
\newtheorem{thm}{Theorem}
  \theoremstyle{plain}
  \newtheorem{prop}[thm]{Proposition}
  \theoremstyle{plain}
  \newtheorem{assumption}[thm]{Assumption}
  \theoremstyle{remark}
  \newtheorem{rem}[thm]{Remark}
  \theoremstyle{definition}
  \newtheorem{example}[thm]{Example}
  \theoremstyle{plain}
  \newtheorem{lem}[thm]{Lemma}
  \theoremstyle{definition}

  \newfloat{algorithm}{hbp}{lop}
\providecommand{\algorithmname}{Algorithm}
\floatname{algorithm}{\protect\algorithmname}

\numberwithin{equation}{section}
\numberwithin{thm}{section}
\newcommand{\E}{\mathbb{E}}

\newcommand{\state}{\mathcal{X}}
\newcommand{\ddata}{d}

\newcommand{\ux}[2]{{\mathcal{T}_{{#1}}^{#2}}}
\newcommand{\lx}[2]{{\mathcal{B}_{{#1}}^{#2}}}
\newcommand{\pux}[2]{{\mathcal{\hat{T}}_{{#1}}^{#2}}}
\newcommand{\plx}[2]{{\mathcal{\hat{B}}_{{#1}}^{#2}}}
\newcommand{\ix}[2]{{\mathcal{I}_{{#1}}^{#2}}}
\newcommand{\pix}[2]{{\mathcal{\hat{I}}_{{#1}}^{#2}}}
\newcommand{\lb}{b_{\lx{}{}}}
\newcommand{\ub}{b_{\ux{}{}}}
\newcommand{\ubk}[1]{b_{\ux{}{},#1}}
\newcommand{\lbk}[1]{b_{\lx{}{},#1}}
\newcommand{\isEvt}{A}

\newcommand{\G}{\mathcal{N}}
\newcommand{\R}{\mathbb{R}}
\newcommand{\rp}{\mathbb{P}}

\newcommand{\hiddensection}[1]{\stepcounter{section}\section*{}}
\newcommand{\1}{\mathds{1}}

\newcommand{\N}{\mathbb{N}}

\newcommand{\C}{\mathcal{C}}

\newcommand{\coloneq}{\mathrel{\mathop:}=}

\newcommand{\norm}[1]{\big\|#1\big\|}

\newcommand{\pr}[2]{\big\langle#1,#2\big\rangle}
\newcommand{\envelope}{(\raisebox{-.5pt}{\scalebox{1.45}{\Letter}}\kern-1.7pt)}

\newcommand{\si}[1]{\lambda}

\newcommand{\jj}{\ell}

\global\long\def\var{{\rm var}}
 \global\long\def\norm#1{{\left|\left|#1\right|\right|}}
 \global\long\def\R{\mathbb{R}}
 \global\long\def\inlaw{\overset{\mathcal{L}}{=}}
 \global\long\def\ind{\1}

\endlocaldefs

\begin{document}

\begin{frontmatter}

% "Title of the Paper"
\title{Unbiased Monte Carlo: posterior estimation for intractable/infinite-dimensional models}

\runtitle{Unbiased Monte Carlo}

\begin{aug}
% indicate corresponding author with \corref{}
% \author{\fnms{John} \snm{Smith}\thanksref{a}\corref{}\ead[label=e1]{smith@foo.com}\ead[label=e2,url]{www.foo.com}}
% \address[a]{\printead{e1};\printead{e2}}

\author{\fnms{Sergios} \snm{Agapiou}\thanksref{a}\ead[label=e1]{Sergios.Agapiou@warwick.ac.uk}},
\author{\fnms{Gareth} \snm{O. Roberts}\thanksref{b}\ead[label=e2]{Gareth.O.Roberts@warwick.ac.uk}}
\and
\author{\fnms{Sebastian} \snm{J. Vollmer}\thanksref{c}\ead[label=e3]{Sebastian.Vollmer@stats.ox.ac.uk}}
\address[a]{Department of Statistics, University of Warwick\\ \printead{e1}}
\address[b]{Department of Statistics, University of Warwick\\ \printead{e2}}
\address[c]{Department of Statistics, University of Oxford\\ \printead{e3}}
\affiliation{University of Warwick \thanksmark{e1},\thanksmark{e2} and University of Oxford \thanksmark{e3}}

\runauthor{S. Agapiou, G. O. Roberts, S. J. Vollmer}

\end{aug}

\begin{abstract}
We provide a general methodology for unbiased estimation for intractable stochastic
models. We consider situations where the target distribution can be written as an appropriate
limit of distributions, and where conventional approaches require truncation of such a representation
leading to a systematic bias. For example, the target distribution might be representable as the
$L^2$-limit of a basis expansion in a suitable Hilbert space; or alternatively the distribution of interest
might be representable as the weak limit of a sequence of random variables, as in MCMC.
Our main motivation comes from infinite-dimensional models which can be
parameterised in terms of a series expansion of basis functions (such as that given by a Karhunen-Loeve
expansion). We consider schemes for direct unbiased estimation along such an expansion, as well
as those based on MCMC schemes which, due to their dimensionality, cannot be directly implemented, but which can be
effectively estimated unbiasedly. For all our methods we give theory to justify the numerical
stability for robust Monte Carlo implementation, and in some cases we illustrate using simulations.
Interestingly the computational efficiency of our methods is usually comparable to simpler
methods which are biased. Crucial to the effectiveness of our proposed methodology is the
construction of appropriate couplings, many of which resonate strongly with the Monte Carlo constructions
used in the coupling from the past algorithm and its variants.
\end{abstract}


% history:
% \received{\smonth{1} \syear{0000}}

%\tableofcontents

\end{frontmatter}

%%%%%%%%%%%%%% INTRODUCTION %%%%%%%%%%%%%%%%%

\section{Introduction}
Bayesian analyses of complex models often lead to posterior distributions which 
are only available indirectly as an appropriate limit of a sequence of probability measures. A classical example of this is Markov Chain Monte Carlo (MCMC), which constructs an algorithm to access the posterior distribution which involves creating Markov chains with the required limiting distribution. Rather different examples come from infinite-dimensional models, for example arising in inference for continuous-time stochastic processes, and in inverse problems where the quantity to be inferred is naturally expressed as a function in an appropriate Hilbert space. In these examples, the exact representation of the posterior distribution is via an infinite sum (perhaps representing a basis expansion) or the limit of a sequence of approximations perhaps derived from time discretisations. Thus, in these contexts we have an indirect representation of the posterior distribution.

The conventional approach to such an indirect representation is to truncate: 
\begin{itemize}
\item to run the MCMC for long enough; 
\item to choose a fixed fine time-discretisation; 
\item or to take  sufficiently many terms in the series expansion. 
\end{itemize}
The main problem with this general approach is that the accuracy of the approximation produced is highly application-specific and very difficult to analyse.

It is a common misconception that {\em exact} methods, which avoid truncation approximations entirely, are either impossible or prohibitively computationally expensive (see \cite{peluchetti2012study} 
for some examples involving simulation of SDEs).
Although stochastic simulation directly from the posterior distribution is generally not feasible,  it turns out to be very commonly feasible and practical to obtain unbiased estimates for any arbitrary posterior expected functional of interest. This is the focus of the present paper, which builds on the contributions of \cite{RheePHD}.
 
 Fundamental to the success of these methods is the construction of suitable couplings to ensure our estimators have finite variances. Much of these constructions resonate with the huge body of literature inspired by the coupling from the past algorithm of Propp and Wilson \cite{PW96}, although crucially our methods are substantially more general as we do not require the strong {\em coalescent} couplings needed for coupling from the past.

Although we shall state most of our results quite generally, our main applications will be in the area of Bayesian inverse problems.
We construct unbiased estimators in four settings: 
\begin{itemize}
\item for linear Gaussian conjugate Bayesian inverse problems (section \ref{sec:Linear} - bias due to discretisation);
\item for chains in a fixed state space, which posses a simulatable contracting coupling between runs started at different positions (section \ref{sec:wasserstein} - bias due to using finite time distributions);
\item for non-linear inverse problems with uniform series priors using the independence sampler (section \ref{sec:is} - bias due to discretisation and finite time);
\item for measures with log-Lipschitz densities with respect to infinite dimensional Gaussians 
using the pCN algorithm (section \ref{sec:trans} - bias due to discretisation and finite time).
\end{itemize}

There are many operational choices in our procedures, and we have only just begun exploring all the options. Optimisation of our procedures is therefore an important and difficult question which leads on from our work here. From the examples we have considered here, we have however been surprised by the apparent efficiency of essentially ad-hoc choices for algorithm parameters. Thus, 
our methods seem very promising as practical and general approaches which circumvent the systematic error of existing approaches.

Although our work  is significantly more technical in nature than \cite{RheePHD}, we see our main contributions here as methodological rather than  mathematical, and in this light have tried to keep technicalities to a minimum, particularly in the main body of the paper. For instance, we refrain from expressing or proving our results for the most general Hilbert space-valued functions, even though a generalisation to this context is completely straightforward.

\subsection{Overview of existing results\label{sec:OverviewOfRhees}}

We now briefly outline the recent results by Chang-han Rhee and Peter
Glynn \cite{rhee2012new,RheePHD}, which we extend in the following
sections (see also work by Don McLeish \cite{DML11}). The objective
is to efficiently simulate an unbiased estimator of the expectation
of a real valued random variable $Y$. We consider settings in which the exact
simulation of $Y$ is impossible due to the infinite cost associated
with generating an exact sample, thus in order to perform a Monte Carlo
simulation one needs to use approximations $Y_{i}$ of $Y$. This
introduces a bias in the Monte Carlo estimator of the expectation,
which in turn results in suboptimal rates of convergence with respect
to the computational budget $c$.
In particular, instead of the optimal $\mathcal{{O}}(c^{-1/2})$ rate
of convergence, we get slower rates. In the aforementioned works,
the goal is twofold: first to construct unbiased estimators of the
expectation of interest using an appropriate combination of biased
ones, and second to find conditions which secure that the variance
and the computational cost of the constructed estimator are such that
the optimal rate of convergence with respect to the budget $c$ is
achieved.

The starting point is a neat randomisation idea for unbiased estimation
of infinite sums, which traces back to John von Neumann and Stanislaw
Ulam in the context of matrix inversion \cite{FL50,WW52}. The idea
was more recently employed by Peter Glynn in the setting of time
integral estimation \cite{PG83}. Assume that the approximations $Y_{i}$
satisfy $\E(Y_{i})\to\E(Y)$ as $i\rightarrow\infty$. Then one can
express the expectation of $Y$ as a telescoping sum 
\[
\E(Y)=\sum_{i=0}^{\infty}\E(Y_{i}-Y_{i-1}),
\]
where $Y_{-1}=0$ by convention. Provided the approximations are good
enough so that Fubini's theorem applies, this suggests that an unbiased
estimator for $\E(Y)$ is the sum $\sum_{i=0}^{\infty}(Y_{i}-Y_{i-1})$.
However, this estimator cannot be generated in finite time, so the
idea is to use a random truncation point $N$ and correct for the
introduced bias. Indeed, let $N$ be an integer-valued random variable
which is independent of the random approximations $Y_{i}$ and is
such that $P(N\geq i)>0$ for all $i\in\mathcal{\mathbb{{N}}}.$ Then,
letting $\Delta_{i}=Y_{i}-Y_{i-1},$ and again assuming that the approximations
are good enough so that we can interchange expectation with summation,
we have that 
\[
\E\left[\sum_{i=0}^{N}\frac{\Delta_{i}}{P(N\geq i)}\right]=\E\left[\sum_{i=0}^{\infty}\frac{1_{\{N\geq i\}}\Delta_{i}}{P(N\geq i)}\right]=\sum_{i=0}^{\infty}\E(\Delta_{i})=\E(Y),
\]
so that the estimator 
\begin{equation}
Z\coloneq\sum_{i=0}^{N}\frac{\Delta_{i}}{P(N\geq i)}\label{eq:estimator}
\end{equation}
is unbiased.

In order for the estimator $Z$ to be practical, we need to also have
that its variance, ${\rm {var}}(Z)$, as well as the expected work
required to generate a copy of it, $\E(\tau)$, are finite. Letting
$t_{i}$ be the expected incremental effort required to calculate
$Y_{i}$, we have that 
\begin{equation}\label{eq:expectedTime}
\E(\tau)=\E\left(\sum_{i=0}^{N}t_{i}\right)=\sum_{i=0}^{\infty}t_{i}P(N\geq i),
\end{equation}
while in \cite{rhee2012new,RheePHD} it is shown that 
\[
{\rm {var}}(Z)=\sum_{i=0}^{\infty}\frac{\beta_{i}}{P(N\geq i)},
\]
where $\beta_{i}=\mathcal{{O}}(\E[(Y-Y_{i})^{2}]).$ It is hence apparent
that there is a competition between $P(N\geq i)$ decaying fast enough
so that the expected work required to generate $Z$ is finite, but
not too fast so that ${\rm {var}(Z)}$ is also finite. In order to
obtain that both the expected work and the variance of the estimator
are finite, the rate of convergence of $\E[(Y-Y_{i})^{2}]$ needs
to be faster than the rate at which the expected incremental
effort $t_{i}$ goes to $\infty $.

The following proposition is proved in \cite{RheePHD} and
is very useful for verifying the unbiasedness and finite variance
of the proposed estimator. Here and elsewhere, we use the notation
$\norm h_{2}\coloneq(\E[h^{2}])^{\frac{1}{2}}$. 
\begin{prop}
(Proposition 6, \cite{RheePHD}) \label{prop:rhee}Suppose that $\left(\Delta_{i}:i\ge0\right)$
is a sequence of real-valued random variables and let $N$ be an integer-valued random variable which is independent of the $\Delta_i$'s and satisfies $\rp(N\geq i)>0$ for all $i\geq0$. Assume that 
\[
\sum_{i\leq l}\frac{\norm{\Delta_{i}}_{2}\norm{\Delta_{l}}_{2}}{\rp\left(N\ge i\right)}<\infty.
\]
Then $Y_{n}\coloneq\sum_{i=0}^{n}\Delta_{i}$ converges in $L^{2}$
to a limit $Y\coloneq\sum_{i=0}^{\infty}\Delta_{i}$ as $n\rightarrow\infty$.
Let $\alpha=\E Y$$\left(=\lim_{n\rightarrow\infty}\E Y_{n}\right)$
%and each $\tilde{\Delta}_{i}$ be an independent copy of $\Delta_{i}$.
and suppose that for all $i$, ${\tilde \Delta _i}$  is a copy of $\Delta_i$ such that $\{\tilde{\Delta} _i\}$ are mutually independent.
Then $\tilde{Z}\coloneq\sum_{i=0}^{N}\frac{\tilde{\Delta}_{i}}{\rp\left(N\ge i\right)}$
is an unbiased estimator for $\alpha$ with finite second moment 
\[
\E\tilde{Z}^{2}=\sum_{i=0}^{\infty}\frac{\tilde{\nu}_{i}}{\rp\left(N\ge i\right)},
\]
where $\tilde{\nu}_{i}=\text{{\rm var}}(\Delta_{i})+\left(\alpha-\E Y_{i-1}\right)^{2}-\left(\alpha-\E Y_{i}\right)^{2}.$ 
\end{prop}
\begin{rem}
In Proposition \ref{prop:genrhee}, we generalise  Proposition \ref{prop:rhee} to cover estimation of expectations of Hilbert space-valued random variables $Y$. Nevertheless, in order to avoid overcomplicating our presentation, we state and prove our results for real-valued random variables $Y$ and only comment on their applicability in the more general Hilbert space setting. 
\end{rem}
Under the assumption that both $\var(Z)$ and $\E(\tau)$ are finite,
Glynn and Whitt's results on general estimators imply that a central
limit theorem holds 
\begin{equation}
c^{1/2}({\hat \alpha}(c)-\E(Y))\Rightarrow\sqrt{\E(\tau)\var(Z)}N(0,1),\label{eq:centrallimit}
\end{equation}
where ${\hat \alpha }(c)$ is the Monte Carlo estimator produced from independent
replicates of $Z$ that can be generated after $c$ units of computer
time \cite{GW92}. This immediately gives that the estimator converges
at the optimal square root rate. Furthermore, the above central limit
theorem supports theoretically the intuition that the product of the
variance and the expected work is a good measure of efficiency of the estimator
and consequently suggests that the choice of distribution for $N$
can be optimised by minimising this product.

In the work of Rhee and Glynn \cite{rhee2012new,RheePHD}, this programme
has been developed and carried out in two general settings. The first
setting is simulation of SDEs, in which these ideas are directly applicable
to many of the available discretisation schemes. An important observation
in this setting is that for lower order schemes, like the Euler-Maruyama
discretisation, this methodology does not work since the convergence
of $E[(Y-Y_{i})^{2}]$ is not quick enough compared to the increase
in the cost of producing $Y_{i}$. On the other hand, with respect
to the bias aspect of the problem, there is no need to use discretisation
schemes of particularly high order, since for example the Millstein
scheme is already enough to secure the optimal square root convergence
rate of the Monte Carlo estimator. The second setting is the study
of ergodic Markov chains, where the aim is to estimate expectations
with respect to the invariant measure and the finite-time distributions
are used to define the approximations $Y_{i}$. In this setting the
theory is not immediately applicable, since although the finite-time distributions
converge to the invariant measure, in general the random variables
$Y_{n}$ defined through the outcome of the Markov chain, may not
converge in the $L^{2}$ sense. For this reason one needs to construct
an appropriate coupling to enable the sequence of approximations to converge
in $L^{2}$ and thus to permit the application of  Proposition \ref{prop:genrhee}. In \cite{RheePHD}, such couplings are constructed for
uniformly ergodic, contracting and Harris chains (see subsection \ref{sec:rheemc}
below). 
 
In infinite-dimensional contexts (such as those arising in Bayesian inverse problems) it is usually impossible to implement the infinite-dimensional MCMC algorithms required to sample from the target distribution (though see \cite{beskos2006exact} for an example where it can be done).

A rather different application of the ideas of unbiasing by taking random differences, has been introduced by \cite{2013HoangComplexityMCMC,TeckentrupMultilevel2013, DG14}, which build on the Multilevel Monte Carlo (MLMC) method of Mike Giles \cite{MG08}. 
This method makes substantial progress in the construction of algorithms which unbiasedly estimate chosen finite-dimensional summaries from infinite-dimensional MCMC methods. However, these methods do not avoid bias due to Markov chain burn-in. In the present paper, we will provide practical unbiasing methods which circumvent bias, either from the need for finite-dimensional approximation, or from Markov chain burn in.

\subsection{Glynn and Rhee's results for exact estimation in the context of ergodic Markov chains}

\label{sec:rheemc} Before moving on with the presentation of our
results, we briefly recall the methodology of \cite{RheePHD} for
constructing an appropriate coupling in the setting of uniformly ergodic
Markov chains; we will build our extension to the MCMC in function
space setting on this methodology. Let $X=\left\{ X_{n}\right\} _{n\in\mathbb{{N}}}$
be a Markov chain in a state space $\state$, with transition probabilities
$P(x,A)$ and invariant distribution $\pi$. A \emph{uniformly recurrent}
Markov chain is one for which there exists a probability measure $\nu$
on $\state$, a constant $\lambda>0$ and an integer $m\geq1$, such that
\[
P^{m}(x,B)\geq\lambda\nu(B),
\]
for any $x\in \state$ and any measurable $B\subset \state$. It is well known
that a uniformly recurrent Markov chain  is uniformly ergodic and hence converges to its invariant
distribution \cite{RR04}, however this does not guarantee that $X$
converges in $L^{2}$. In order to find a coupling of $X_{n}$ and
$X_{n+1}$ such that they come closer in $L^{2}$ as $n$ increases,
the authors of \cite{RheePHD} define the random functions
\[
\varphi_{n}(\cdot)\coloneq I_{n}\xi_{n}+(1-I_{n})\phi_{n}(\cdot),
\]
where $I_{n}$ are independent and identically distributed Bernoulli
random variables with success probability $\lambda,$ $\xi_{n}$ are
independent random variables drawn according to $\nu,$ and $\phi_{n}$
are random functions representing the transition $Q(x,B)=\frac{P(x,B)-\lambda\nu(B)}{1-\lambda}$,
that is, $P(\phi_{n}(x)\in B)=Q(x,B)$. They then recursively express
the chain $X_{n}$ as $X_{n+1}=\varphi_{n}(X_{n})$, where $X_{0}=x.$
Since $\varphi_{n}(x)$ are independent and identically distributed
according to $P(x,\cdot)$, one can then define ${\tilde{{X_{n}}}}$
to be the backwards process 
\begin{eqnarray}
\tilde{X}_{n+1} & \coloneq & \varphi_{0}\circ\dots\circ\varphi_{n}(x)\nonumber \\
 & \overset{\mathcal{L}}{=} & \varphi_{n}\circ\dots\circ\varphi_{0}(x)\label{eq:backwardProcesses-1}\\
 & = & X_{n+1}.\nonumber 
\end{eqnarray}
Note that $\varphi_{n}$ is constant with positive probability $\lambda$,
so that with probability $1-(1-\lambda)^{n}$, at least one of the
$\varphi_{k},$ $k\in\left\{ 1,...,n\right\} $, is a constant (random)
function. The advantage of working with the backwards process is that
contrary to the forward process, if $\varphi_{n}$ is a constant function
then all $\tilde{{X}_{k}}$ for $k>n$ are equal to the same constant.
We hence have that as $n$ increases, with probability which goes
to 1, $\tilde{{X}}_{n}=\tilde{{X}}_{n+1}$.%, thus they come closer in $L^{2}.$ 

This is particularly useful for estimating the expectation of a bounded
function $f:\state\to\R$ with respect to the equilibrium distribution
$\pi$, $\E_{\pi}[f].$ An obvious choice of approximating
sequence in this setting is the sequence of the images under $f$
of the chain after a finite number of steps, hence we let $Y_{i}=f(X_{i}).$
Then
\begin{align*}
\E[(Y_{i}-Y_{i-1})^{2}] & =\E[(f(\tilde{X}_{i})-f(\tilde{X}_{i-1}))^{2}]\\
 & \leq\norm{f}_{\infty}^{2}P(\varphi_{j}\;\text{is not constant for all}\; j\leq i)\\
 & \leq\norm{f}_{\infty}^{2}(1-\lambda)^{i}.
\end{align*}
We thus have that the $Y_{i}$ converge in $L^{2}$ and the unbiasing programme
described in the previous subsection can be applied.

\subsection{Implementation of the backwards construction}

\label{ssec:implback}

At a high level, Rhee and Glynn's general approach is to represent
the chain using random functions $\varphi(x,W)$, where $W$ represents
all the randomness needed to simulate the transition. Then the evolution
of the chain is written as $X_{n}=\varphi_{n}\circ\dots\circ\varphi_{0}(x),$
where $\varphi_{i}(\cdot)=\varphi\left(\cdot,W_{i}\right)$ for some
independent identically distributed sequence $W_{i}.$ As described
above, the backwards technique consists in considering the chain $\tilde{X}_{n}=\varphi_{0}\circ\dots\circ\varphi_{n}(x)\inlaw X_{n}.$
Under appropriate assumptions (contraction or uniform ergodicity)
this technique turns the weak convergence of the chain $X_{n}$ to
its equilibrium distribution, to almost sure convergence of $\tilde{{X}}_{n}$
to a limiting random variable $\tilde{{X}}$. The chain $\tilde{{X}}$
is then used to obtain the approximations $Y_{i}=f(\tilde{X}_{i})$,
and hence the differences $\Delta_{i}=Y_{i}-Y_{i-1}=f(\tilde{X}_{i})-f(\tilde{X}_{i-1})$
which are used for generating the unbiased estimator $Z$. It is important
to observe, that completely independent copies of $\Delta_{i}$ at
different levels $i$ are used both for the algorithm and the analysis,
see Proposition \ref{prop:rhee}.

{We remark that the above described coupling is also used as
the fundamental idea in the \emph{coupling from the past} algorithm
for sampling perfectly from the invariant distribution of a Markov
chain \cite{PW96}. Furthermore, note that the backwards technique
in the above described form, has the disadvantage that in order to
pass from $\tilde{X}_{n}$ to $\tilde{X}_{n+1}$ we need to recompute
the whole chain. This means that in order to compute $\Delta_{i}$,
we first need to produce $Y_{i-1}$ and then start from scratch to
produce $Y_{i}$ (this discussion does not apply for producing $\Delta_{i}$
and $\Delta_{i+1}$ since they are assumed to be independent). For
the benefit of the reader and since no implementation details are
given in \cite{RheePHD}, we now describe a reasonable implementation.
This implementation is easier than the coupling from the past algorithm,
however the probabilistic construction is very similar.} {We
will later generalise this construction to cover sampling from infinite
dimensional target measures, using the finite-time distributions of
a hierarchy of Markov chains with state spaces of increasing dimension
(see sections \ref{sec:is} and \ref{sec:trans}).}

We start
by noticing that it is not necessary to construct $Y_{i}$'s that
have the correct distribution, but rather it suffices to generate
$\Delta_{i}$'s which have the correct expectation (this is silently
observed in \cite{RheePHD} but the authors do not seem to exploit
it). We present this in a more general setting, and in particular
we consider approximation levels $i$ that correspond to $a_{i}$ time steps, where
$\left\{ a_{i}\right\} _{i\in\N}$ is a strictly increasing sequence
of positive integer numbers. We will show later in section \ref{sec:contr} that the choice of $a_i$ 
has a huge impact on the efficiency of the estimator.

The random variables $\tilde{X}_{a_{i}}$ and $\tilde{X}_{a_{i-1}}$
needed to generate $\Delta_{i}$ when using the backwards technique,
are given as 
\begin{eqnarray*}
\tilde{X}_{a_{i}} & = & \varphi\bigg(\varphi\bigg(\dots\varphi\big(\varphi\left(x_{0},W_{a_{i}}\right),W_{a_{i}-1}\big)\dots,W_{2}\bigg),W_{1}\bigg),\\
\tilde{X}_{a_{i-1}} & = & \varphi\bigg(\varphi\bigg(\dots\varphi\big(\varphi\left(x_{0},W_{a_{i-1}}\right),W_{a_{i-1}-1}\big)\dots,W_{2}\bigg),W_{1}\bigg).
\end{eqnarray*}
The same set of random variables, can be generated sequentially as
the algorithm progresses. To do this, we introduce the chains $\ux{}{i},\lx{}{i}$
corresponding to the "top" and "bottom" approximation levels, respectively,
which appear in the definition of $\Delta_{i}$. We set 
\[
\ux{-a_{i}}{i}=x_{0},
\]
and to get $\ux{-a_{i-1}}{i}$we simulate until $-a_{i-1}$ , that
is we set 
\begin{equation}
\ux{-a_{i-1}}{i}=\varphi\bigg(\dots\varphi\big(\varphi(x_{0},W_{a_{i}}),W_{a_{i}-1}\big)\dots,W_{a_{i-1}+1}\bigg).\label{eq:backSeqInitialStep}
\end{equation}
We then set $\lx{-a_{i-1}}{i}=x_{0},$ and simulate
$\ux{}{i}$ and $\lx{}{i}$ jointly up to time $0$, hence obtaining
\begin{equation}
\begin{array}[t]{rl}
\ux{0}{i} & =\varphi\bigg(\varphi\bigg(\dots\varphi\big(\varphi\left(\ux{-a_{i-1}}{i},W_{a_{i-1}}\right),W_{a_{i-1}-1}\big)\dots,W_{2}\bigg),W_{1}\bigg),\\
\lx{0}{i}& =\varphi\bigg(\varphi\bigg(\dots\varphi\big(\varphi\left(x_{0},W_{a_{i-1}}\right),W_{a_{i-1}-1}\big)\dots,W_{2}\bigg),W_{1}\bigg).
\end{array}\label{eq:backSeqJointStep}
\end{equation}
Thus we have $\lx{0}{i}=\tilde{X}_{a_{i-1}}$ and $\ux{0}{i}=\tilde{X}_{a_{i}}$,
and can define $\Delta_{i}=f(\ux{0}{i})-f(\lx{0}{i})$. Furthermore,
observe that the direction of enumeration of the $W$'s does not matter
in this construction, since the $a_{i}$'s are a priori fixed and
can hence be generated as the algorithm progresses.

Alternatively, one can think of this construction in terms of couplings.
Let 
\[
P\left(x,\cdot\right)=\mathcal{L}(\varphi\left(x,W\right))
\]
be the transition kernel of $X_{i}$. Moreover, 
\[
K\left(\left(x,y\right),\cdot\right)\coloneq\mathcal{L}\left(\varphi\left(x,W\right),\varphi\left(y,W\right)\right)
\]
is a coupling of $P(x,\cdot)$ and $P(y,\cdot)$. This coupling allows
us to write (\ref{eq:backSeqInitialStep}) and (\ref{eq:backSeqJointStep})
as 
\begin{enumerate}
\item $\ux{-a_{i}}{i}=x_{0}$, then simulate according to $P$ up to $\ux{-a_{i-1}}{i}$; 
\item set $\lx{-a_{i-1}}{i}=x_{0}$, then simulate $(\ux{}{i},\lx{}{i})$ jointly
according to $K$ up to $\left(\ux{0}{i}, \lx{0}{i}\right).$ 
\end{enumerate}
Under the assumption that 
\begin{equation}
\sum_{i\le l}\frac{\norm{\Delta_{i}}_{2}\norm{\Delta_{l}}_{2}}{\rp(N\ge i)}<\infty,\label{eq:varest}
\end{equation}
which has to be verified for different classes of Markov chains, we
can define retrospectively the approximations $Y_{i}\coloneq\sum_{k=0}^{i}\Delta_{k},$
and apply Proposition 1 and more generally the programme developed
by Rhee and Glynn, to get an unbiased estimator of $\E_{\pi}[f]$ with
optimal cost.

\subsection{Notation}
We always denote the state space by $\state$, although we work under assumptions on the state space which differ between
sections. As stated earlier, we use the notation $\norm{h}_2=(\E[h^2])^\frac12$. We use $f$ to denote the function whose expectations we want to estimate and denote by $\E_{\pi}[f]$ the expectation of $f$ under a probability measure $\pi$. For two sequences $k_j$ and $h_j$ of positive real numbers,
$k_j\asymp h_j$ means
that $\frac{k_j}{h_j}$ is bounded away from zero and infinity as
$j\rightarrow\infty$, $k_j\lesssim h_j$ means that $\frac{k_j}{h_j}$
is bounded as $j\rightarrow\infty$,
and $k_j\sim h_j$ means that $\frac{k_j}{h_j}\rightarrow
1$ as $j\rightarrow\infty$.

\subsection{{Organisation of the paper}}
In section \ref{sec:Linear} we consider unbiased estimation of posterior expectations in Gaussian-conjugate
Bayesian linear inverse problems in Hilbert space. Since in this setting the posterior is also Gaussian, the source of the bias is only the discretisation. 

In section \ref{sec:wasserstein} we consider unbiased estimation of expectations with respect to the limiting distribution of an ergodic Markov chain in a fixed state space. Since we consider a fixed state space, the source of the bias is only the use of finite time distributions to approximate the limiting distribution (burn-in). Compared to the contracting chain setting of \cite{RheePHD}, we work under the weaker assumption that there exists a simulatable contracting coupling between runs of the chain started at different states (see Remark \ref{rem:pairwisecoupling}). 

In section \ref{sec:is} we consider estimation of posterior expectations in a nonlinear Bayesian inverse problem setting in function space, with uniform series priors and under assumptions which ensure the uniform ergodicity of the independence sampler at any fixed discretisation level of the state space. In this case the bias is both due to discretisation and burn-in. We achieve unbiased estimation by constructing a hierarchy of coupled independence samplers in state spaces of increasing dimension. 

In section \ref{sec:trans} we consider target measures which are absolutely continuous with respect to a Hilbert space Gaussian reference measure, under assumptions on the log-density which secure the existence of a simulatable contracting coupling of the pCN algorithm at any fixed discretisation level of the state space. In this case the bias is again  due to both discretisation and burn-in. We achieve unbiased estimation by constructing a hierarchy of coupled pCN algorithms in state spaces of increasing dimension. 

In section \ref{sec:contr} we present a comparison between the performance of the ergodic average of an MCMC run and the performance of the Monte Carlo estimator constructed using the unbiasing procedure. This is first done in a 1-dimensional Gaussian autoregression setting and then for a Bayesian logistic regression model. 

The main body of the paper ends with concluding remarks in section \ref{sec:concl}. All the proofs of our results, as well as the statements and proofs of some necessary intermediate results are contained in section \ref{sec:proofs}. Finally, in section \ref{sec:appendix} we provide the generalisation of Proposition \ref{prop:rhee} to Hilbert space-valued random variables, as well as some other required technical results.

%%%%%%%%%%%%%%%% LINEAR INVERSE PROBLEM %%%%%%%%%%%%%

\section{Unbiased estimation for Bayesian linear inverse problems\label{sec:Linear}}

In this section we consider the problem of estimating expectations
with respect to the posterior distribution arising in Bayesian linear
inverse problems in function space. We assume Gaussian prior and noise distributions, hence the 
posterior is available analytically and the only source of bias is the discretisation. We show that Glynn and Rhee's 
programme can directly be adapted in this setting to perform unbiased estimation of posterior expectations.
\subsection{Setup}We work in a separable Hilbert
space $(\state,\pr{\cdot}{\cdot},\norm{\cdot})$ and consider the
inverse problem of finding an unknown function $u\in\state$ from
a blurred, noisy observation $y$. In particular, we consider the
data model 
\[
y=Ku+\xi,
\]
where $\xi\sim\G(0,I)$ is additive Gaussian white noise and $K$
is the forward operator which is assumed to be linear and bounded.
We put a Gaussian prior $\mu_{0}=\G(0,\C_{0})$ on the unknown $u$,
where $\C_{0}$ is a positive definite, selfadjoint and trace class
linear operator. We make the following assumption on the operators
$K$ and $\C_{0}$. 
\begin{assumption}
\label{linearass1} \label{ass1} The linear operators $K$ and $\C_{0}$
commute with each other and $K^{\ast}K$ and $\C_{0}$ are mutually
diagonalizable with common complete orthonormal basis $\{e_{\jj}\}_{\jj\in\N}$
in $\state$. In particular, there exist $p\geq0,a>\frac{1}{2}$ such
that the eigenvalues of $K^{\ast}K$ and $\C_{0}$ decay as $\jj^{-4p}$
and $\jj^{-2a}$ respectively. 
\end{assumption}
In this diagonal setting it is straightforward to check that the posterior,
denoted by $\mu^{y}$, is also Gaussian almost surely with respect
to the joint distribution of $(u,y)$, \cite{ALS13}. We hence have
$\mu^{y}=\G(m,\C)$, where the mean and precision operator (inverse
covariance) are given by 
\begin{align}
\C^{-1} & =\C_{0}^{-1}+K^{\ast}K,\label{eq:prec}\\
\C^{-1}m & =K^{\ast}y.\label{eq:mean}
\end{align}

We make the following assumption concerning the observed data. 
\begin{assumption}
\label{ass2} We have a fixed realisation of the data, $y$, which
has the regularity of the noise, that is, there exist $c_{-},c_{+}>0$
such that for all $\jj\in\N$, $y_{\jj}=\pr{y}{e_{\jj}}\in(c_{-},c_{+})$. 
\end{assumption}
This assumption is reasonable, since given that $u\in\state$, in
order to have that the inverse problem is ill-posed and hence worthy
of consideration, the noise needs to be outside of the range of $K$.
This means that the noise needs to be the roughest part of the data.

Gaussianity suggests that we can in theory draw exactly from $u|y\sim\G(m,\C)$,
however in practice this is impossible to achieve in finite time due
to the infinite-dimensionality of the posterior. 
In the present setting, the approximation is achieved by considering
truncations of the Karhunen-Loeve expansion of $\G(m,\C)$. Let $x\sim\G(\mu,\Sigma)$
be a Gaussian random variable in $\state$, where $\mu\in\state$
and $\Sigma:\state\to\state$ is a selfadjoint, positive definite
and trace class linear operator in $\state$. Then the operator $\Sigma$
possesses a set of eigenvalue-eigenfuction pairs $\{\sigma_{\jj},\psi_{\jj}\}_{\jj\in\N},$
where $\sigma_{\jj}>0,\;\jj\in\N$ are summable and $\{\psi_{\jj}\}_{\jj\in\N}$
forms a complete orthonormal basis in $\state$. We can then write
$x=\sum_{\jj=1}^{\infty}(\mu_{\jj}+\sqrt{\sigma_{\jj}}\gamma_{\jj})\psi_{\jj},$
where $\mu_{\jj}=\pr{\mu}{\psi_{\jj}}$ and $\{\gamma_{\jj}\}_{\jj\in\N}$
are independent and identically distributed standard Gaussian random
variables in $\R$; this is the Karhunen-Loeve expansion of $x$,
\cite{RA90}.

In particular, the Gaussian random variable $u|y\sim\G(m,\C)$ can
be written as 
\[
u|y=\sum_{\jj=1}^{\infty}(m_{\jj}+\sqrt{{c_{\jj}}}\zeta_{\jj})e_{\jj},
\]
where $c_{\jj}$ are the eigenvalues of $\C$ (which is also diagonalizable
in the basis $\left\{ e_{\jj}\right\} _{\jj\in\N}$), $\zeta_{\jj}$
are independent and identically distributed standard normal random
variables and $m^{\jj}=\pr{m}{e_{\jj}}$. One can then define the
approximations $u^{i}|y^{i}$ of $u|y$ at level $i\in\N$, by truncating
its Karhunen-Loeve expansion to the first $j_{i}$ terms, 
\[
u^{i}|y^{i}\coloneq\sum_{\jj=1}^{j_{i}}(m_{\jj}+\sqrt{{c_{\jj}}}\zeta_{\jj})e_{\jj},
\]
where $\{j_{i}\}_{i\in\N}$ is an increasing sequence of positive
integers. Using equations (\ref{eq:mean}) and (\ref{eq:prec}), together
with Assumption \ref{linearass1}, we get that 
\begin{equation}
u^{i}|y^{i}=\sum_{\jj=1}^{j_{i}}\frac{\jj^{-2p}y_{\jj}}{\jj^{2a}+\jj^{-4p}}e_{\jj}+\sum_{\jj=1}^{j_{i}}\frac{\zeta_{\jj}}{(\jj^{2a}+\jj^{-4p})^{\frac{1}{2}}}e_{\jj}.\label{eq:tudr}
\end{equation}

Approximating expectations with respect to the posterior $\mu^y$, by expectations with respect to the laws $\mu^y_j$ of 
the truncated Karhunen-Loeve expansion, introduces a bias. 
In the next subsection, we demonstrate
how Glynn and Rhee's unbiased estimation programme for SDE's (see
section \ref{sec:OverviewOfRhees}), can be applied directly in the
setting of linear inverse problems to obtain unbiased estimates of
expectations with respect to the posterior $\mu^{y}$.

\subsection{Main results}
Suppose that we want to estimate $\E_{\mu^{y}}[f]=\E[Y]$,
where $Y\coloneq f(u|y)$ and $f:\state\to\R$ is $s$-H\"older continuous
for some $s\in(0,1].$ We define the approximations $Y_{i}=f(u^{i}|y^{i})$
for $i\in\N$ and as in section \ref{sec:OverviewOfRhees} the differences
$\Delta_{i}=Y_{i}-Y_{i-1},$ where $Y_{-1}\coloneq0$. We make the
following assumption which will be needed for controlling the expected
computing time of the proposed estimator.
\begin{assumption}
\label{ass:lincost} The expected computing time $t_{i}$ for generating
$\Delta_{i}$ satisfies 
\[
t_{i}\lesssim j_{i}.
\]
\end{assumption}
This is a reasonable assumption, since we require $j_{i}$ Gaussian
draws to produce $u^{i}|y^{i}$. We have the following result on the
estimator $Z$ defined in equation (\ref{eq:estimator}), which holds
under Assumptions \ref{ass1}, \ref{ass2}, \ref{ass:lincost}: 
\begin{thm}
\label{thm:linear} Let $f:\state\to\R$ be $s$-H\"older continuous for some $s\in(0,1]$ and assume $a>\frac{{1+s}}{2s},$ that is, that the
eigenvalues of the prior covariance decay sufficiently fast. Then,
there exist choices of $j_{i}$ and $\rp\left(N\ge i\right)$, such
that 
\[
\tilde{Z}=\sum_{i=0}^{N}\frac{\tilde{\Delta}_{i}}{\rp\left(N\ge i\right)}
\]
is an unbiased estimator of $\E_{\mu^{y}}[f]$ with finite variance
and finite expected computing time. Here, as in Proposition \ref{prop:rhee},
each $\tilde{\Delta}_{i}$ is an independent copy of $\Delta_{i}$
as defined above. In particular, two examples of such choices are: 
\begin{enumerate}
\item i) $j_{i}=2^{i}$ and $\rp\left(N\ge i\right)\propto2^{\frac{(2-\epsilon)s(1-2a)i}{2}}$,
for any $\epsilon\in(0,\frac{2+2s-4as}{s(1-2a)})$; 
\item ii) $j_{i}\lesssim i^{q}$, and $\rp\left(N\ge i\right)\propto i^{s(q-1-2aq)+2+\epsilon}$,
for $q>\frac{s-3}{1+s-2as}$ and for any $\epsilon\in(0,s-3-q(1+s-2as))$. 
\end{enumerate}
\end{thm}
The assumption on the regularity of the prior, $a>\frac{1+s}{2s}$,
in Theorem \ref{thm:linear}, is more severe than the usual $a>\frac{1}{2}$
which is required for the formulation of the Bayesian linear inverse
problem. We now show how to modify the estimator $Z$ in order to
relax this assumption, in the case that $f$ is a linear functional
hence Lipschitz continuous. In this case, Theorem \ref{thm:linear}
requires $a>1,$ while we will show that it is possible to get an
estimator which only requires $a>\frac{1}{2}$. 

We modify $u^{i}|y^{i}$ by taking the truncated Karhunen-Loeve expansion
in $\R^{j_{i}}$ as before, but now we take a draw from the prior
in $\state\setminus\R^{j_{i}}$. In other words we define the approximations
$\tilde{u}^{i}|y^{i}$ of $u|y$ at level $i\in\N$, by 
\begin{equation}
\tilde{u}^{i}|y^{i}\coloneq\sum_{\jj=1}^{j_{i}}\frac{\jj^{-2p}y_{\jj}}{\jj^{2a}+\jj^{-4p}}e_{\jj}+\sum_{\jj=1}^{j_{i}}\frac{\zeta_{\jj}}{(\jj^{2a}+\jj^{-4p})^{\frac{1}{2}}}e_{\jj}+\sum_{\jj=j_{i}+1}^{\infty}\jj^{-a}\zeta_{\jj}e_{\jj},\label{eq:hmdr}
\end{equation}
where $\{\zeta_{\jj}\}$ are independent and identically distributed
standard normal random variables, $\G(0,1)$. It is important to notice
that since the estimator $Z$ only contains the differences $\Delta_{i}\coloneq f(\tilde{u}_{i}|y_{i})-f(\tilde{u}_{i}|y_{i})=f(\tilde{u}_{i}|y_{i}-\tilde{u}_{i}|y_{i})$,
we can still construct $Z$ in finite time; here we used the fact
that $f$ is linear. The motivation for doing this is that the posterior
under our assumptions in the present linear inverse problem setting,
is dominated by the prior \cite[section 4]{ABPS14} and so we expect
the new differences $\Delta_{i}$ to have faster decay compared to
the ones defined through $u^{i}|y^{i}$. Indeed, this proves to be
true (see Lemmas \ref{lem1} and \ref{lem2} in section \ref{sec:proofs}
for estimates of the differences $\Delta_{i}$ using the two different
approximations of $u|y$) and we have the following result on the
estimator $Z$, which again holds under Assumptions \ref{ass1}, \ref{ass2},
\ref{ass:lincost}: 
\begin{thm}
\label{thm:linearhm} Let $f:\state\to\R$ be linear and assume $a>\frac{1}{2}$,
that is that the prior is supported in $\state$. Then, there exist
choices of $j_{i}$ and $\rp\left(N\ge i\right)$, such that 
\[
\tilde{Z}=\sum_{i=0}^{N}\frac{\tilde{\Delta}_{i}}{\rp\left(N\ge i\right)}
\]
is an unbiased estimator of $\E_{\mu^{y}}[f]$ with finite variance
and finite expected computing time. In particular, two examples of
such choices are: 
\begin{enumerate}
\item i) $j_{i}=2^{i}$ and $\rp\left(N\ge i\right)\propto2^{\frac{(2-\epsilon)(1-4p-4a)i}{2}}$,
for any $\epsilon\in(0,\frac{4-8p-8a}{1-4p-4a})$; 
\item ii) $j_{i}\lesssim i^{q}$, and $\rp\left(N\ge i\right)\propto i^{s(q-1-2aq)+2+\epsilon}$,
for $q>\frac{s-3}{1+s-2as}$ and for any $\epsilon\in(0,s-3-q(1+s-2as))$. 
\end{enumerate}
\end{thm}
\begin{rem}
Using Proposition \ref{prop:genrhee} which generalises Proposition
\ref{prop:rhee}, it is straightforward to check that Theorems \ref{thm:linear}
and \ref{thm:linearhm} can be extended to hold for estimating posterior
expectations of functions $f:\state\to H$ which are respectively
$s$-H\"older continuous and bounded linear, where $(H,\pr{\cdot}{\cdot}_{H},\norm{\cdot}_{H})$
is a Hilbert space. In particular, both theorems hold for unbiased
estimation of the posterior mean. We comment here that the unbiased
estimation of linear functions can be useful, despite the explicit
availability of the posterior mean $m$. To see this, note that the
cost of estimating $m$ using the unbiased estimator $Z$ and the
Monte Carlo principle at an error level $\varepsilon>0$, is always
proportional to $\varepsilon^{-2}$. On the other hand the cost of
approximating $m$ at the same error level by using a high enough
level of approximation, varies depending on the particular problem
at hand. For example, in the assumed diagonal setting, it is straightforward
to check that this cost is $\varepsilon^{\frac{1}{1-4p-4a}}$, which,
since $a>\frac{1}{2}$, is cheaper than the Monte Carlo method. Nevertheless,
the situation is different in the case of more difficult setups in
which the cost of approximating at level $i$ grows superlinearly
with $i$. 
\end{rem}
%%%%%%%%%%%%%% WASSERSTEIN MARKOV CHAINS %%%%%%%%%%

\section{\label{sec:wasserstein}Wasserstein convergence of Markov chains
and unbiased estimators of equilibrium expectations}
In this section we consider the problem of constructing unbiased estimators
for expectations with respect to limiting distributions of Markov
chains. As discussed in subsection \ref{sec:OverviewOfRhees}, the
techniques developed in \cite{rhee2012new}, have been applied in
\cite{RheePHD} in this setting and in particular for uniformly recurrent,
contracting and Harris chains. The approximation is achieved by considering
the finite-time distributions, and then the challenge is to construct
a coupling which guarantees that the chain comes close in $L^{2}$
as time increases. In general, the approach taken in \cite{RheePHD},
is to use intelligent techniques that turn convergence in distribution
to almost sure convergence (for example the backwards process technique
described in section \ref{sec:OverviewOfRhees}). We now show that
this is not necessary, but instead a simulatable coupling between
chains started at different positions is sufficient, provided this
coupling drives the two chains towards each other quickly enough in
expectations.

Let $\state$ be a general state space. Throughout this section $d$
denotes a distance-like function, that is a function $d:\state\times\state\to\R^{+}$
which is symmetric, lower semi-continuous and which vanishes when
the two arguments are equal. Let $X=\left\{ X_{n}\right\} _{n\in\N}$
be a Markov chain with transition probabilities $P(x,.)$ and invariant
distribution $\pi$. Our aim is to find an unbiased estimate for the
expectation $\E Y\coloneq\E_{\pi}[f]$, where $f:\state\to\R$ is
an $s$-H\"older continuous function with respect to $d$ for some $s\in(0,1]$,
that is 
\[
\norm f_{s}\coloneq\sup_{x\neq y}\frac{|f(x)-f(x)|}{d^{s}(x,y)}<\infty.
\]

\begin{assumption}
\label{assu:wass}We work under the following assumptions on the chain
$X$ in terms of the distance-like function $d$: 
\begin{enumerate}
\item[i.] there exists a simulatable coupling $K\left((x,y),(dx^{\prime},dy^{\prime})\right)$
between the transition probabilities $P\left(x,dx'\right)$ and $P\left(y,dy'\right)$,
which satisfies 
\begin{equation}
K^{n}d^{2s}\leq cr^{n}d^{2s}\text{ for some }r<1;\label{eq:contrKernel}
\end{equation}

\item[ii.] there exists a point $x_{0}\in\state$ such that 
\begin{equation}
\sup_{n}P^{n}d(x_{0},\cdot)<\infty.\label{eq:contractionBoundedness}
\end{equation}

\end{enumerate}
\end{assumption}
\begin{rem}\label{rem:wassersteincontr}
We comment the following about Assumptions \ref{assu:wass}. \end{rem}
\begin{enumerate}
\item Assumption \ref{assu:wass}.i. is more general than the contracting
chains case considered in Chapter 3.4 in \cite{RheePHD}, since the
coupling is allowed to depend on both $x$ and $y$; for more details
see Remark \ref{rem:pairwisecoupling} below. 
\item Assumption \ref{assu:wass}.i. is related to the $s$-Wasserstein
distance-like function $d_{s}$ associated with $d$, which is given
by 
\begin{eqnarray*}
d_{s}(\nu_{1},\nu_{2}) & = & \left(\inf_{\pi\in\Gamma(\nu_{1},\nu_{2})}\int_{\state\times\state}d^{s}(x,y)\pi(dx,dy)\right)^{\frac{1}{s}},
\end{eqnarray*}
with $\Gamma(\nu_{1},\nu_{2})$ being the set of couplings of $\nu_{1}$
and $\nu_{2}$ (all measures on $\state\times\state$ with marginals
$\nu_{1}$ and $\nu_{2}$). Since $K$ constitutes a particular coupling,
it follows that 
\[
d_{2s}^{2s}\left(P^{n}(x,\cdot),P^{n}(y,\cdot)\right)\leq K^{n}d^{2s}.
\]
That is, our assumption is stronger than the corresponding assumption
on the transition probabilities in terms of $d_{s}$ because we need
$K$ to be simulatable. 
\item Finally, observe that Assumption \ref{assu:wass}.ii. can be established
by picking a distance $d$ that is bounded or compatible with a Lyapunov
function of the underlying Markov chain. %\item Under these conditions it is well known that $X$ converges to its%invariant distribution $\pi$. 
\end{enumerate}
As discussed in subsection \ref{ssec:implback}, we generate the differences
$\Delta_{i}$ directly and through them define the approximations
$Y_{i}$. Let $\{a_{i}\}$ be an increasing sequence of integers.
We generate $\Delta_{i}$ as specified in Algorithm \ref{alg:CouplingWasserstein}
and where $x_{0}$ is defined in Assumption \ref{assu:wass}.ii..
We denote by $\ux{k}{i}$ and $\lx{k}{i}$ the chains coupled through
the kernel $K$ for $k=-a_{i-1}+1,\dots0.$

\begin{algorithm}
For $i=0$ run the Markov chain $\ux{-a_{0}}{0}$ up to time $0$
and set $\Delta_{0}=f(\ux{0}{0}).$ For $i\ge1$ do 
\begin{itemize}
\item set $\ux{-a_{i}}{i}=x_{0}$ and run the chain until $\ux{-a_{i-1}}{i}$
; 
\item set $\lx{-a_{i-1}}{i}=x_{0}$; 
\item evolve $\lx{k}{i}$ and $\ux{k}{i}$ jointly according to $K$ up
to time $0$; 
\item set $\Delta_{i}=f(\ux{0}{i})-f(\lx{0}{i})$. 
\end{itemize}
\protect\protect\protect\caption{\label{alg:CouplingWasserstein}Coupled contraction for unbiased estimation}
\end{algorithm}

In order to follow the general idea of the unbiasing technique as
outlined in section \ref{sec:OverviewOfRhees}, we now make an assumption
about the computing time of generating $\Delta_{i}.$ 
\begin{assumption}
\label{assu:wasserCost}The expected computing time $t_{i}$ of generating
$\Delta_{i}$ satisfies 
\[
t_{i}\lesssim a_{i}.
\]

\end{assumption}
This seems a reasonable assumption as $\Delta_{i}$ can be produced
using $a_{i}$ steps following $K$. We have the following result
on the estimator $Z$ defined in equation (\ref{eq:estimator}): 
\begin{thm}
\label{thm:wasserstein}Suppose Assumption \ref{assu:wass} (existence
of contracting coupling) and Assumption \ref{assu:wasserCost} are
satisfied, and that $f:\state\to\R$ is $s$-H\"older continuous with respect to
$d$, for some $s\in(0,1]$. Then there exist choices for $a_{i}$ and $\rp\left(N\ge i\right)$,
such that 
\[
Z=\sum_{i=0}^{N}\frac{\Delta_{i}}{\rp\left(N\ge i\right)}
\]
is an unbiased estimator of $\E_{\pi}[f]$ with finite variance and
finite expected computing time. In particular, an example of such
choices is $a_{i}\lesssim r^{(2\epsilon-1)i}$ and $\rp(N\geq i)\propto r^{(1-\epsilon)i}$
for any $0<\epsilon<\frac{1}{2}$. 
\end{thm}
Note that the exponential convergence in Assumption \ref{assu:wass}.i.
makes the calculations easier, however the same argument works for
sufficiently fast sub-exponential convergence. 
\begin{assumption}
\label{ass:weakcass} There exists a simulatable coupling $K\left((x,y),(dx^{\prime},dy^{\prime})\right)$
between the transition probabilities $P\left(x,dx'\right)$ and $P\left(y,dy'\right)$,
which satisfies
\begin{equation}
K^{n}d^{2s}\leq Cn^{-2r}d^{2s},\label{eq:ContrSubGeom}
\end{equation}
where $r>\frac{1}{2}$. 
\end{assumption}

\begin{thm}
\label{thm:wassersteinsubg}Suppose Assumption \ref{ass:weakcass},
Assumption \ref{assu:wass}.ii. and Assumption \ref{assu:wasserCost}
are satisfied, and that $f:\state\to\R$ is $s$-H\"older continuous with respect
to $d$, for some $s\in(0,1]$. Then there exist choices for $a_{i}$ and $\rp\left(N\ge i\right)$,
such that 
\[
Z=\sum_{i=0}^{N}\frac{\Delta_{i}}{\rp\left(N\ge i\right)}
\]
is an unbiased estimator of $\E_{\pi}[f]$ with finite variance and
finite expected computing time. In particular, an example of such
choices is $a_{i}\asymp i^{k}$ and $\rp(N\geq i)\propto i^{-2rk+2+\epsilon}$
for $k>\frac{3}{2s-1}$ and any $0<\epsilon<-3-(1-2s)k$. \end{thm}

\begin{rem}
The Assumption \ref{ass:weakcass} can be verified using drift conditions
and coupling sets which are provided in the article \cite{durmus2014new}.
Note that in this case it is not even clear that the ergodic average
of the underlying Markov chain satisfies a Central Limit Theorem,
while the construction above remains valid. For $r\leq\frac{1}{2}$,
the decay of $\norm{\Delta_{i}}_{2}$ is not fast enough to allow
for $Z$ to have both finite variance and finite expected computing
time.% this is true even if we choose $a_i=2^i$. 
\end{rem}

\begin{rem}
Using Proposition \ref{prop:genrhee} which generalises Proposition
\ref{prop:rhee}, it is straightforward to check that Theorems \ref{thm:wasserstein} and \ref{thm:wassersteinsubg}
 can be extended to hold for estimating
expectations with respect to $\pi$ of functions $f:\state\to H$ which are
$s$-H\"older continuous with respect to $d$, where $(H,\pr{\cdot}{\cdot}_{H},\norm{\cdot}_{H})$
is a Hilbert space.
\end{rem}

\begin{rem}
\label{rem:pairwisecoupling}This section is a genuine generalisation
of section 3.4 of \cite{RheePHD}. In this reference, the authors
consider Markov chains that can be represented through iterated random
functions which satisfy 
\[
X_{n+1}=\varphi_{n}\left(X_{n}\right)=\varphi\left(X_{n},\xi_{n}\right),
\]
with $\xi_{n}$ independent and identically distributed, without loss of generality, ${\rm U}[0,1]$ random variables.
Under the assumption that 
\begin{equation}
\sup_{x\ne y}\E\left(\frac{d\left(\varphi\left(x\right),\varphi(y)\right)}{d\left(x,y\right)}\right)^{2\gamma}=r<1,\label{eq:contractionRhee}
\end{equation}
for some $r<1$ the general procedure can be applied to $\Delta_{i}=f(\tilde{X}_{i})-f(\tilde{X}_{i-1})$
where $\tilde{X_{i}}$ is the backwards chain discussed in subsection
\ref{sec:rheemc}. In the language of the present section, the coupling in \cite[Section 3.4]{RheePHD} is specified
through the random function, that is we can use 
\[
K\left((x,y),\left(d\tilde{x},d\tilde{y}\right)\right)=\mathcal{L}\left(\varphi(x,\xi),\varphi(y,\xi)\right)
\]
which turns (\ref{eq:contractionRhee}) into (\ref{eq:contrKernel}).
We now show an example of a coupling of a Markov chain that leads to a faster contraction in \eqref{eq:contrKernel} than any representation of the Markov chain through a random function. In particular, the coupling cannot be represented by a random function. Consequently, this section indeed genuinely generalises the results of \cite{RheePHD}. \end{rem}
\begin{example}
Consider the Markov chain given by 
\begin{equation}
X_{n+1}\sim(X_{n}+U)\mod2\pi,
\end{equation}
where $U\sim{\rm U}[-2,2]$. We denote the corresponding transition
kernel by $P(x,\cdot)$ and note that it is of the form $P(x,\cdot)=U(A_{x})$
where $A_{x}\subset[0,2\pi]$ . It is easy to check that $|A_{x}\cap A_{\tilde{x}}|\ge8-2\pi$
for any $x,\tilde{x}\in[0,2\pi]$, so that we have coupling probability
for the 1-step maximal coupling at least $\frac{8-2\pi}{4}=2-\frac{\pi}{2}>\frac{1}{3}$.
More precisely, this maximal coupling can be written as 
\[
Q(x_{1},x_{2})=\mathcal{L}\left((Y,Y)\1_{[0,\frac{\left|A_{x_{1}}\cap A_{x_{2}}\right|}{4}]}(U)+(Y_{1},Y_{2})\left(\1_{(\frac{\left|A_{x_{1}}\cap A_{x_{2}}\right|}{4},1]}(U)\right)\right)
\]
where $Y_{i}\sim U(A_{x_{i}}\setminus A_{x_{1}}\cap A_{x_{2}}),Y\sim U(A_{x_{1}}\cap A_{x_{2}})$
and $U\sim{\rm U}[0,1]$ are independent random variables.  Note that this coupling 
clearly satisfies Assumption \ref{assu:wass} with 
\[
	Qd \leq (1-\frac{8-2\pi}{4})d,
\]
for $d$ the discrete
metric.

In contrast, suppose there is a random function $\varphi\left(\cdot,\xi\right)$
such that $P(x,\cdot )=\mathcal{L}\left(\varphi(x,\xi)\right)$ 
and
\begin{equation}\label{eq:exampRandomFunc}
	\mathbb{E}d\left(\varphi(x,\xi) , \varphi(y,\xi) \right)<\frac{2}{3}d(x,y),
\end{equation}
{for every $x,y\in[0,2\pi]$.}
Then consider the three points:
$x_{1}=0$, $x_{2}=\frac{2\pi}{3}$, $x_{3}=\frac{4\pi}{3}$. It is
easy to check that the minorisation measures between $P(x_{1},\cdot)$
and $P(x_{2},\cdot),$ $P(x_{2},\cdot)$ and $P(x_{3},\cdot),$ and
$P(x_{1},\cdot)$ and $P(x_{3},\cdot),$ necessarily lie in the intervals
$[0,\frac{2\pi}{3}]$, $[\frac{2\pi}{3},\frac{4\pi}{3}]$ and $[\frac{4\pi}{3},2\pi]$,
respectively (that is, $A_{x_{1}}\cap A_{x_{2}}\subset[0,\frac{2\pi}{3}]$,
$A_{x_{2}}\cap A_{x_{3}}\subset[\frac{2\pi}{3},\frac{4\pi}{3}]$ and
$A_{x_{1}}\cap A_{x_{3}}\subset[\frac{4\pi}{3},2\pi]$). This observation implies that 
the sets
\begin{align*}
 & \left\{ \xi\mid\varphi(x_{1},\xi)=\varphi(x_{2},\xi)\right\}, \\
 & \left\{ \xi\mid\varphi(x_{1},\xi)=\varphi(x_{3},\xi)\right\} \text{ and }\\
 & \left\{ \xi\mid\varphi(x_{2},\xi)=\varphi(x_{3},\xi)\right\} 
\end{align*}
are pairwise disjoint. {Since $d$ is the discrete metric}, for \eqref{eq:exampRandomFunc} to hold each of the above sets needs to have probability exceeding $\frac{1}{3}$. This is a contradiction. 
\end{example}

%%%%%%%%%%% INDEPENDENCE SAMPLER %%%%%%%%%%%

\section{\label{sec:is}Unbiased estimation for Bayesian inverse problems with uniform priors, using the independence sampler}

\global\long\def\iid{\overset{\text{i.i.d.}}{\sim}}

In this section we consider infinite dimensional state spaces and
extend the considerations of Glynn and Rhee on unbiased estimation
of expectations with respect to the limiting distribution of a Markov
chain, to remove not only the bias introduced due to the use of the
finite-time distributions as approximations of the target distribution,
but also the bias introduced due to the necessity to discretise. For
expository reasons, we do this in an idealised nonlinear Bayesian
inverse problem setting, and present an unbiased version of the independence
sampler to approximate expectations with respect to the posterior.
Later on in section \ref{sec:trans}, we extend our results to more
elaborate settings and present an unbiased version of the preconditioned
Crank-Nicholson algorithm.

\subsection{Setup}

We consider the inverse problem of finding an unknown function $u$
from noisy indirect observations $y\in\R^{d}$. We assume the data
model
\[
y=G(u)+\eta,
\]
where $G:\state\to\R^{d}$ is the observation operator and $\eta\sim N(0,I)$
is the observational noise. A typical example in the inverse problems
literature, is the situation that $G$ maps the diffusion coefficient
$u$ of an elliptic partial differential equation, to the solution
evaluated at a set of finite points \cite[section 3.4]{DS13}. Henceforward,
we identify the function $u$ with a sequence $u=\{u_{k}\}_{k\in\N}\in\R^{\infty},$
which represents the coefficients of the unknown function in some
series expansion.

Let $u_{k}^{\star}\downarrow0$ and consider the sequence of $j$-dimensional
state spaces 
\[
\state_{j}=\prod_{k=1}^{j}\left[-u_{k}^{\star},u_{k}^{\star}\right],
\]	
assumed to be embedded in the infinite dimensional state space $\state\coloneq\prod_{k=1}^{\infty}\left[-u_{k}^{\star},u_{k}^{\star}\right].$
We denote by $\Pi_{j}$ the projection onto $\state_{j}\subset\state$,
$\Pi_{j}:\state\rightarrow\state_{j}$, $\Pi_{j}u=\left(u_{1},\dots,u_{j},0,\dots\right)$.
The reader should think of an element $u\in\state$ as the collection
of coefficients (for example Fourier) of a function which decay at
a prescribed rate. Depending on the particular expansion used, the
decay of the coefficients translates to smoothness of the corresponding
function. We put a uniform prior on $u\in\state$, 
\begin{equation}
\mu_{0}=\bigotimes_{k=1}^{\infty}(\lambda\mid_{[-u_{k}^{\ast},u_{k}^{\ast}]}),\label{eq:is_prior}
\end{equation}
where $\lambda$ denotes the Lebesque measure, treating all components
as uniformly distributed over the range and independent of all the
other components. Such priors have been used in the inverse problem
setting in \cite{SS12}; see again \cite[section 3.4]{DS13} for a
less technical version. In particular, in these references it is shown
that under certain conditions on the basis used in the series expansion
and on the continuity and boundedness of the forward operator $G$,
the posterior distribution $\mu^{y}$ of $u\vert y$ is well defined
and given by 
\[
\frac{d\mu^{y}}{d\mu_{0}}(u)\propto\exp\left(-\norm{y-G\left(u\right)}_{\R^{d}}^{2}\right).
\]
However, in general $\mu^{y}$ is not available in closed form and
on the contrary it can be a very complicated infinite dimensional
probability measure. In order to probe the posterior, one needs to
discretise and sample. We discuss how to do this naively using an
independence sampler in section \ref{sec:nis}, while in section \ref{sec:uis}
we modify the independence sampler to achieve unbiased estimation
of expectations with respect to $\mu^{y}$.

\subsection{Approximations to the forward problem and a naive independence sampler}
\label{sec:nis} In the assumed inverse problem setting, it is natural
to discretise $u$ in $\state_{j}$ and to approximate the observation
operator $G$ by $G_{j}:\state\rightarrow\mathbb{R}^{\ddata}$ which
depends on $u$ only through the projection $\Pi_{j}u$, that is 
\begin{equation}
G_{j}(u)=G_{j}(\Pi_{j}u).\label{eq:obsop}
\end{equation}
We use the notation $G_{\infty}=G$ and work under the following assumption. 
\begin{assumption}
\label{ass:indep}\textup{ There exists some $\beta>1$, such that
the observation operator and its approximations satisfy} 
\begin{eqnarray*}
\sup_{u\in\state}\norm{G_{j}(u)-G\left(u\right)}_{\R^{d}} & \lesssim j^{-\beta},\\
\sup_{u\in\state}\norm{G\left(u\right)}_{\R^{d}} & <\infty.
\end{eqnarray*}
\end{assumption}
Notice that Assumption \ref{ass:indep} implies 
\begin{equation}
\sup_{u\in\state}\norm{G_{j}(u)-G_{\tilde{j}}\left(u\right)}_{\R^{d}}\lesssim\left(j\wedge\tilde{j}\right)^{-\beta}.\label{eq:approx}
\end{equation}
For
a concrete example of $G$ and the relevant discretisations, which
satisfies Assumption \ref{ass:indep} see subsection \ref{ssec:elipdisc}.

We define the projected priors $\mu_{0,j}=\Pi_{j\star}\mu_{0}$
on $\state_{j}$, which combined with the approximation of the observation
operator give rise to the approximate posteriors
\[
\frac{d\mu_{j}^{y}}{d\mu_{0,j}}(u)\propto\exp\left(-\frac{1}{2}\left\Vert y-G_{j}\left(u\right)\right\Vert _{\R^{d}}^{2}\right).
\]

Approximating an expectation with respect to $\mu^{y}$ by an expectation
with respect to $\mu_{j}^{y},$ results in a discretisation error
which is quantified in \cite{cotter2010approximation}, \cite{2013HoangComplexityMCMC}
and \cite{DS13}. 
Moreover,
the expectations with respect to $\mu_{j}^{y}$ are not available
analytically but they are amenable to approximation using Markov chain
Monte Carlo algorithms. Again for illustration, we consider the (regular)
independence sampler, the Metropolis-Hastings algorithm arising from
the state-independent proposal $\mu_{0,j}$, see Algorithm \ref{alg:indep}.
We denote the resulting Markov chain by $X_{\cdot}^{j}$ and its transition
kernel by $P_{j}$. It is shown in \cite{vollmer2013dimension}, that
the boundedness of $G_{j},$ implies a deterministic lower bound on
the acceptance probability 
\begin{equation}
\alpha_{j}\ge\alpha_{\star}>0.\label{eq:BIPlowACC}
\end{equation}
In this case, the Monte Carlo error can be controlled explicitly because
the Markov chain $X_{\cdot}^{j}$ is uniformly ergodic due to (\ref{eq:BIPlowACC}),
\cite{2013HoangComplexityMCMC}. The overall error in the approximation
\[
\E_{\mu^{y}}[f]\approx\E_{\mu_{j}^{y}}[f]\approx\frac{1}{K_{max}}\sum_{k=1}^{K_{max}}f\left(X_{k}^{j}\right)
\]
has two contributions, the Monte Carlo error and the discretisation error. In particular, the discretisation error is chosen at the beginning
of the MCMC computation and can only be reduced by restarting the
computations from scratch. In the next subsection we formulate the modified
independence sampler which leads to unbiased estimation of posterior
expectations.
\begin{algorithm}
Generate $X_{0}^{j}$. Iterate the following steps for $k=1,...,K_{max}$: 
\begin{enumerate}
\item $\xi^{j}\sim\mu_{0,j}$ 
\item set $X_{k+1}^{j}=\xi^{j}$ with probability 
\begin{equation}
\alpha_{j}(X_{k}^{j},\xi^{j})=1\wedge{\exp\left(\frac{1}{2}\left\Vert y-G_{j}\left(X_{k}^{j}\right)\right\Vert _{\R^{d}}^{2}-\frac{1}{2}\left\Vert y-G_{j}\left(\xi^{j}\right)\right\Vert _{\R^{d}}^{2}\right)}\label{eq:acceptance}
\end{equation}
and $X_{k+1}^{j}=X_{k}^{j}$ otherwise. 
\end{enumerate}
\protect\protect\caption{\label{alg:indep}Independence sampler}
\end{algorithm}

\subsection{Unbiased estimation using the independence sampler}
\label{sec:uis}
We now present a version of the independence sampler which leads to
the removal of both the bias due to the use of the finite-time distributions
and the bias due to the discretisation of the posterior.

We use the unbiasing programme of Glynn and Rhee as introduced in
subsections \ref{sec:OverviewOfRhees} and \ref{ssec:implback}, in
order to construct an unbiased estimator $Z$ of the posterior expectation
$\E_{\mu^{y}}[f]$, for some function $f:\state\to\R$. For two increasing
sequences of integers $a_{i}$ and $j_{i}$, representing the time-step
and the discretisation level respectively, we would like to set $\Delta_{i}=f(X_{a_{i}}^{j_{i}})-f(X_{a_{i-1}}^{j_{i-1}})$
in the definition of $Z$ in Proposition \ref{prop:rhee}, where the
chains $X_{\cdot}^{j_{i}}$ and $X_{\cdot}^{j_{i-1}}$ are the (regular)
independence sampler chains introduced in the previous subsection
following the transition kernels $P_{j_{i}}$ and $P_{j_{i-1}}$,
respectively. For the unbiasing technique to work, we need to construct
an appropriate coupling between the two chains, so that $\norm{\Delta_{i}}_{2}$
decays sufficiently quickly for Proposition \ref{prop:rhee} to apply,
and the expected computing time is finite. In order to achieve this, we generate $\Delta_{i}$
using a "top" level chain in $\state_{j_{i}}$
and a "bottom" level chain in $\state_{j_{i-1}}$,
which we denote by $\ux{\cdot}{i}$ and $\lx{\cdot}{i},$  and which perform $a_i$ and $a_{i-1}$ steps,  respectively.
According to Proposition \ref{prop:rhee}, we need $\Delta_{i}$ to
be independent for different $i$, hence the two chains $\ux{\cdot}{i}$
and $\lx{\cdot}{i+1}$ both following the transition kernel $P_{j_{i}}$
in $\state_{j_{i}}$, are constructed independently. Nevertheless,
the chains at different levels are coupled as follows: 
\begin{enumerate}
\item $\ux{\cdot}{i}$ is coupled to $\lx{\cdot}{i}$ which follows the
transition kernel $P_{j_{i-1}}$ on $\state_{j_{i-1}}$; 
\item $\lx{\cdot}{i+1}$ is coupled to $\ux{\cdot}{i+1}$ which follows
the transition kernel $P_{j_{i+1}}$ on $\state_{j_{i+1}}.$ 
\end{enumerate}
The following diagram illustrates the construction of the $\Delta_{i}$:

\[
\begin{array}[t]{cccccccccl}
 &  &  &  & x_{0}= & \ux{-a_{0}}{0} & \dots & \ux{0}{0} & \}\Delta_{0} & =f(\ux{0}{0})\\
 &  &  &  & x_{0}= & \lx{-a_{0}}{1} & \dots & \lx{0}{1}\\
 &  &  &  &  & \mid & \mid & \mid & \}\Delta_{1} & =f(\ux{0}{1})-f(\lx{0}{1})\\
 &  & x_{0}= & \ux{-a_{1}}{1} & \dots & \ux{-a_{0}}{1} & \dots & \ux{0}{1}\\
 &  & x_{0}= & \lx{-a_{1}}{2} & \dots & \lx{-a_{0}}{2} & \dots & \lx{0}{2}\\
 &  &  & \mid & \mid & \mid & \mid & \mid & \}\Delta_{2} & =f(\ux{0}{2})-f(\lx{0}{2})\\
x_{0}= & \ux{-a_{2}}{2} & \dots & \ux{-a_{1}}{2} & \dots &  & \dots & \ux{0}{2}\\
x_{0}= & \lx{-a_{2}}{3} & \dots & \lx{-a_{1}}{3} & \dots &  & \dots & \lx{0}{3}
\end{array}
\]
Here $|$ indicates coupling between two chains. We would like to point out a connection to Multilevel Markov Chain Monte Carlo  (MLMCMC) \cite{2013HoangComplexityMCMC,TeckentrupMultilevel2013, DG14}. Both the present method and MLMCMC couple Markov chains in different dimensions. However, the method presented in this section can be seen as taking a diagonal approach between the unbiasing approach of \cite{RheePHD} and the MLMCMC idea; this also applies for our method of coupling pCN algorithms presented in the next section. More precisely, we couple Markov chains in different dimensions performing a different number of steps. In this way, we remove the bias due to both discretisation and the finite number of iterations. In contrast in MLMCMC both contributions to the bias remain, however it achieves an efficient distribution of computations between discretisation levels, which  reduces the cost of producing estimators with a certain error level compared to standard MCMC.

The couplings above arise
form the minorisation due to the lower bound on the acceptance probability.
They can be represented using the random functions $\varphi_{\ux{}{}}^{i}$
and $\varphi_{\lx{}{}}^{i}$, defined as: 
\begin{equation}
\begin{aligned}\varphi_{\ux{}{}}^{i}(x,W^{i})= & \1_{[0,\alpha_{\star}]}(U_{1}^{i})\xi_{1}^{i}+\1_{(\alpha_{\star},1]}(U_{1}^{i})\Big(\1_{[0,\frac{\alpha_{j_{i}}(x,\xi_{2}^{i})-\alpha_{\star}}{1-\alpha_{\star}}]}(U_{2}^{i})\xi_{2}^{i}\\
 & +\1_{(\frac{\alpha_{j_{i}}(x,\xi_{2}^{i})-\alpha_{\star}}{1-\alpha_{\star}},1]}(U_{2}^{i})x\Big),\\
\varphi_{\lx{}{}}^{i}(x,W^{i})= & \1_{[0,\alpha_{\star}]}(U_{1}^{i})\Pi_{j_{i-1}}\xi_{1}^{i}+\1_{(\alpha_{\star},1]}(U_{1}^{i})\Big(\1_{[0,\frac{\alpha_{j_{i-1}}(x,\Pi_{j_{i-1}}\xi_{2}^{i})-\alpha_{\star}}{1-\alpha_{\star}}]}(U_{2}^{i})\Pi_{j_{i-1}}\xi_{2}^{i}\\
 & +\1_{(\frac{\alpha_{j_{i-1}}(x,\Pi_{j_{i-1}}\xi_{2}^{i})-\alpha_{\star}}{1-\alpha_{\star}},1]}(U_{2}^{i})x\Big),
\end{aligned}
\label{eq:independenceCoupling}
\end{equation}
where $W^{i}=(U_{1}^{i},U_{2}^{i},\xi_{1}^{i},\xi_{2}^{i})$, for
$U_{l}^{i}\sim{\rm U}[0,1]$ and $\xi_{l}^{i}\sim\mu_{0}^{j_{i}}$,
$l=1,2$, which are all independent of each other.

The functions $\varphi_{\ux{}{}}^{i}$ and $\varphi_{\lx{}{}}^{i}$
are constructed by minorising the transition kernels $P_{j_{i}}$
and $P_{j_{i-1}},$ using the proposal distributions $\mu_{0}^{j_{i}}$
and $\mu_{0}^{j_{i-1}}$, respectively. The uniform random variable
$U_{1}^{i}$ is used to construct the "coin'' for switching
between the minorising measure and the residual kernel. The residual
kernel is still a Metropolis-Hastings kernel with a corrected acceptance
probability and $U_{2}^{i}$ is used for acceptance and rejection.
The coupling between the "top" and "bottom"
chains used to construct $\Delta_{i}$, will be achieved through the
use of the same random seeds in the random functions $\varphi_{\ux{}{}}^{i}$
and $\varphi_{\lx{}{}}^{i}$.

\begin{algorithm}
\protect\protect\caption{\label{alg:unbis}Coupled independence samplers for unbiased estimation}

Fix a starting point $x_{0}\in\state_{j_{0}}$ once and for all. For
$i=0$, generate $\Delta_{0}$ as follows: 
\begin{enumerate}
\item set $\ux{-a_{0}}{0}=x_{0}$ on $\state_{j_{0}}$ and simulate according
to Algorithm \ref{alg:indep} up to $\ux{0}{0}$; 
\item set $\Delta_{0}=f\left(\ux{0}{0}\right)$. 
\end{enumerate}
For $i\ge1$, generate $\Delta_{i}$ as follows: 
\begin{enumerate}
\item set $\ux{-a_{i}}{i}=x_{0}$ and simulate according to Algorithm \ref{alg:indep}
upto $\ux{-a_{i-1}}{i}$ in dimension $j_{i}$; 
\item set $\lx{-a_{i-1}}{i}=x_{0}$; 
\item for $k=-a_{i-1}+1,\dots,0$ simulate $\ux{k}{i}$ and $\lx{k}{i}$
as coupled independence samplers as described below:

\begin{enumerate}
\item draw $U_{l}^{i}\iid{\rm U}[0,1]$ and $\xi_{l}^{i}\iid\mu_{0}^{j_{i}}$
for $l=1,2$ independently from everything else and set $W^{i}=(U_{1}^{i},U_{2}^{i},\xi_{1}^{i},\xi_{2}^{i})$
as the collection of all random input to do the $k$-th step; 
\item set 
\begin{equation}
\begin{aligned}\ux{k}{i} & =\varphi_{\ux{}{}}^{i}\left(\ux{k-1}{i},W^{i}\right),\\
\lx{k}{i} & =\varphi_{\lx{}{}}^{i}\left(\lx{k-1}{i},W^{i}\right);
\end{aligned}
\label{eq:indepenceSamplerCouplingHierarchy}
\end{equation}

\end{enumerate}
\item set $\Delta_{i}=f(\ux{0}{i})-f(\lx{0}{i})$. \end{enumerate}
\end{algorithm}

The construction of $\Delta_{i}$ is given in detail in Algorithm
\ref{alg:unbis}. In Lemma \ref{lem:ch5lem}, we derive bounds on the decay of  $\norm{\Delta_{i}}_{2}$ which are sufficient
for the unbiasing programme to work.  In order to achieve this, we use the decomposition\begin{align}
\norm{\Delta_{i}}_{2}^{2} & \leq2\norm{f(\lx{0}{i})-f(\Pi_{j_{i-1}}\ux{0}{i})}_{2}^{2}+2\norm{f(\ux{0}{i})-f(\Pi_{j_{i-1}}{}\ux{0}{i})}_{2}^{2}\nonumber \\
 & \coloneq2(E_{1}+E_{2}).
\end{align}
The first term $E_{1}$ measures the difference in the lower level
of the two coupled chains $\ux{\cdot}{i}$ and $\lx{\cdot}{i}$ used
to generate $\Delta_{i}$, while $E_{2}$ has to do with the dependence
of the function $f$ on higher modes. By the definition of the couplings,
see (\ref{eq:independenceCoupling}), it is clear that in order to
control $E_{1}$ it suffices to make sure that the two chains have
the same acceptance behaviour with high probability; we use Assumption
\ref{ass:indep} and the implied uniform ergodicity to show this. On the other hand, to control $E_{2}$
we make the following assumption on $f$:
\begin{assumption}
\label{ass:indepObservable}We assume that $f:\state\rightarrow\mathbb{R}$
satisfies 
\[
\sup_{x\in\state}\left|f(\Pi_{j}x)-f(\Pi_{\tilde{j}}x)\right|\lesssim(j\wedge\tilde{j})^{-\frac{\kappa}{2}},
\]
for some $\kappa>1$. 
\end{assumption}
Note that for specific examples of the function $f$, the last assumption
is essentially an assumption on the decay of the sequence defining
the space $\state$, $u_{k}^{\star}$. Under our assumptions, in Lemma \ref{lem:ch5lem} we derive bounds on the decay of $\norm{\Delta_{i}}_{2}$ which are sufficient for the unbiasing procedure to work. 

In order to control the expected computing time of the estimator $Z$,
we make the following assumption on the cost of generating $\Delta_{i}$. 
\begin{assumption}
\label{ass:cost} Let $r\coloneq\beta\wedge\kappa$, where $\beta,\kappa>1$
are defined in Assumptions \ref{ass:indep} and \ref{ass:indepObservable},
respectively. We assume that the computational cost of one step of
the chain at level $j_{i}$ is 
\[
s_{i}\lesssim j_{i}^{\theta},
\]
with $\theta<r$. {Therefore, since we need $a_{i}$ steps of the
chain to generate $\Delta_{i}$, the expected computing time $t_{i}$
of $\Delta_{i}$ satisfies} 
\[
t_{i}\lesssim a_{i}j_{i}^{\theta}.
\]
\end{assumption}
\begin{rem}
The simultaneous validity of Assumptions \ref{ass:indep}, \ref{ass:indepObservable}
and \ref{ass:cost} depends on a relationship between the properties
of $G$, the regularity of $f$ and, most importantly, the smoothness
of the space $\state$ as expressed by the decay of the sequence $u_{k}^{\star}$.
Making this explicit in full generality is beyond the scope of this
paper, however we do provide an example in subsection \ref{ssec:elipdisc}. 
\end{rem}
We have the following result on the estimator $Z$ defined in equation
(\ref{eq:estimator}): 
\begin{thm}
\label{thm:ExistenceIndep}Suppose that the forward model satisfies
Assumption \ref{ass:indep} with $\beta>1$ and the observable $f:\state\to\R$
satisfies Assumption \ref{ass:indepObservable} with $\kappa>1$ and
let $r\coloneq\beta\wedge\kappa>1$. Furthermore, assume that the computational
cost of one step of the chain satisfies Assumption \ref{ass:cost}
with $\theta<r$. Then there is a choice of $a_{i}$, $j_{i}$ and
$\rp\left(N\ge i\right)$, such that 
\[
Z=\sum_{i=0}^{N}\Delta_{i}
\]
is an unbiased estimator of $\E_{\mu^y}[f]$ with finite variance and finite expected
computing time. For example this works for the choice $j_{i}=i^{q}$,
$a_{i}\sim\frac{q\beta}{c_{\star}}\log(i)$, for $c_{\star}=-\log(1-\alpha_{\star})$
and $\rp(N\geq i)\propto i^{-t}$ where $q>\frac{3}{r-\theta}$ and
$t\in(1+\theta q,rq-2)$. Note that under our assumptions the choices
of $q$ and $t$ are simultaneously admissible. 
\end{thm}
\begin{rem}
Using Proposition \ref{prop:genrhee} which generalises Proposition
\ref{prop:rhee}, it is straightforward to check that Theorem \ref{thm:ExistenceIndep}
 can be extended to hold for estimating
expectations with respect to $\mu^y$ of functions $f:\state\to H$ which 
satisfy an assumption of the type of Assumption \ref{ass:indepObservable}.
\end{rem}

%%%%%%%%%%%%%%%% pCN %%%%%%%%%%%%%%%%%%%%%

\section{Unbiased estimation for Gaussian-based target measures, using coupled pCN algorithms\label{sec:trans}}
In section \ref{sec:is} we showed that it is possible to couple the independence sampler in order to achieve unbiased estimation for an idealised Bayesian inverse problem setting in function space. Our coupling construction relied on assumptions on the inverse problem, which secured that the independence sampler is uniformly ergodic. However, for many measures of interest the independence sampler is not uniformly ergodic; in fact, if there exist areas of positive target measure, in which the density of the proposal with respect to the target vanishes, the independence sampler is not even geometrically ergodic  \cite{MT96}.  

In this section we extend the methodology of the last section, and couple the preconditioned Crank-Nicholson (pCN) algorithm in order to perform unbiased estimation in more difficult situations. The pCN algorithm first appeared in \cite{BRSV08} as the PIA algorithm, and recently has received a lot of interest from the Bayesian inverse problem community due to the fact that it is well defined in the function space setting. In particular, it was shown in
\cite{hairer2011spectral} that  pCN achieves a dimension-independent
geometric rate of convergence for Gaussian-based target measures, that is, measures that have density with respect to a Gaussian measure. Below we also consider Gaussian-based target measures, and although we do not have uniform ergodicity of the pCN algorithm, we show that it is possible to perform unbiased estimation, by extending known contraction results for the pCN algorithm and using a combination of the techniques applied in sections \ref{sec:wasserstein} and \ref{sec:is}.

\subsection{Setup\label{sub:transSetup}}

We work in a separable Hilbert space $(\state,\pr{\cdot}{\cdot},\norm{\cdot})$,
and consider target mesaures $\mu$ which can be expressed as log-Lipschitz changes of measure
from a Gaussian reference measure $\mu_0$. In particular, let $\mu_0$ be a Gaussian measure in $\state$ with Karhunen-Loeve expansion (see section \ref{sec:Linear}) of the form \begin{equation}
\mu_0=\mathcal{L}\left(\sum_{\ell=0}^{\infty}\sqrt{\lambda_{\ell}}\gamma_{\ell}e_{\ell}\right),\quad\gamma_{\ell}\overset{\text{i.i.d.}}{\sim}\mathcal{N}\left(0,1\right),\;\lambda_{\ell}\lesssim\ell^{-2a},\label{eq:transTarget}
\end{equation}
where $\{e_\ell\}_{\ell\in\N}$ is a complete orthonormal basis in $\state$ and $a>\frac12$ is a regularity parameter. We consider the target measure $\mu$, given as
\begin{equation}
\frac{d\mu}{d\mu_{0}}(x)\propto\exp(-g(x)),\label{eq:tranTarget}
\end{equation}
where  $g:\state\to\R$
is Lipschitz continuous. 

We define the approximate reference measures $\mu_{0,j}$ through
the truncated Karhunen-Loeve expansion 
\[
\mu_{0,j}=\mathcal{L}\left(\sum_{l=0}^{j}\sqrt{\lambda_{\ell}}\gamma_{\ell}e_{\ell}\right).
\]
The measures $\mu_{0,j}$ are then supported on the $j-$dimensional space $\state_{j}\coloneq\text{span}\left\{ e_{1},\dots,e_{j}\right\} \subset\state$.
In the following, we identify the spaces $\state_{j}$ with the corresponding
subsets of $\state$ and denote by $\Pi_{j}$ the projection onto
$\state_{j}$. We consider the sequence of truncated target measures
$\mu_{j}$ defined through 
\begin{equation}
\frac{d\mu_{j}}{d\mu_{0,j}}(x)\propto\exp(-g(x)).\label{eq:transTrunc}
\end{equation}

Approximating an expectation with respect to $\mu$ by an expectation with respect to $\mu_j$, results in a discretisation error
which is quantified in  \cite{cotter2010approximation}, \cite{2013HoangComplexityMCMC}
and \cite{DS13}. Furthermore, the expectations with respect to $\mu_j$ are in general not available analytically but they are amenable to approximation using Markov chain Monte Carlo algorithms. In particular, we consider the Markov chains corresponding to the pCN algorithms
applied to $\mu^{j}$, see Algorithm \ref{alg:pCN}. We denote the resulting Markov chain by $X^j_{\cdot}$ and the corresponding Metropolis-Hastings Markov kernel by $P_{j}$. 
In a similar way to section \ref{sec:is}, in the next subsection we use appropriately coupled pCN algorithms to achieve the removal of both the discretisation bias as well as the bias introduced by the use of finite time distributions.

\begin{algorithm}
Fix $\rho\in(0,1)$. Generate $X_0^j$. Iterate the following steps for $k=1,\dots,K_{max}$:
\begin{enumerate}
\item[1.] $\xi^j\sim\mu_{0,j}$;
\item[2.] set $\hat{X}^j_{k+1}=\rho X^j_k+\sqrt{1-\rho^2}\xi^j$;
\item[3.] set $X_{k+1}^j=\hat{X}^j_{k+1}$ with probability 
\[
\alpha\left(X_k^j,\xi^j\right)=1\wedge\exp\left(g(X^j_k)-g(\hat{X}^j_{k+1})\right)
\]
and $X^j_{k+1}=X^j_k$ otherwise.
\end{enumerate}
\caption{\label{alg:pCN} pCN algorithm}
\end{algorithm}

\subsection{Unbiased estimation using the pCN algorithm}
Our aim is to obtain an unbiased estimator of $\E_{\mu}[f]$, for some function $f:\state\to\R$. As in section \ref{sec:is}, for two increasing sequences of integers $a_i$ and $j_i$ we would like to set $\Delta_i=f(X_{a_i}^{j_i})-f(X_{a_{i-1}}^{j_{i-1}})$ in the definition of $Z$ in Proposition \ref{prop:rhee}, where the
chains $X_{\cdot}^{j_{i}}$ and $X_{\cdot}^{j_{i-1}}$ are the (regular)
pCN chains introduced in the previous subsection
following the transition kernels $P_{j_{i}}$ and $P_{j_{i-1}}$,
respectively. For the unbiasing technique to work, we need to construct
an appropriate coupling between the two chains, so that $\norm{\Delta_{i}}_{2}$
decays sufficiently quickly for Proposition \ref{prop:rhee} to apply,
and the expected computing time is finite. In order to achieve this, we again generate $\Delta_{i}$
using a "top" level chain in $\state_{j_{i}}$
and a "bottom" level chain in $\state_{j_{i-1}}$,
which we denote by $\ux{\cdot}{i}$ and $\lx{\cdot}{i},$ and which perform $a_i$ and $a_{i-1}$ steps, respectively.
According to Proposition \ref{prop:rhee}, we need $\Delta_{i}$ to
be independent for different $i$, hence the two chains $\ux{\cdot}{i}$
and $\lx{\cdot}{i+1}$ both following the transition kernel $P_{j_{i}}$
in $\state_{j_{i}}$, are constructed independently. Nevertheless, 
the chains at different levels are coupled as follows: 
\begin{enumerate}
\item $\ux{\cdot}{i}$ is coupled to $\lx{\cdot}{i}$ which follows the
transition kernel $P_{j_{i-1}}$ on $\state_{j_{i-1}}$; 
\item $\lx{\cdot}{i+1}$ is coupled to $\ux{\cdot}{i+1}$ which follows
the transition kernel $P_{j_{i+1}}$ on $\state_{j_{i+1}}.$ 
\end{enumerate}
The  following diagram illustrates the construction of the $\Delta_{i}$:
\[
\begin{array}[t]{cccccccccc}
 &  &  &  & x_{0}= & \ux{-a_{0}}{0} & \dots & \ux{0}{0} & \}\Delta_{0} & =f(\ux{0}{0})\\
 &  &  &  & x_{0}= & \lx{-a_{0}}{1} & \dots & \lx{0}{1}\\
 &  &  &  &  & \mid & \mid & \mid & \}\Delta_{1} & =f(\ux{0}{1})-f(\lx{0}{1})\\
 &  & x_{0}= & \ux{-a_{1}}{1} & \dots & \ux{-a_{0}}{1} & \dots & \ux{0}{1}\\
 &  & x_{0}= & \lx{-a_{1}}{2} & \dots & \lx{-a_{0}}{2} & \dots & \lx{0}{2}\\
 &  &  & \mid & \mid & \mid & \mid & \mid & \}\Delta_{2} & =f(\ux{0}{2})-f(\lx{0}{2})\\
x_{0}= & \ux{-a_{2}}{2} & \dots & \ux{-a_{1}}{2} & \dots &  & \dots & \ux{0}{2}\\
x_{0}= & \lx{-a_{2}}{3} & \dots & \lx{-a_{1}}{3} & \dots &  & \dots & \lx{0}{3}
\end{array}
\]
where $|$ indicates coupling between two chains. The coupling is achieved by using the same random seed $\xi^i\sim\mu_{0,i}$ in the pCN proposal for $\ux{\cdot}{i}$ and $\lx{\cdot}{i}$, as well as the same uniform variable $U^i$ for acceptance or rejection. The random variables $\xi^i$ and $U^i$ are taken to be independent of each other, as well as independent of $\xi^j$ and $U^j$ for $j\neq i$. We use the random functions $\varphi_{\pux{}{}}^{i}, \varphi_{\plx{}{}}^{i}$ to denote the pCN proposals $\pux{\cdot}{i}$ and $\plx{\cdot}{i}$ for the chains $\ux{\cdot}{i}$ and $\lx{\cdot}{i}$, respectively, where 
\begin{align*}
\varphi_{\pux{}{}}^{i}\left(x,\xi^{i}\right)&\coloneq\rho x+\left(1-\rho^{2}\right)^{\frac{1}{2}}\xi^{i},\\
\varphi_{\plx{}{}}^{i}\left(x,\xi^{i}\right)&\coloneq\rho x+\left(1-\rho^{2}\right)^{\frac{1}{2}}\Pi_{j_{i-1}}\xi^{i}.
\end{align*}
Furthermore,
we use the random functions $\varphi_{\ux{}{}}^{i}, \varphi_{\lx{}{}}^{i}$ to represent the chains $\ux{\cdot}{i}$ and $\lx{\cdot}{i}$, respectively, where
\begin{align*}
\varphi_{\ux{}{}}^{i}\left(\ux{k-1}{i},W_{k}^{i}\right) & \coloneq\ind_{[0,\alpha\left(\ux{k-1}{i},\pux{k}{i}\right)]}(U_{k}^{i})\pux{k}{i}+\ind_{(\alpha\left(\ux{k-1}{i},\pux{k}{i}\right),1]}(U_{k}^{i})\ux{k-1}{i},\\
\varphi_{\lx{}{}}^{i}\left(\lx{k-1}{i},W_{k}^{i}\right) & \coloneq\ind_{[0,\alpha(\lx{k-1}{i},\plx{k}{i})]}(U_{k}^{i})\plx{k}{i}+\ind_{(\alpha(\lx{k-1}{i},\plx{k}{i}),1]}(U_{k}^{i})\lx{k-1}{i}.
\end{align*}
The construction of $\Delta_i$ is given in detail in Algorithm \ref{alg:transpCN}.

\begin{algorithm}
Fix a starting point $x_{0}\in\state_{j_{0}}$ once and for all. For
$i=0$, generate $\Delta_{0}$ as follows: 
\begin{enumerate}
\item set $\ux{-a_{0}}{0}=x_{0}$ on $\state_{j_{0}}$ and simulate according
to $P_{j_{0}}$ up to $\ux{0}{0}$; 
\item set $\Delta_{0}=f\left(\ux{0}{0}\right)$. 
\end{enumerate}
For, $i\geq1$, generate $\Delta_{i}$ as follows: We generate $\Delta_{i}$
for $i\ge1$ as follows: 
\begin{enumerate}
\item set $\ux{-a_{i}}{i}=x_{0}$ and run the chain until $\ux{-a_{i-1}}{i}$
according to $P_{j_{i}}$; 
\item set $\lx{-a_{i-1}}{i}=x_{0}$; 
\item for $k=-a_{i-1}+1,\dots,0$ run $\ux{\cdot}{i}$ and $\lx{\cdot}{i}$
as coupled pCN algorithms, as described below: 
\begin{enumerate}
\item draw $\xi_{k}^{i}\sim\mu_{0,j_{i}}$ and $U_{k}^{i}\sim{\rm U}[0,1]$
independently from everything else and set $W_{k}^{i}=\left(\xi_{k}^{i},U_{k}^{i}\right)$
as the collection of all random inputs for the
$k$-th step; 
\item propose % 
\begin{align*}
\pux{k}{i} & =\varphi_{\pux{}{}}^{i}\left(\ux{k-1}{i},\xi_{k}^{i}\right)=\rho\ux{k-1}{i}+\left(1-\rho^{2}\right)^{\frac{1}{2}}\xi_{k}^{i},\\
\plx{k}{i} & =\varphi_{\plx{}{}}^{i}\left(\lx{k-1}{i},\xi_{k}^{i}\right)=\rho\lx{k-1}{i}+\left(1-\rho^{2}\right)^{\frac{1}{2}}\Pi_{j_{i-1}}\xi_{k}^{i};
\end{align*}
\item set 
\begin{align*}
\ux{k}{i}=\varphi_{\ux{}{}}^{i}\left(\ux{k-1}{i},W_{k}^{i}\right) & =\ind_{[0,\alpha\left(\ux{k-1}{i},\pux{k}{i}\right)]}(U_{k}^{i})\pux{k}{i}+\ind_{(\alpha\left(\ux{k-1}{i},\pux{k}{i}\right),1]}(U_{k}^{i})\ux{k-1}{i},\\
\lx{k}{i}=\varphi_{\lx{}{}}^{i}\left(\lx{k-1}{i},W_{k}^{i}\right) & =\ind_{[0,\alpha(\lx{k-1}{i},\plx{k}{i})]}(U_{k}^{i})\plx{k}{i}+\ind_{(\alpha(\lx{k-1}{i},\plx{k}{i}),1]}(U_{k}^{i})\lx{k-1}{i};
\end{align*}
\end{enumerate}
\item Set $\Delta_{i}=f(\ux{0}{i})-f(\lx{0}{i})$. 
\end{enumerate}
\protect\protect\caption{\label{alg:transpCN}Coupled pCN algorithms for unbiased estimation}
\end{algorithm}

For $W\sim\mu_{0,j_{i}}\otimes{\rm U}[0,1]$, we define 
\begin{eqnarray}
K_{j_{i-1}}^{j_{i}}\left((x_{1},x_{2}),\cdot\right) & \coloneq & \mathcal{L}\left(\varphi_{\lx{}{}}^{i}\left(x_{2},W\right),\varphi_{\ux{}{}}^{i}\left(x_{1},W\right)\right),\label{eq:coupl1}\\
K_{j_{i}}^{j_{i}}\left((x_{1},x_{2}),\cdot\right) & \coloneq & \mathcal{L}\left(\varphi_{\ux{}{}}^{i}\left(x_{1},W\right),\varphi_{\ux{}{}}^{i}\left(x_{2},W\right)\right),\label{eq:coupl2}
\end{eqnarray}
that is, $K_{j_{i-1}}^{j_i}$ and $K_{j_i}^{j_i}$ are the couplings between $P_{j_{i-1}}(x_{1},\cdot)$ and $P_{j_{i}}(x_{2},\cdot)$
and $P_{j_{i}}(x_{1},\cdot)$ and $P_{j_{i}}(x_{2},\cdot)$, respectively.
In order to simplify the notation we will write $K_{i}$ for $K_{j_{i}}^{j_{i}}$.

In contrast to section \ref{sec:is}, in the present setting we do not have uniform ergodicity, thus we can only rely on the contracting property of the pCN algorithm
in some distance (or distance-like function) $d$, in order to get the required decay of $\norm{\Delta_i}_2$ for the unbiasing programme to work. We stress here, that the readily available results in the literature concern the contraction of the pCN algorithm at a fixed dimension $j_i$, and in particular the contraction of the coupling $K_i$ in certain distances; see the results of Durmus and Moulines  \cite{durmus2014new}, Durmus et al. \cite{durmus2014geom} and Hairer et al. \cite{hairer2011spectral}. Instead, we need to work harder in order to show a form of contraction of the transdimensional coupling $K_{j_{i-1}}^{j_i}$ in the same distances, which happens asymptotically as $i\to\infty$. We achieve this by using the triangle inequality to combine the existing contraction results at a fixed level $i$, with estimates on the large $i$ behaviour of the transdimensional coupling when the two chains are started from the same initial condition. Once we get the appropriate behaviour of $K_{j_{i-1}}^{j_i}$ in $d$, it is straightforward to obtain estimates of the decay of $\norm{\Delta_i}_2$  in a similar way to section \ref{sec:wasserstein}, which hold for functions $f$ having sufficient H\"older regularity in the same distance $d$.

We work under the following assumption on the log-change of measure $g$, which ensures the contraction of the pCN algorithm in a fixed state space (see subsection \ref{sec:apppCN} for details).
\begin{assumption}\label{ass:lippCN}
The function $g:\state\to\R$ is globally Lipschitz and there exist positive constants
$C,\, R_{1},\: R_{2},$ such that for $x\in\state$ with $\left\Vert x\right\Vert \ge R_{1}$
\begin{equation}\label{eq:pCNballCondition}
\inf_{z\in B(\rho x,R_{2})}\exp\left(g(x)-g(z)\right)>C,
\end{equation}
where $\rho$ is as in the definition of the pCN algorithm.%if $\rho=0$ g needs to be bdd
\end{assumption}

We first consider the distance $d_{\tau}=1\wedge\frac{\norm{x-y}}{\tau}$. In Lemma \ref{lem:appPCNdBound}, we derive results of the
form 
\begin{equation}
\E d_\tau\left(\ux{0}{i},\lx{0}{i}\right)\leq Cr^{a_{i-1}}+C_{j_{i-1},j_{i}},\label{eq:transBdd}
\end{equation}
where $C_{j_{i-1},j_{i}}$ is a constant only depending on $j_{i-1}<j_{i}$,
such that 
\[
C_{j_{i-1},j_{i}}\rightarrow0\text{ as }i\rightarrow\infty.
\]
Explicit bounds on $C$ and $r$ can be obtained as outlined in subsection \ref{sec:apppCN}.
We note here, that the bound in \eqref{eq:transBdd} agrees with the qualitative
behaviour that we observe in simulations, see Figure \ref{fig:couplingDifferentDimensions}.

\begin{figure}
\begin{center}
\includegraphics[width=0.75\textwidth]{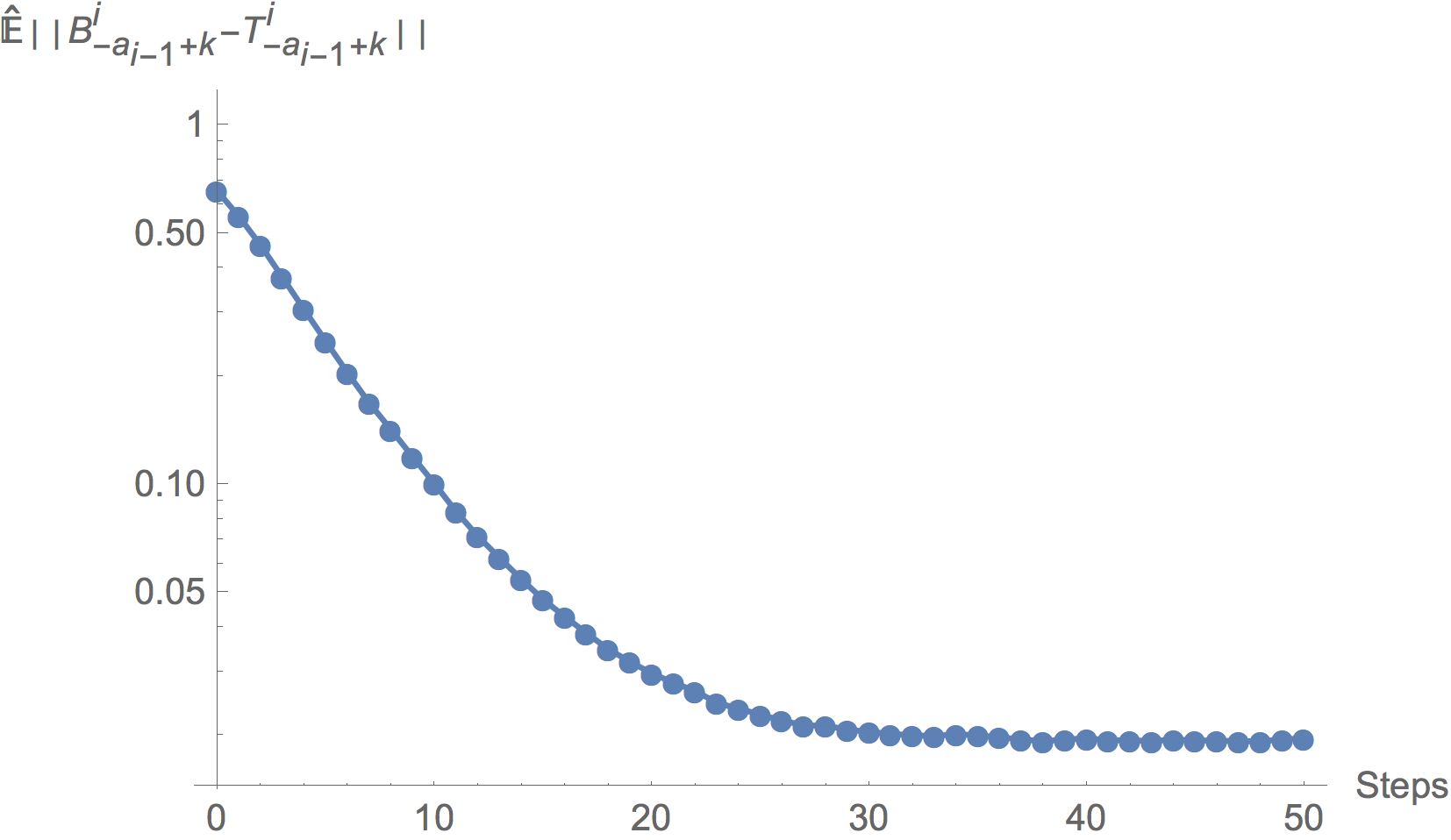}
\caption{\label{fig:couplingDifferentDimensions}Estimates of $\E\left\Vert \ux{-a_{i-1}+k}{i}-\lx{-a_{i-1}+k}{i}\right\Vert $ based on 10000 runs, plotted against $k$. Here $g(x)=\left\Vert x\right\Vert $, $\rho=0.7$, $j_{i-1}=15$,
$j_{i}=17$ and $a_i-a_{i-1}=8$.}
\end{center}
\end{figure}

We consider unbiased estimation of $\E_{\mu}[f]$, where $f$ is $s$-H\"older for $s\in[\frac{1}{2},1]$ with
respect to the distance $d_{\tau}$. We note that
this class of functions $f$ does not depend on the choice of $\tau>0$. For such a function $f$, the boundedness of the distance $d_\tau$ implies the bound
\begin{equation}
\norm{\Delta_{i}}_{2}^{2}\leq\left\Vert f\right\Vert _{s}^{2}\E d_{\tau}(\ux{0}{i},\lx{0}{i})^{2s}\leq\left\Vert f\right\Vert _{s}^{2}\E d_{\tau}(\ux{0}{i},\lx{0}{i}).\label{eq:pcnDeltaBound}
\end{equation}
Balancing the two terms on the right hand side of \eqref{eq:transBdd}, gives rise
to sufficiently sharp bounds on $\left\Vert \Delta_{i}\right\Vert _{2}$, see Lemma \ref{lem:appPCNdBound} again.

In order to follow the unbiasing programme, we pose the following assumption
on the expected computing time.
\begin{assumption}
\label{ass:pcncost} The expected computing time to simulate $K_{j_{i-1}}^{j_{i}}$
satisfies 
\[
s_{i}\lesssim j_{i}^{\theta}
\]
with $\theta\ge1.$Therefore, since we need $a_{i}$ steps of the
chain to generate $\Delta_{i}$, the expected computing time $t_{i}$
of $\Delta_{i}$ satisfies 
\[
t_{i}\lesssim a_{i}j_{i}^{\theta}.
\]
\end{assumption}
We have the following result on the estimator $Z$ defined in \eqref{eq:estimator}:
\begin{thm}
\label{thm:SuffCond}Assume that the target measure $\mu$ is given as in \eqref{eq:transTarget}, where $g$ satisfies Assumption \ref{ass:lippCN}. Suppose that  Assumption \ref{ass:pcncost} is
satisfied for $\theta\geq1$ and let $f:\state\to\R$ be $s$-H\"older continuous with respect to $d_{\tau}$, for some $s\in[\frac12,1]$. Assume that $a>\theta+\frac{1}{2}$, where $a$ represents the regularity of the reference measure, see
\eqref{eq:transTarget}. Then there are choices of $a_{i}$, $j_{i}$ and
$\rp(N\geq i)$, such that\[
Z=\sum_{i=0}^{N}\Delta_{i}
\]
is an unbiased estimate of $\E_{\mu}[f]$ with finite variance and finite
expected computing time. For example, for any $m\in\N$, this works if we choose $a_{i}=m i$,
$j_i\sim r^{\frac{2mi}{1-2a}}$ and $\rp(N\ge i)\propto r^{(m-\epsilon)i}$, where $\epsilon\in(0,\frac{2\theta m}{1-2a}+m)$.
\end{thm}

Note that the choice of $m$ does not affect the finiteness of the variance or the expected computing time of $Z$. However, our intuition from the numerical experiments presented in section \ref{sec:contr} for problems of fixed dimension, suggests that a good choice of $m$ has a large impact on the efficiency of the algorithm (see Figure \ref{fig:conNormChoiceOfM}). We expect this to be the case in the transdimensional setting too, and for this reason choose to allow this flexibility in the formulation of the theorem.

The last result shows that the unbiasing procedure can be applied for estimating posterior expectations with respect to functions that are H\"older continuous with respect to the bounded distance $d_{\tau}$. In particular $f$ needs to be bounded which does not
allow the estimation of the mean or the second moment. We now show that it is possible to obtain unbiased estimates for unbounded functions,
under a stronger assumption on the regularity of the reference measure $\mu_0$. This is achieved
by considering the distance-like function $\tilde{d}(x,y)\coloneq\sqrt{d_{\tau}(x,y)\left(1+V(x)+V(y)\right)}$
with $V(x)=\exp\left(\left\Vert x\right\Vert \right)$. 

Indeed, in Lemma \ref{lem:appPCNdtildeBound} we obtain bounds of the form
\begin{equation}
\E \tilde{d}\left(\ux{0}{i},\lx{0}{i}\right)\lesssim r^{a_{i-1}}+C_{j_{i-1},j_{i}}^{\frac12},\label{eq:transuBdd}
\end{equation}
where $C_{j_{i-1}}^{j_i}$ is the same constant as in \eqref{eq:transBdd} and $r\in(0,1)$.
Since $\tilde{d}$ is unbounded, a bound of the type of \eqref{eq:pcnDeltaBound} is not possible for general $s\geq \frac12$, and so we need to restrict ourselves 
to the estimation of $\E_{\mu}[f]$ where $f$ is $\frac12$-H\"older continuous in $\tilde{d}$. In this case we immediately have \begin{equation}\label{eq:delunb}
\norm{\Delta_i}^2_2\leq\norm{f}^2_{\frac12}\E\tilde{d}(\ux{0}{i},\lx{0}{i}),
\end{equation}
and as before, we can balance the two terms on the right hand side of (\ref{eq:transuBdd}) to get sufficiently sharp bounds on $\norm{\Delta_i}_2$, see Lemma \ref{lem:appPCNdtildeBound}. Note that the square root on $C_{j_{i-1}}^{j_i}$ is the source of the stronger assumption on the regularity of the reference measure $\mu_0$. We get the following result:

\begin{thm}
\label{thm:pCNSuffCondUnbounded}Assume that the target measure $\mu$ is given as in \eqref{eq:transTarget}, where $g$ satisfies Assumption \ref{ass:lippCN}. Suppose that  Assumption \ref{ass:pcncost} is
satisfied for $\theta\geq1$ and let $f:\state\to\R$ be $\frac12$-H\"older continuous with respect to $\tilde{d}$. Assume that $a>2\theta+\frac{1}{2}$, where $a$ represents the regularity of the reference measure, see
\eqref{eq:transTarget}. Then there are choices of $a_{i}$, $j_{i}$ and
$\rp(N\geq i)$, such that\[
Z=\sum_{i=0}^{N}\Delta_{i}
\]
is an unbiased estimate of $\E_{\mu}[f]$ with finite variance and finite
expected computing time. For example, for any $m\in\N$, this works if we choose $a_{i}=m i$,
$j_i\sim r^{\frac{4mi}{1-2a}}$ and $\rp(N\ge i)\propto r^{(m-\epsilon)i}$, where $\epsilon\in(0,\frac{4\theta m}{1-2a}+m)$.\end{thm}

\begin{rem}
Let $(H, \pr{\cdot}{\cdot}_{H}, \norm{\cdot}_H)$ be another Hilbert space.
Using Proposition \ref{prop:genrhee} which generalises Proposition
\ref{prop:rhee}, it is straightforward to check that Theorems \ref{thm:SuffCond}
and \ref{thm:pCNSuffCondUnbounded} can be extended to the estimation of
expectations of functions $f:\state\to H$ which are H\"older
continuous. In particular, using Theorem \ref{thm:pCNSuffCondUnbounded}, we can perform unbiased estimation of all
moments of $\mu$. 

Indeed, observe that all functions $f:\state\rightarrow H$ satisfying
$\norm{f(x)-f(y)}_{H}\leq C\norm{x-y}^{\frac{1}{4}}\exp\left(\frac{1}{8}(\norm x\vee\norm y)\right)$
are $\frac{1}{2}$-H\"older continuous with respect to $\tilde{d};$
this follows by separate inspection of the cases $\frac{\norm{x-y}}{\tau}\leq1$ and $\frac{\norm{x-y}}{\tau}>1$. In the former 
\begin{align*}
\norm{f(x)-f(y)} & _{H}\leq\frac{C\tau^{\frac{1}{4}}\norm{x-y}^{\frac{1}{4}}}{\tau^{\frac{1}{4}}}\left(\exp\Big(\frac{1}{2}(\norm x\vee\norm y)\Big)\right)^{\frac{1}{4}}\leq C\tau^{\frac{1}{4}}\tilde{d}(x,y)^{\frac{1}{2}},
\end{align*}
while in the latter 
\begin{eqnarray*}
\norm{f(x)-f(y)}_{H} & \leq & C\left(\norm x^{\frac{1}{4}}\vee\norm y^{\frac{1}{4}}\right)\exp\left(\frac{1}{8}(\norm x\vee\norm y)\right)\\
 & \leq & \tilde{C}\exp\left(\frac{1}{4}(\norm x\vee\norm y)\right)\leq\left(\exp(\norm x\vee\norm y)\right)^{\frac{1}{4}}\\
 & \leq & \tilde{C}\tilde{d}(x,y)^{\frac{1}{2}}.
\end{eqnarray*}
Using this observation, it is straightforward to check that we can apply the unbiasing procedure to $f(x)=x$ and $f(x)=x\otimes x$
(or to the finite dimensional approximations $f(x)=\Pi_{j}x$ and
$f(x)=\Pi_{j}x$$\left(\Pi_{j}x\right)^{t}$) to obtain unbiased estimates
of the mean and the second moment, respectively. 
\end{rem}

\begin{rem}
In this section we focused on the discretisation of the
input of $g$, $x$. However, in most practical scenarios like those
arising in Bayesian inverse problems, $g$ is based on a solution operator
to a Partial Differential Equation and  hence $g$ itself needs to be discretised, say by $g^{l}$. We provide
an example of how it is possible to do this in the setting for uniformly ergodic Markov chains in section
\ref{ssec:elipdisc}. In order to make possible the unbiased estimation using the pCN algorithm in practical problems, the analysis in
this section needs to be adapted accordingly. This is beyond the scope of the present paper, but it will be the topic of follow-up work.
\end{rem}

%%%%%%%%%%%%%%%%% NUMERICS %%%%%%%%%%%%%%%%
\section{Comparison of the unbiasing procedure and the ergodic average\label{sec:contr}}

In section \ref{sec:wasserstein} we have shown how the unbiasing
procedure can be applied to the estimation of expectations with respect to the invariant distribution $\pi$ of a Markov chain that exhibits a simulatable
contracting coupling. The existence of such a coupling implies that
the Markov chain is ergodic, thus, the ergodic average constitutes
a consistent estimator of $\E_{\pi}[f]$, for sufficiently nice functions $f$. In this section we investigate how estimators constructed
by averaging over independent runs of the unbiasing procedure perform compared to the ergodic average.

We compare the two methods using the Mean Square Error - work product (MSE-work product) 
\begin{equation}
\text{MSE}\times\E\left(\text{ computing time}\right),\label{eq:performMeasure}
\end{equation}
which has also been used as a performance measure in \cite{rhee2012new},
in the setting of unbiased estimation of expectations with respect to diffusions. For estimators constructed by averaging over unbiased estimators, the MSE-work product has the attractive property that it does not depend on the
number of instances $L$ that are averaged over. The reason for this
is that the variance is scaled by $\frac{1}{L}$ whereas the expected
computing time is multiplied by $L$. Using Proposition \ref{prop:rhee}
and the expression \eqref{eq:expectedTime}, we see that the MSE-work product
for the unbiasing procedure studied in the present paper is 
\begin{equation}
\left(\sum_{i}\frac{\nu_{i}}{\bar{F_{i}}}-\left(\E_\pi [f]\right)^{2}\right)\left(\sum_{i}\bar{F}_{i}t_{i}\right).\label{eq:contrUBmsework}
\end{equation}
Here $t_{i}$ denotes the expected computing time to generate $\Delta_{i}$, $\bar{F_i}=\rp(N\geq i)$ and  
\begin{equation}\label{eq:contrNormnus}
\nu_{i}=\norm{\Delta_i}^2_2+2\E\Delta_{i}\left(\E Y-\E Y_{i}\right)=\text{Var}(\Delta_{i})+\left(\E Y-\E Y_{i-1}\right)^{2}-\left(\E Y-\E Y_{i}\right)^{2},
\end{equation}
where $Y\sim f_\star\pi$ and $Y_i=\sum_{k=0}^i\Delta_k$.
 
There are (uncountably) many choices of the number of time steps $a_{i}$ used to construct $\Delta_i$ in Algorithm \ref{alg:CouplingWasserstein}, and the probabilities $\bar{F}_{i}$, that yield unbiased
estimators with finite variance and finite expected computing time.
For a fair comparison with the ergodic average we need to optimise
the MSE-work product with respect to $a_{i}$ and $\bar{F}_{i}$. 
Since this is difficult in general, we consider the example of 1-dimensional
contracting normals in section \ref{sub:compContracting-Normals}. We note that this example is also covered by the 
theory in \cite{RheePHD}, however we use it to
\begin{itemize}
\item compare the performance of the ergodic average of the Markov chain with
the average of unbiased estimators of the type presented in section \ref{sec:wasserstein};
\item show that the added flexibility of choosing $a_i$, is crucial for optimizing the performance of the unbiased estimator (note that in \cite{RheePHD} $a_i$ is restricted to be equal to $i$);
\item illustrate that we do not need sharp bounds on the properties of the
coupling in order to tune the unbiased estimator;
\item show numerical results suggesting that in a parallel setting the unbiasing procedure can
be superior.
\end{itemize}
In section \ref{sub:Logistic-Regression} 
we consider posterior inference for a Bayesian logistic
regression model and get the same findings as for contracting normals. Even though we cannot verify the contracting assumption of section
\ref{sec:wasserstein}, we demonstrate that even a naive implementation of the unbiasing
procedure leads to a competitive algorithm.

\subsection{Contracting normals\label{sub:compContracting-Normals}}
We consider the example of 1-dimensional contracting normals, that is,
 the Markov chain defined by 
\begin{equation}
X_{n+1}=\rho X_{n}+\sqrt{1-\rho^{2}}\xi_{n+1},\label{eq:contractingNormals}
\end{equation}
for $\rho\in(0,1)$ and $\xi_{n}\overset{\text{i.i.d.}}{\sim}\mathcal{N}(0,1).$ This Markov chain is {ergodic} with the standard normal distribution as invariant distribution, that is $\pi=\mathcal{N}(0,1)$.
The construction of the unbiased estimator $Z$ follows from section
\ref{sec:wasserstein}, by considering the coupling 
\[
K\left((x,y),(dx^{\prime},dy^{\prime})\right)=\mathcal{L}\left(\rho x+\sqrt{1-\rho^{2}}\xi,\rho y+\sqrt{1-\rho^{2}}\xi\right),
\] where $\xi\sim\mathcal{N}(0,1)$.
It is straightforward to check that this coupling satisfies Assumption \ref{assu:wass}.i. with geometric rate of contraction $r=\rho$, for the distance $d(x,y)=|x-y|$. The corresponding "top" and "bottom" chains have the form 
\begin{align*}
\ux{k+1}{i} & =\rho\ux{k}{i}+\sqrt{1-\rho^{2}}\xi_{k+1}^{i},\\
\lx{k+1}{i} & =\rho\lx{k}{i}+\sqrt{1-\rho^{2}}\xi_{k+1}^{i},
\end{align*}
where $\xi_{k}^i\overset{\text{i.i.d.}}{\sim}\mathcal{N}(0,1)$.
The expected computing
time is $t_{i}=T_{\text{step}}\times a_{i}$, where $T_{\text{step}}$ is the expected computing time to simulate one step of the chain, while the $\nu_{i}$ can be bounded
using the bounds on $\norm{\Delta_i}_2$. %

For this chain there are analytic expressions for $\nu_{i}$ if
we consider the estimation of $\E_{\pi}\left[f\right]$ for $f$ being
a polynomial. In the following we consider the simple function $f(x)=x$, which is trivially Lipschitz in $d$ so that Theorem \ref{thm:wasserstein} applies. In this case we simply have that $\Delta_i=\ux{0}{i}-\lx{0}{i}$.

In subsection \ref{sub:ContrMSE-Work-MCMC}, we find an explicit asymptotic expression
for the MSE-work product for the ergodic average.
We discuss the problem of finding good choices of $a_{i}$ and
$\bar{F}_{i}$ for the unbiasing procedure in subsection \ref{sub:ContrMSE-Work-Product-UB}.
Even though we are not able to give a satisfying answer to the optimisation
problem, we show in subsection \ref{sub:contrNormalsInformed} that informed
choices of $a_{i}$ and $\bar{F}_{i}$ lead to a competitive performance of the unbiased estimator compared to the ergodic average, as measured by the MSE-work product. 
Such infromed choices require precise knowledge of $\nu_i$, which in practice is not available. In section \ref{sub:tuning}, we investigate the effect on the optimisation over $\bar{F}_{i}$ for fixed $a_i$, of using the exact values $\nu_i$ for $i\leq i_0$ and only upper bounds for $i>i_0$. We demonstrate that this already leads to a considerable improvement over using upper bounds for all $i$.  
Finally,  in subsection \ref{sub:ContrComparisonInParallelSetting} we present a comparison of the unbiasing procedure
and the ergodic average in terms of computing time in the parallel
computing setting.
This comparison is not exhaustive but suggests future investigation. 

\subsubsection{\label{sub:ContrMSE-Work-MCMC}The MSE-work product for the ergodic
average}

The MSE-work product of the ergodic average for $f(x)=x$ for contracting
normals can be calculated explicitly. Indeed, we first iterate \eqref{eq:contractingNormals}
to obtain 
\[
\sum_{i=0}^{n}X_{i}=\frac{1-\rho^{n+1}}{1-\rho}X_{0}+\sum_{i=1}^{n}\xi_{i}\sum_{j=0}^{n-i}\rho^{j}\sqrt{1-\rho^{2}}.
\]
Using this formula, we obtain an expression for the MSE as follows

\begin{align*}
\E\left(\frac{1}{n}\sum_{i=0}^{n}X_{i}-0\right)^{2} & =\frac{\E X_{0}^{2}}{n^{2}}\left(\frac{1-\rho^{n+1}}{1-\rho}\right)^{2}+\frac{\left(1-\rho^{2}\right)}{n^{2}}\sum_{i=1}^{n}\left(\frac{1-\rho^{n-i+1}}{1-\rho}\right)^{2}\\
 & =\frac{1}{n^{2}}\left(\frac{1-\rho^{n+1}}{1-\rho}\right)^{2}\E X_{0}^{2}+\frac{1}{n}\frac{1+\rho}{1-\rho}\frac{1}{n}\sum_{i=1}^{n}\left(1-\rho^{n-i+1}\right)^{2}.
\end{align*}
This allows us to calculate the asymptotic performance as $n\rightarrow\infty$
\begin{equation}
\lim_{n\rightarrow\infty}\text{MSE}\times\E\,\left(\text{computing time}\right)=\frac{1+\rho}{1-\rho}T_{\text{step}}.\label{eq:contrNormalsMCMCeff}
\end{equation}
It is important to note that non-asymptotic
effects such as burn-in lead to a worse MSE-work product for finite
$n.$

For general Markov chains the expression in \eqref{eq:contrNormalsMCMCeff}
generalises to
\[
 \text{Var}_{\pi} (f)\frac{1+\rho}{1-\rho}T_{\text{step}}
\]
 which is an asymptotic upper bound on the MSE-work product if $\rho$ denotes the
 $L^{2}$-spectral gap. This result can be found in \cite{explicitbdd}.

\subsubsection{\label{sub:ContrMSE-Work-Product-UB}The MSE-work product for estimators
based on the unbiasing procedure}

For contracting normals the expressions for $\nu_{i}$ can be derived
analytically using \eqref{eq:contrNormnus}.  For simplicity we consider $X_0=0$ (so that in Algorithm \ref{alg:CouplingWasserstein}, we set $x_0=0$) for which we obtain
\begin{align}
\nu_{0} & =\left(1-\rho^{2a_{0}}\right),\nonumber\\
\nu_{i} & =\rho^{2a_{i-1}}\left(1-\rho^{2(a_{i}-a_{i-1})}\right).\label{eq:nufor0} \end{align}

Thus, the optimisation of the MSE-work product is similar to the one
encountered in \cite{rhee2012new} for unbiased estimation of
expectations based on diffusions. More precisely, the authors of \cite{rhee2012new} consider the optimisation
problem 
\begin{eqnarray}
\min_{\bar{F}} &  & \left(\sum_{n}\frac{\nu_{n}}{\bar{F}_{n}}\right)\left(\sum_{n}\bar{F}_{n}t_{n}\right)\label{eq:optimisation}\\
\text{subject to} &  & \bar{F}_{i}\ge\bar{F}_{i+1}\nonumber \\
 &  & \bar{F}_{i}>0\nonumber \\
 &  & \bar{F}_{0}=1.\nonumber
\end{eqnarray}
They show using the Cauchy-Schwarz inequality, that the choice
\begin{equation}
\rp(N\ge i)=\bar{F}_{i}=\frac{\sqrt{\frac{\nu_{i}}{t_{i}}}}{\sqrt{\frac{\nu_{0}}{t_{0}}}},\label{eq:conNormOptimalChoice}
\end{equation}
gives rise to the lower bound 
\begin{equation}
\left(\sum_{n}\sqrt{\nu_{n}t_{n}}\right)^{2}.\label{eq:toyoptimalMSEwork}
\end{equation}
Therefore the minimum is attained by this choice of $\bar{F}_i$ provided that it is feasible, that is, provided $\nu_i/t_i$
is decreasing.

In the setting of \eqref{eq:nufor0}, we have the following explicit optimisation problem

\begin{eqnarray}
\text{min} &  & \left(\sum_{i=1}^{\infty}\frac{\rho^{2a_{i-1}}\left(1-\rho^{2(a_{i}-a_{i-1})}\right)}{\bar{F}_{i}}+1-\rho^{2a_{0}}\right)\sum_{i=0}^{\infty}\bar{F}_{i}a_{i}\label{eq:contrNormUBmseWork}\\
\text{subject to} &  & \bar{F}_{0}=1\ge\bar{F}_{1}\ge\bar{F}_{2}\ge\dots\nonumber \\
 &  & \bar{F}_{i}>0\nonumber \\
 &  & 0<a_{0}<a_{1}<\dots\nonumber \\
 &  & a_{i}\in\mathbb{N}.\nonumber 
\end{eqnarray}
In contrast to \cite{rhee2012new}, we want to optimise
the MSE-work product with respect to both $\bar{F}_i$ and $a_{i}$.
However, even in this simple case we do not know the solution, but instead 
present a comparison based on informed choices of $a_i$ and $\bar{F}_i$ in the next subsection.

\subsubsection{Initial results based on informed parameter choices\label{sub:contrNormalsInformed}}

The minimisation over both $a_{i}$ and $\bar{F}_{i}$
could be achieved by first minimising over $\bar{F}_{i}$
for fixed $a_{i}$ and then minimising the resulting
expression over $a_{i}$. If for $a_{i}$
the choice of  $\bar{F}_{i}$  given in  (\ref{eq:conNormOptimalChoice})
is feasible, then the minimum is given by  \eqref{eq:toyoptimalMSEwork}.
If it is not feasible the minimisation over $\bar{F}_i$
is not clear. 

Even though we cannot optimise explicitly over all choices of $a_{i}$,
we do so over the sub class $a_i=mi+m$ for $m\in\mathbb{N}.$ The
expected  computing time of $\Delta_i$, $t_{i}=T_{\text{step}}a_{i}$,
is monotonically increasing. Moreover, it is straightoforward to check for this choice of $a_i$, that
$\nu_{i}$ is decreasing such that the choice of $\bar{F}_i$ in (\ref{eq:conNormOptimalChoice})
is feasible. As a result, this choice of $\bar{F}_i$ gives rise to the optimal MSE-work product of the unbiased estimator for any fixed $m$, and the corresponding (optimal) MSE-work product can be obtained using \eqref{eq:toyoptimalMSEwork} as follows
\begin{align*}
 & \left(\sum_{i=1}^{\infty}\sqrt{\rho^{2mi}(1-\rho^{2m})T_{\text{step }}m(i+1)}+\sqrt{T_{\text{step }}m(1-\rho^{2m})}\right)^{2}\\
 & =T_{\text{step}}\left(\sqrt{m(1-\rho^{2m})}\sum_{j=1}^{\infty}\rho^{m(j-1)}\sqrt{j}\right)^{2}\\
 & =T_{\text{step}}\left(\rho^{-m}\sqrt{m(1-\rho^{2m})}\text{Li}_{-\frac{1}{2}}\rho^{m}\right)^{2},
\end{align*}where $\text{Li}$ denotes the polylogarithm function. Subsequently,
we assume that $T_{\text{step}}=1$ since it is only a multiplicative
constant of the minimum and it does not change the optimal choice of $\bar{F}_{i}$
in \eqref{eq:conNormOptimalChoice}.

We compare the MSE-work product of the ergodic average, given in \eqref{eq:contrNormalsMCMCeff},
to the optimal MSE-work product of the unbiased estimator for a fixed $m$. Again,
we would like to stress that this comparison is advantageous for the
ergodic average because we disregard non-asymptotic effects such as
burn-in. In Figure
\ref{fig:conNormFixedai} we plot the MSE-work product of the ergodic average, the optimal MSE-work product of the unbiased estimator for $m=4$, and their ratio, as functions of $\rho$. We observe that as $\rho$ increases
towards $1$ the ratio of the MSE-work product of the unbiasing procedure over
 the one of the ergodic average explodes. 
\begin{figure}
 \begin{center}
\includegraphics[width=0.65\textwidth]{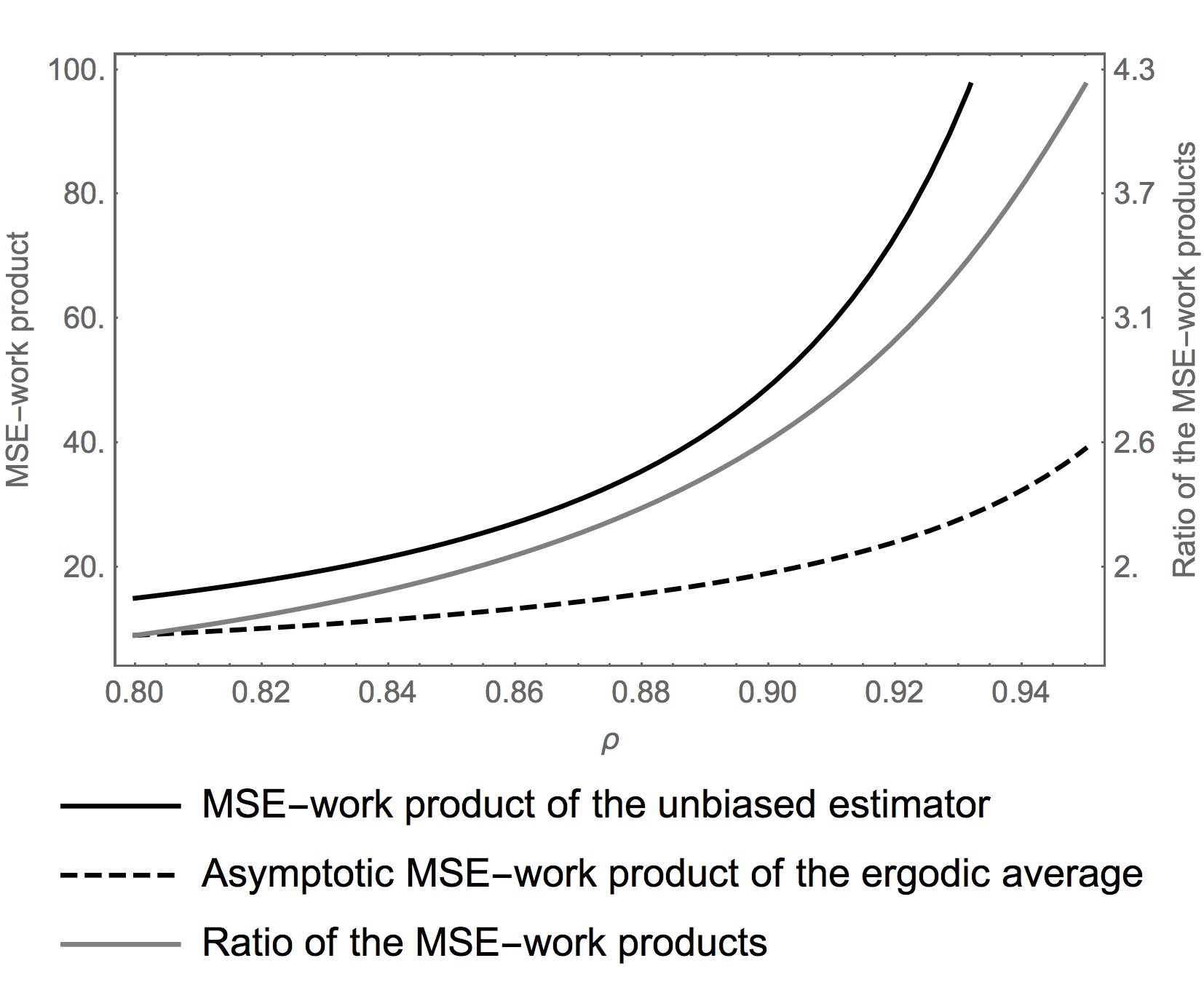} \protect\caption{\label{fig:conNormFixedai}MSE-work products for the ergodic  average and the unbiasing procedure for $a_{i}=4i+4$ and $\bar{F}_i$ chosen optimally, plotted against $\rho$, in the contracting normals example. In different scale we plot the ratio of  the MSE-work product of the unbiased estimator over the MSE-work product of the ergodic average.}
\end{center}\end{figure}

In an effort to improve the performance of the unbiased estimator, we allow $m$ to depend on $\rho.$ In order to illustrate the impact of $m$, we plot the MSE-work product
as a function of $m$ for different values of $\rho$ in Figure \ref{fig:conNormChoiceOfM}. We observe that for small values of $m$ the MSE-work product is very large, however for the optimal choice of $m$ the value of MSE-work product is relatively small. As $\rho\to1,$ the optimal value of $m$ increases.
\begin{figure}\begin{center}
\includegraphics[width=0.65\textwidth]{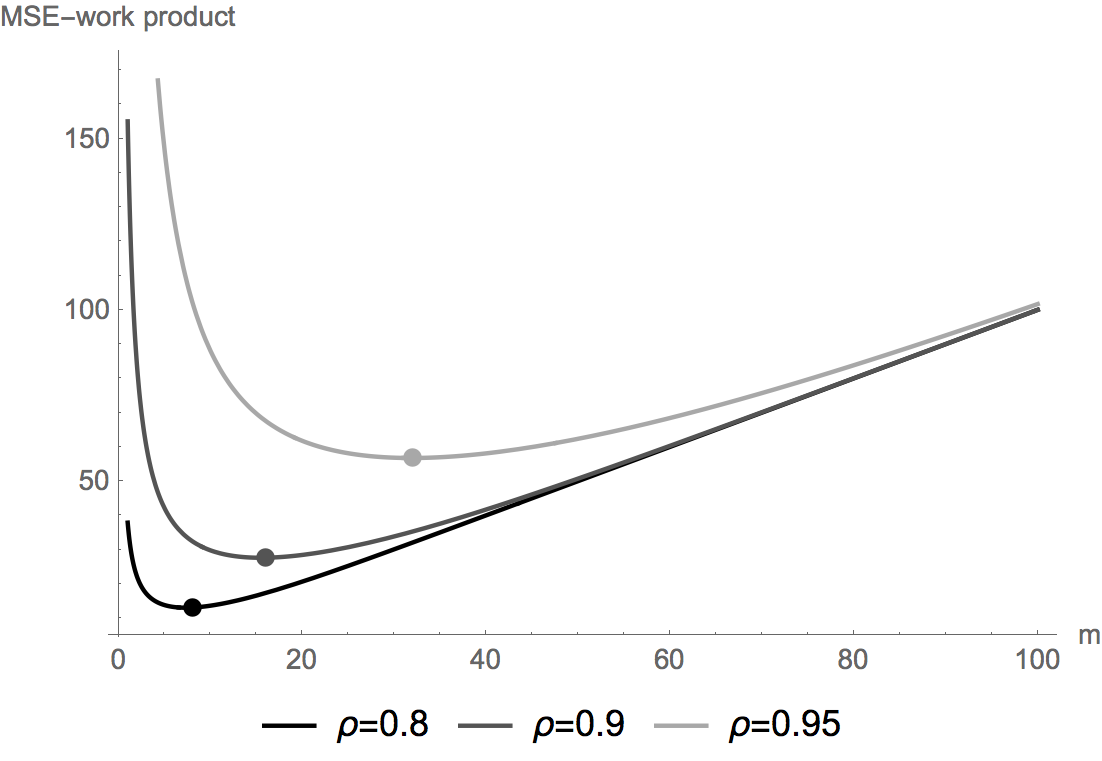}

\protect\caption{\label{fig:conNormChoiceOfM} MSE-work products for the unbiasing procedure with $a_i=m i+m$ and $\bar{F}_i$ chosen optimally, for $\rho=0.8,\,0.9$ and $0.95$, plotted against $m$, in the contracting normals example.
The markers correspond to the choice $m=\left\lceil \frac{w}{\log\rho}\right\rceil$ with $w=-1.632$. }
\end{center}\end{figure}

We next try to roughly find the optimal value of $m$ for a given $\rho$, and to do this we make the ansatz that $m$ should be of the form $\left\lceil \frac{w}{\log\rho}\right\rceil $
for $w<0$. The reason for this choice is that it at least keeps the values of $\nu_{i}$ roughly
at the same magnitude as $\rho\to1$, even though the value of $t_{i}$ increases. This choice will be justified further subsequently. Let's suppose for
the moment that $m$ is a continuous variable and we set it to $\frac{w}{\log\rho}$.
In this case we consider the ratio of the MSE-work product of the unbiasing procedure over the one of the ergodic average, given by 
\begin{align*}
rMSE\text{-work}&=\frac{\left(\rho^{-m}\sqrt{m(1-\rho^{2m})}\text{Li}_{-\frac{1}{2}}\rho^{m}\right)^{2}}{\frac{1+\rho}{1-\rho}}\\
&=\left(e^{-w} \sqrt{1-e^{2 w}} \text{Li}_{-\frac{1}{2}}\left(e^w\right) \sqrt{{w}}\right)^2\left(\frac{1-\rho}{{(1+\rho)\log(\rho)}}\right).
\end{align*}
 It is clear, that minimisation of this ratio over $w$ does not depend on $\rho$. Optimisation of the first parenthesis gives that it attains its minimum at $w=-1.632$. This choice of $w$ gives rise to the circular markers in Figure
\ref{fig:conNormChoiceOfM} which are clearly close to the optimal values of $m$ for all the plotted values of $\rho$. 

In Figure \ref{fig:conNormAdapt},  we plot again the MSE-work product of the ergodic average, the (optimised over $\bar{F}_i$) MSE-work product of the unbiased estimator for $m=\left\lceil \frac{w}{\log\rho}\right\rceil $
 and $w=-1.632$, and their ratio. 
In this case we observe that the ratio stays bounded above by $1.5$ as $\rho\rightarrow1,$ that is even as the convergence of the underlying chain deteriorates.
Notice that the oscillation of the ratio comes from the use of the ceiling function. 
\begin{figure}[H]
\begin{center}
\includegraphics[width=0.65\textwidth]{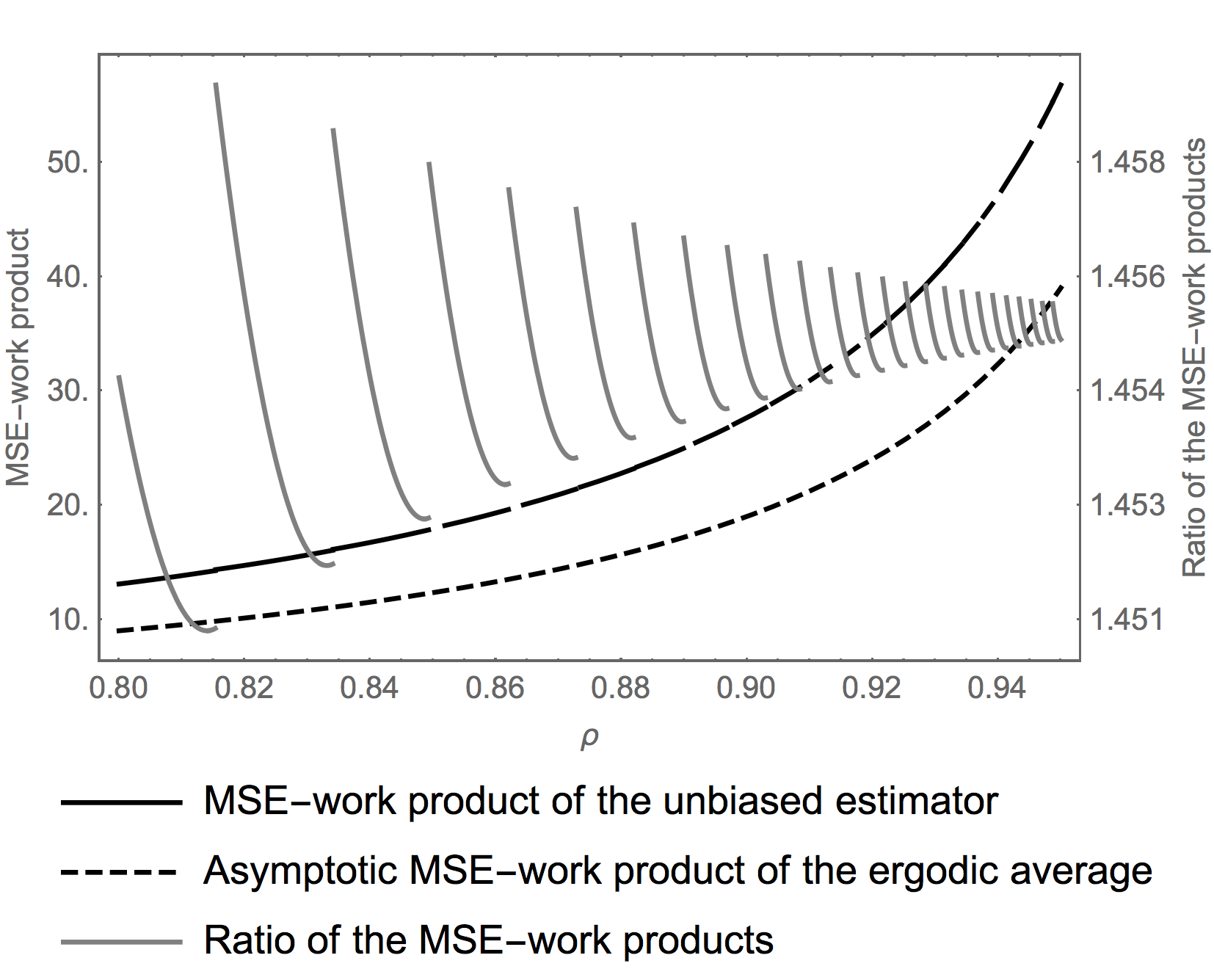}
\protect\caption{\label{fig:conNormAdapt}MSE-work products for the ergodic  average and the unbiasing procedure for $a_{i}=\left\lceil \frac{w}{\log\rho}\right\rceil i+\left\lceil \frac{w}{\log\rho}\right\rceil $ with $w=-1.632$  and $\bar{F}_i$ chosen optimally, plotted against $\rho$, in the contracting normals example. In different scale we plot the ratio of  the MSE-work product of the unbiased estimator over the MSE-work product of the ergodic average.}
\end{center}
\end{figure}
\subsubsection{Tuning}\label{sub:tuning}
At first sight it seems necessary to have a very precise knowledge of the coupling, in terms of
for example tight bounds on $\nu_{i},$ in order to tune the
unbiased estimator. In this subsection we show that if we only have good estimates $\nu_{i}\approx\hat{\nu}_{i}$
for $i\leq i_{0}$ and use a crude bound on $\nu_{i}$ for $i>i_{0}$,
then the performance of the unbiased estimator remains close to the optimal behaviour. More
precisely, instead of the optimisation problem (\ref{eq:optimisation}), we consider
\begin{eqnarray}
\min &  & \left(\sum_{i=0}^{i_{0}}\frac{\hat{\nu}_{i}}{\bar{F_{i}}}+\sum_{i=i_{0}+1}^{\infty}\frac{\nu_{i}^{\star}}{\bar{F_{i}}}\right)\left(\sum_{i=0}^{\infty}a_{i}\bar{F}_{i}\right)\label{eq:optimiseBasedEsti1}\\
\text{subject to} &  & \bar{F}_{0}=1\ge\bar{F}_{1}\ge\dots\nonumber \\
 &  & \bar{F}_{i}>0.\nonumber 
\end{eqnarray}
In order to illustrate this, we again consider the behaviour of the unbiased estimator
 with the fixed choice $a_{i}=4i+4$. We fix $\rho=0.5$
and suppose that $\nu_{i}^{\star}=\nu_{i}\left(\tilde{\rho}\right)$
for $i>i_{0}=3$ are our upper bounds on $\nu_{i}\left(\rho\right)$
for some $\tilde{\rho}\ge0.5$. Moreover, we use the exact value of $\nu_i$ for
$i\leq 3.$ We then optimise $\bar{F}_{i}^\star$ in the parametric
family $C\tilde{\rho}^{a_{i-1}}$for $i>i_{0}$ leading to the following
optimisation problem: 
\begin{eqnarray}
\min &  & \left(\sum_{i=0}^{i_{0}}\frac{\hat{\nu}_{i}}{\bar{F_{i}}}+\sum_{i=i_{0}+1}^{\hat{i}}\frac{\nu_{i}^{\star}}{\bar{F}_{i}^{\star}(C)}\right)\left(\sum_{i=0}^{i_{0}}a_{i}\bar{F}_{i}+\sum_{i=0}^{i_{0}}a_{i}\bar{F}_{i}^{\star}(C)\right)\label{eq:optimiseBasedEsti2}\\
\text{with respect to} &  & \bar{F}_{1},\dots,\bar{F}_{i_{0}},C\nonumber \\
\text{subject to} &  & \bar{F}_{0}=1\ge\bar{F}_{1}\ge\dots\ge\bar{F}_{i_{0}}\ge\bar{F}_{i}^{\star}(C)\nonumber \\
 &  & \bar{F}_{i}>0.\nonumber 
\end{eqnarray}
A numerical solution to this optimization problem using Mathematica  results in a significant improvement in the performance of the unbiased estimator as shown in Figure \ref{fig:conNormAdapt-1}. We see that having good estimates of $\nu_i$ even for just the first three levels and using crude bounds for the higher levels, greatly improves the performance of the unbiasing procedure. Naturally, as the bounds for the higher levels get worse (that is, as $\tilde{\rho}$ increases), the performance deteriorates.
\begin{figure}
\begin{centering}
\includegraphics[width=0.65\textwidth]{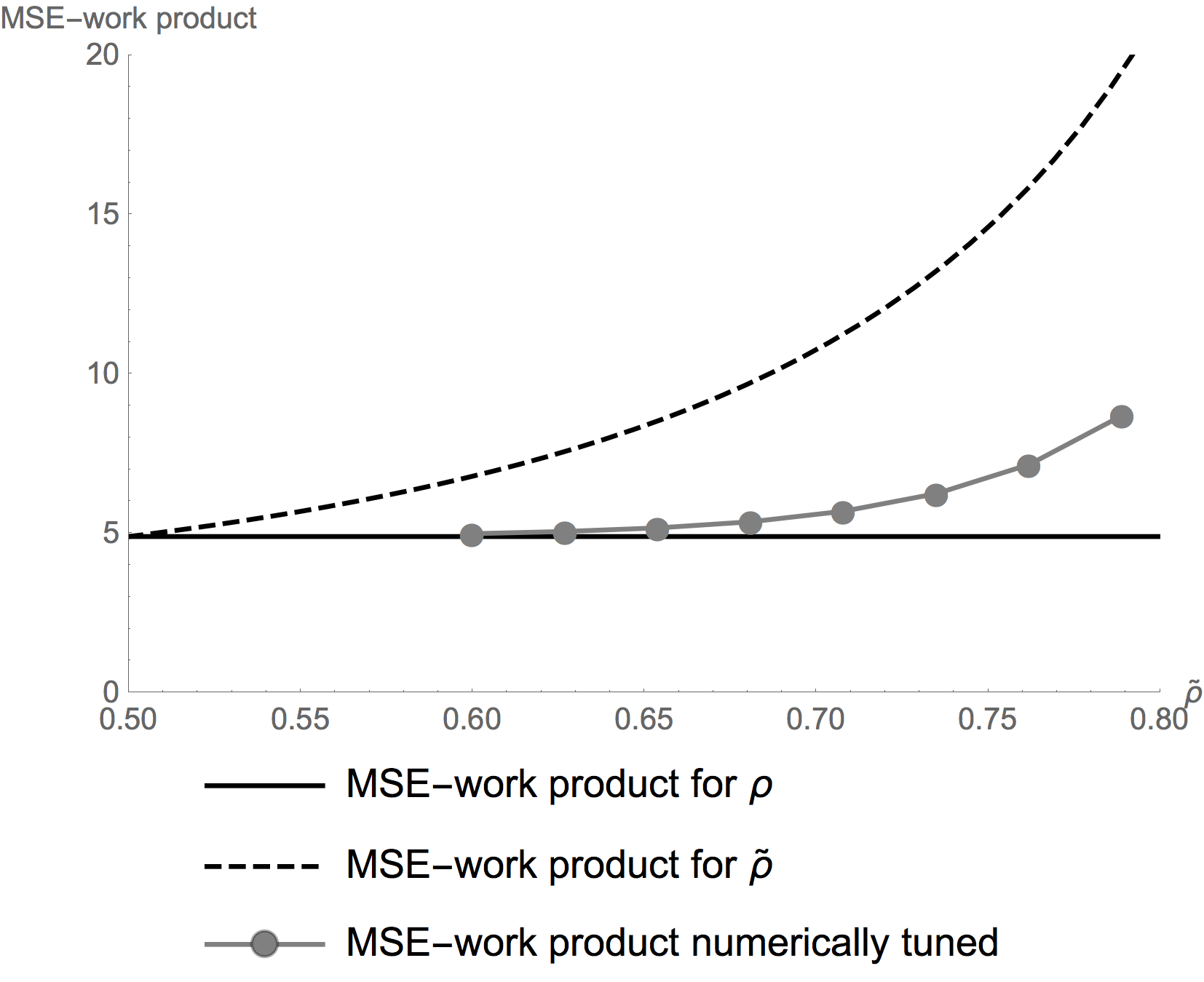} 
\par\end{centering}
\caption{\label{fig:conNormAdapt-1}
MSE-work products in the contracting normals example, plotted against $\rho$, for unbiased procedures based on knowledge of true $\nu_i$ for $\rho=0.5$ (black), only an upper bound on $\nu_i$ using $\nu_i(\tilde{\rho})$ for $\tilde{\rho}\geq0.5$ (dashed) and numerical optimisation of $\bar{F}_i$ using the exact values of $\nu_{i}$ for $i=1,2,3$ and upper bounds $\nu_i(\tilde{\rho})$ for $i>3$ (grey). More precisely, we optimise $\bar{F}_{i}$ for $i=1,\dots,3$ and $C$ in $\bar{F}_{i}^{\star}(C)=C\tilde{\rho}^{a_{i-1}}$ subject to the constraint that $\bar{F}_{3}> \bar{F}_{4}$. }\end{figure}

\subsubsection{Comparison in the parallel setting\label{sub:ContrComparisonInParallelSetting}}

We compare the ergodic average to the unbiasing procedure by measuring
CPU time. We consider $\rho=0.8$. To make the comparison fair after each step of the Markov
chain the algorithm sleeps for $1$ millisecond. In this way the generation
of $N$ has negligible effect on the comparison as it should do for
most large scale inference procedures and the computing time is determined by the distribution of $N$ and the number of steps performed. Subsequently, we describe the
procedure both for the ergodic average and the unbiasing procedure
in a 10 core parallel setting. 
\begin{enumerate}
\item For the ergodic average we draw a random number $M$ between $10$
and $10000$. Each core performs $M$ steps and we measure the time
it takes to do these steps. We average over the chains and the steps
of each chain. We plot the squared error versus the time, which gives rise to one black dot in
the left panel of Figure \ref{fig:Timed-comparison}. 
\item We draw a random time uniformly distributed on a log-scale between
$0.1$ and $10$ seconds and let each core produce unbiased estimates.
When the time is up we plot the squared error against the time giving
rise to the grey dots in the left panel Figure \ref{fig:Timed-comparison}. 
\end{enumerate}
In the right panel of Figure \ref{fig:Timed-comparison} we smooth the results of the above simulation procedure and produce 95\% confidence tubes for the MSE for the ergodic average 
(black) and the unbiased estimator (white). In this particular setting it seems that the unbiasing method
is competitive. Whereas this result is in no way conclusive, it suggests
further investigation.

\begin{figure}[H]\begin{center}
\includegraphics[width=1\textwidth]{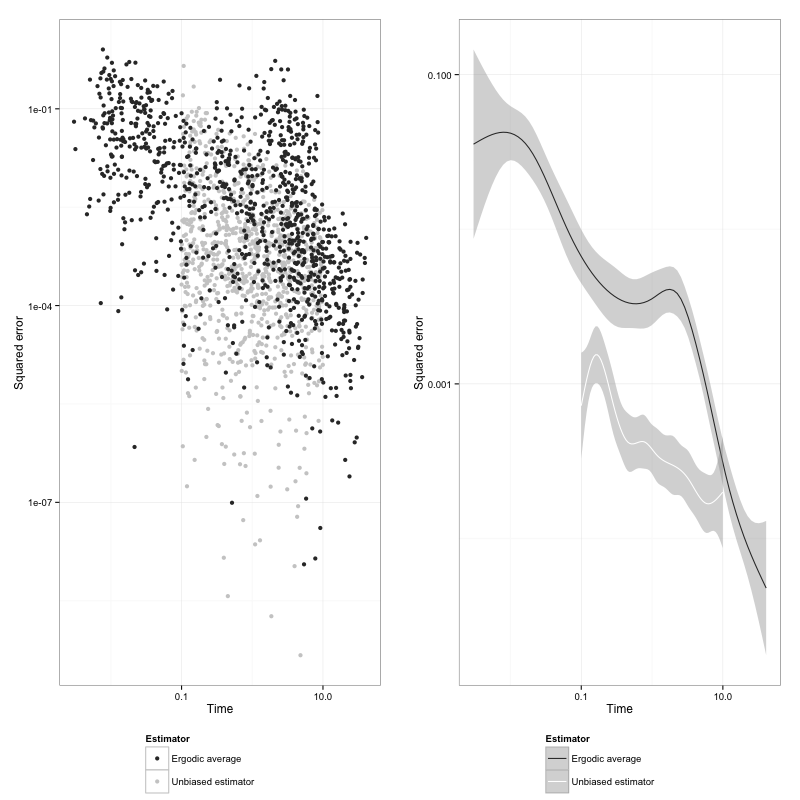} \protect\caption{\label{fig:Timed-comparison} (a) Measurements of squared errors and running times for the ergodic average and unbiasing procedure for independent runs, in a $10$-core parallel setting (left panel) and (b) $95$\% confidence tube for the MSE based on the data  on the left panel 
using a generalised additive model (right panel), in the contracting normals example with $\rho=0.8$.}
\end{center}\end{figure}

\subsection{Logistic regression\label{sub:Logistic-Regression}}
\global\long\def\logit{h}
We apply the findings of section \ref{sec:wasserstein} on unbiased
estimators based on contracting couplings to posterior inference for
a Bayesian logistic regression model. Even though we cannot verify
the assumption of section \ref{sec:wasserstein} and we cannot tune
the unbiasing procedure, we demonstrate in this section that a hands-on
application of the unbiasing procedure leads to a competitive algorithm.

We assume the data $y_{i}\in\{-1,1\}$ for $i=1,...,M$ is modelled by 
\begin{equation}
p(y_{i}\vert T_{i},\beta)=\logit(y_{i}\beta^{t}T_{i})\label{eq.logistic}
\end{equation}
where $\logit(z)=\frac{1}{1+\exp(-z)}\in[0,1]$. We put a Gaussian prior $\mathcal{N}(0,I)$
on the regression coefficient $\beta\in\R^d$ and consider a fixed design matrix $T\in\R^{M\times d}$ which we specify later on. By
Bayes' rule the posterior $\pi$ satisfies 
\[
\pi(\beta)\propto\exp\left(-\frac{1}{2}\norm{\beta}_{}^{2}\right)\prod_{i=1}^{N}\logit(y_{i}\beta^{t}T_{i}).
\]
Thus, the target measure has a density with respect to a centred Gaussian
distribution, which is such that the pCN algorithm satisfies the Assumption
\ref{assu:wass} of section \ref{sec:wasserstein} as shown in \cite{hairer2011spectral}
and \cite{durmus2014new}. We provide a brief summary of the relevant
results to the contraction of the pCN algorithm in section \ref{sec:apppCN}.

For the problem at hand the prior mean is $0$ and the posterior mean
is typically far from $0$. The proposal of the pCN algorithm only
takes into account the prior and pushes towards $0$. Furthermore, the
covariance matrix changes from prior to posterior as well. This has to be corrected
by the rejection step of the Metropolis-Hastings algorithm. The result
is that the coupling of the corresponding pCN algorithm with  the same random
input and different initial states has a contraction rate close to $1$. For this reason, the unbiasing
procedure is difficult to apply for this coupling.

A solution to this difficulty is to modify the pCN algorithm as in \cite{StuartModPCN2014}, and in particular to
consider the Metropolis-Hastings algorithm with proposal given by
\begin{equation}
\beta^\prime=c+\rho(\beta-c)+\sqrt{1-\rho^{2}}\xi,\text{ with }\xi\sim\mathcal{N}\left(0,C\right) \label{eq:andrewpcn}
\end{equation}
The resulting Markov chain preserves the non-centered Gaussian distribution $\mathcal{N}\left(c,C\right)$. Reasonable choices for $c$ and $C$ are: 
\begin{enumerate}
\item posterior mean and posterior covariance as estimated by a MCMC run;
\item Laplace approximation based on a maximum a posteriori estimator;
\item the minimiser of the Kullback-Leibler divergence
 of 
\[
D_{KL}\left(\nu,\pi\right)=\int\log\frac{d\nu}{d\pi}\frac{d\nu}{d\pi}d\pi
\]
with $\nu=\mathcal{N}\left(c,C\right)$, as suggested in \cite{StuartModPCN2014}. 
\end{enumerate}
For simplicity we take the first approach using $10^6$ steps of the random walk Metropolis (RWM) algorithm to estimate the values of the posterior mean and covariance. We consider $d=3$ and
$N=100$ data points and choose the design matrix to be 
\[
T=\left(\begin{array}{ccc}
T_{1,1} & T_{1,2} & 1\\
T_{2,1} & T_{2,2} & 1\\
\vdots & \vdots & \vdots\\
T_{100,1} & T_{100,2} & 1
\end{array}\right),
\]
for a fixed sample of $T_{i,j}\overset{\text{i.i.d.}}{\sim}\mathcal{N}\left(0,1\right)$
for $i=1,\dots100$ and $j=1,2$.

We now apply the unbiasing procedure to the coupling arising from using the same
  $\xi\sim\mathcal{N}\left(0,C\right)$ in the proposal \eqref{eq:andrewpcn}, and the same uniform random variable for the accept and reject step of the corresponding Metropolis-Hastings algorithm. The contraction
property of this coupling has not been 
established, however, we estimate the contraction factor by fitting
a line with slope $s\approx 0.75$ to the log-plot of the averaged distance, see Figure \ref{fig:logisticDecay}. This
suggests that Assumption \ref{assu:wass} is satisfied with $r=\exp\left(s\right)$.
We take a more conservative approach and set $r\coloneq\exp\left(\frac{1}{2}s\right)$.
We choose $a_{i}=mi+m$ with $m=\left\lceil \frac{-1.632}{\log r}\right\rceil $
and $\bar{F}_{i}=r^{m\cdot i}$ which closely resembles our "optimised" choice
for the contracting normals chain with $\rho=r$ in section \ref{sub:compContracting-Normals}. 

In the following we compare the MSE-work product for 
\begin{enumerate}
\item the ergodic average of the modified pCN algorithm over $10000$ steps started at $c$;
\item the average of $100$ independent realisations of the unbiased
estimator, as described in section \ref{sec:wasserstein} and for $x_0=c$.
\end{enumerate}
For both algorithms we record the squared error and the CPU time it
took to generate the estimator. Because we are using CPU-time it actually matters how many unbiased estimators we average over. This is in contrast to the idealised properties of the MSE-work error described at the beginning of section \ref{sec:contr}. This is the reason for averaging over 100 independent realisations of the unbiased estimator, rather than just taking one sample as we did in section \ref{sec:contr}.

We repeat this $10000$ times and
visualise the results using box plots in Figure \ref{fig:logisticData}\label{fig:logisticData-1}.
Notice that the distribution of the squared error for the unbiased
estimator is much more heavy-tailed compared to that of the modified
pCN algorithm. This becomes more apparent in the histograms in Figure \ref{fig:logisticDataHist}, where we can see that there exist outliers with large squared error for the unbiasing procedure.
We use this data to estimate the ratio of MSE work products to be
$4.15$ and obtain a $95\%$-confidence interval $(3.86,4.55)$ for
the ratio using the pivotal bootstrap method. 

In conclusion, we again see that the unbiasing procedure has competitive performance compared to the ergodic average, even with a crude choice of parameters and without using parallelisation.

\begin{figure}\begin{center}
\includegraphics[width=0.65\textwidth]{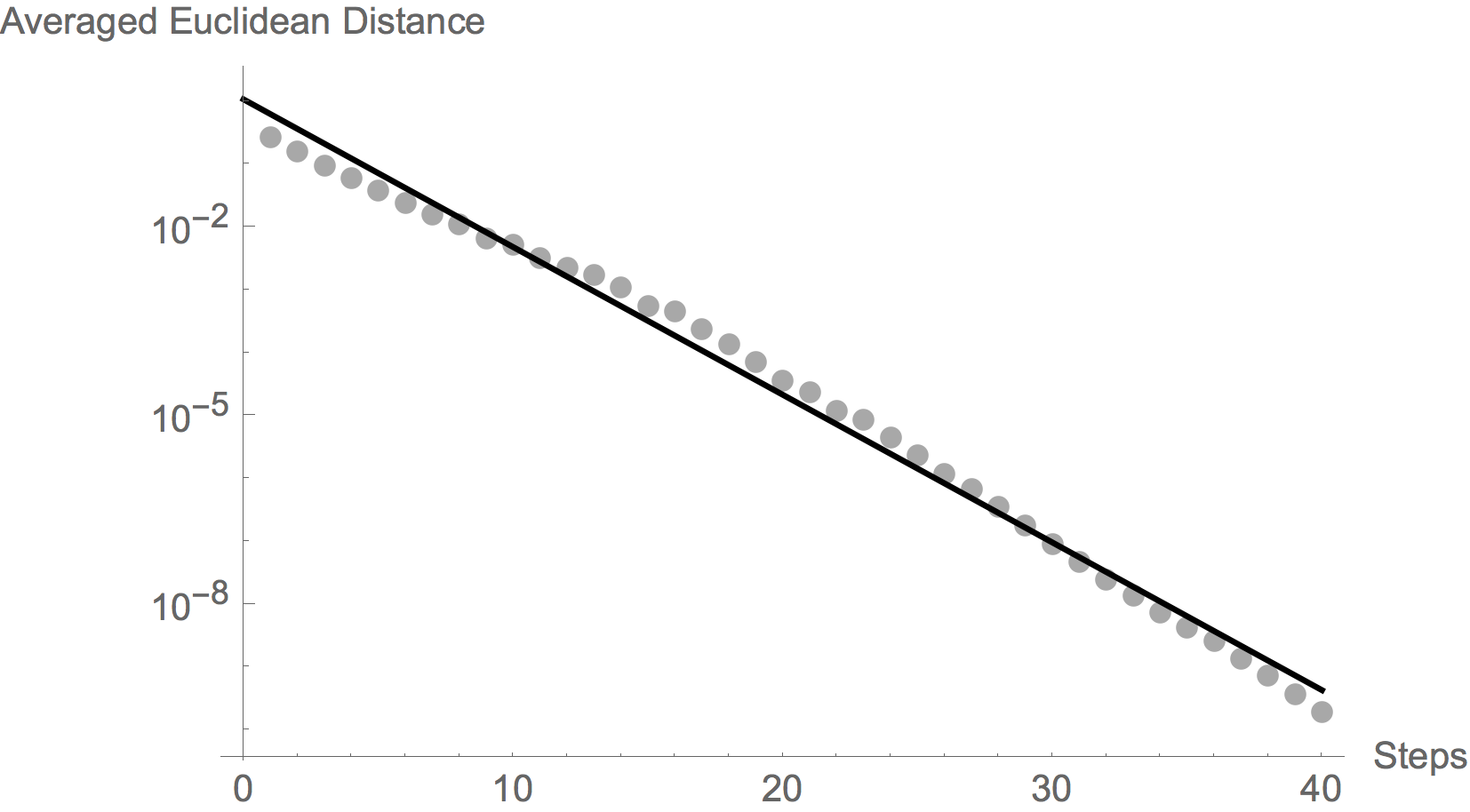} \protect\caption{\label{fig:logisticDecay}Empirical contraction property of the modified
pCN algorithm based on simulations of the coupling and averaging over independent runs, in the logistic regression example. The contraction factor is estimated using a least squares fit.}
\end{center}\end{figure}%

\begin{figure}[H]
\begin{center}
\includegraphics[width=0.65\textwidth]{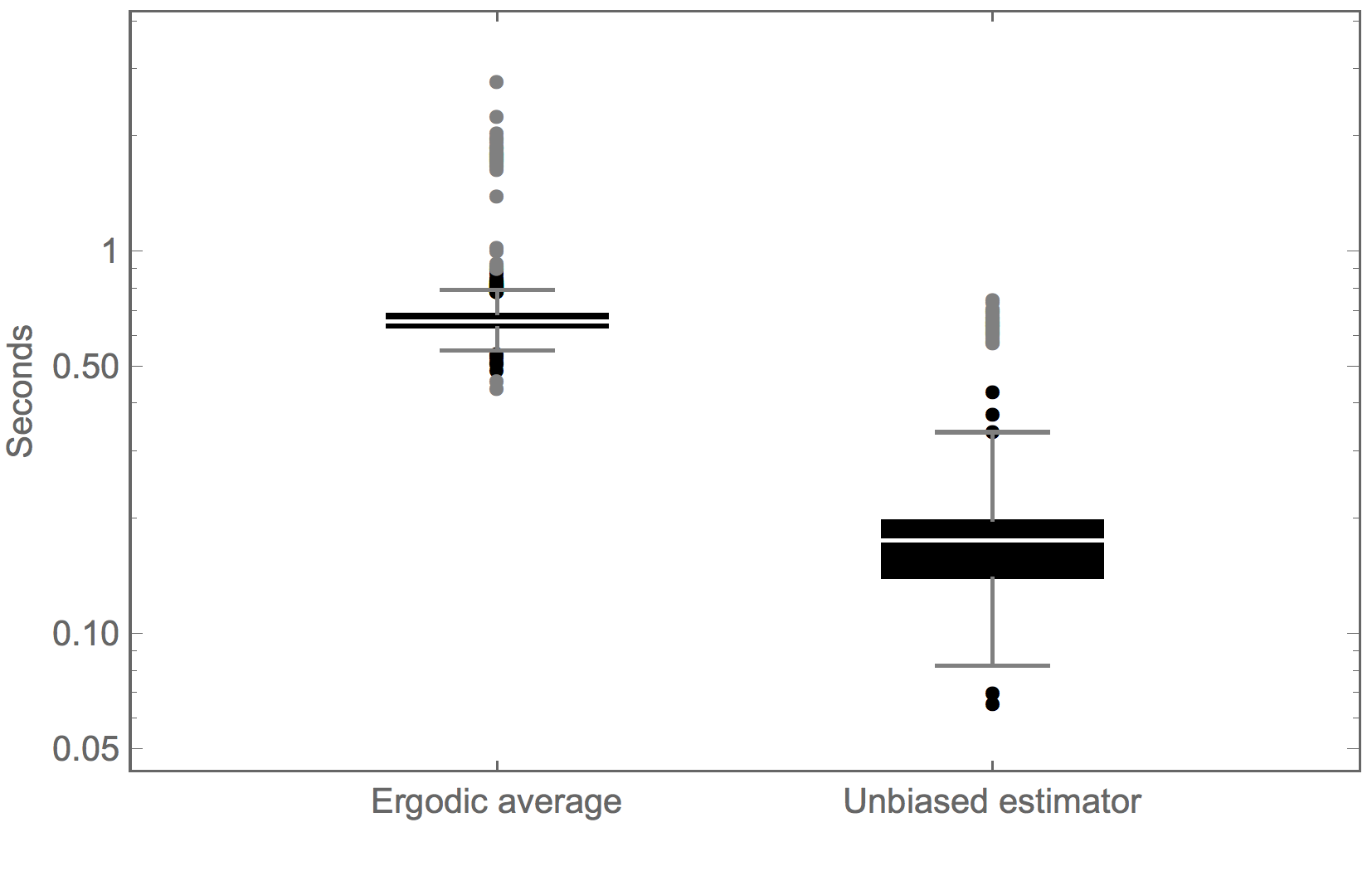}

\includegraphics[width=0.65\textwidth]{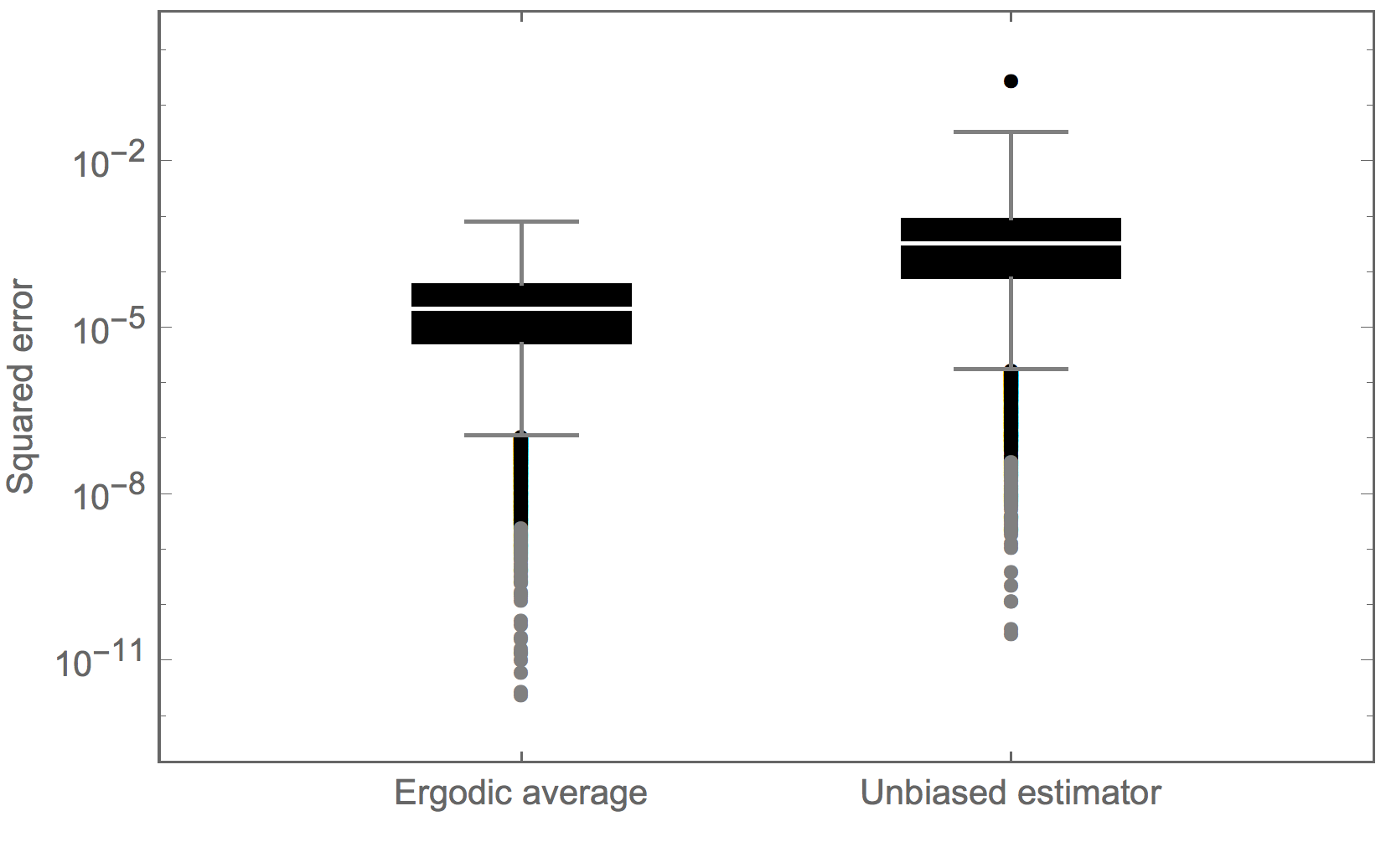}

\protect\caption{\label{fig:logisticData}Box plots for the simulation time and the squared
error for the ergodic average ($10000$ steps) and the unbiasing procedure (average over $100$ instances) of $10000$ independent simulations, in the logistic regression example.}

\end{center}
\end{figure}

\begin{figure}[H]\begin{center}
\includegraphics[width=0.65\textwidth]{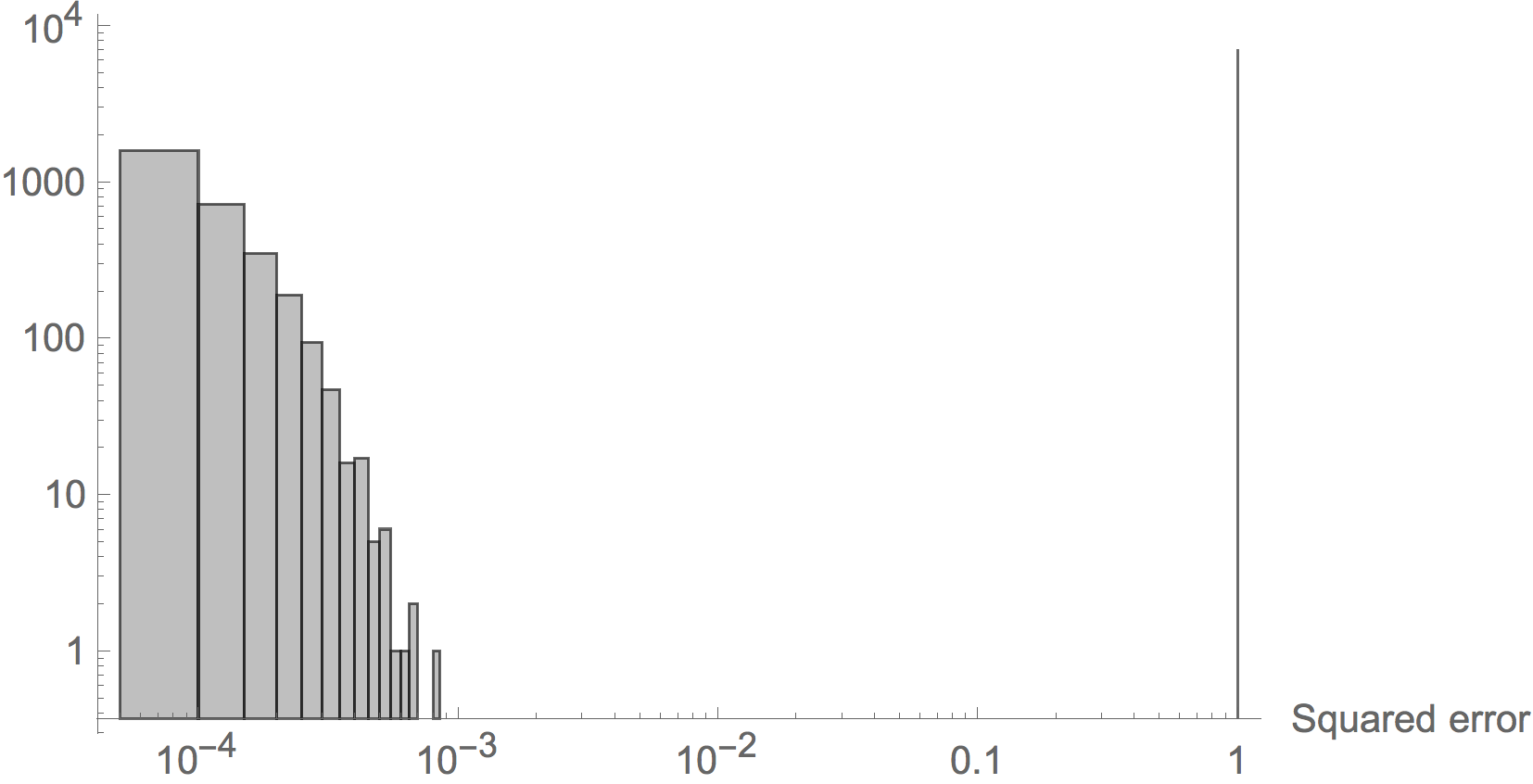}

\includegraphics[width=0.65\textwidth]{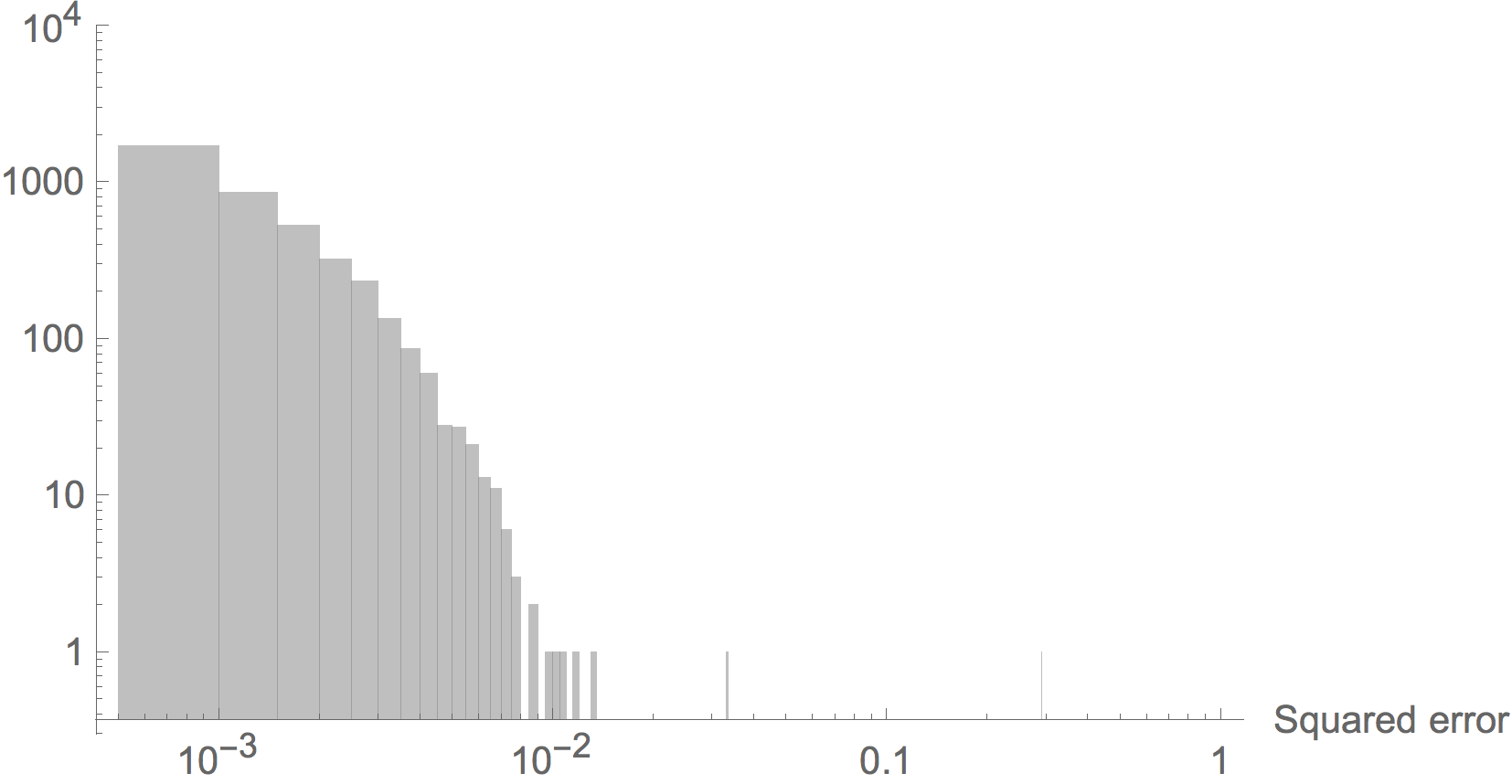}

\protect\caption{\label{fig:logisticDataHist}Histograms of the squared error for the ergodic average (top) and the unbiasing procedure (bottom), in the logistic regression example. The outliers for the unbiasing procedure illustrate the heavy tails of the distribution of the squared error.}
\end{center}

\end{figure}

%%%%%%%%%%%%%%% CONCLUSION %%%%%%%%%%%%%
%
\section{Conclusion and future directions}\label{sec:concl}
We considered unbiased estimation in intractable and/or infinite dimensional settings. In particular, we showed how to unbiasedly estimate expectations with respect to the limiting distributions of Markov chains in possibly infinite dimensional state spaces. To do this, we generalised the methodology developed in \cite{RheePHD} for removing the bias due to the burn-in time of the Markov chain, to cover the case that only a simulatable contracting coupling between runs of the chain started at different states is available (see Section \ref{sec:wasserstein}). We then used a hierarchy of coupled Markov chains in state spaces of increasing dimension, to remove the bias due to the discretisation of the infinite-dimensional state space (see Sections \ref{sec:is} and \ref{sec:trans}). 

Our focus has been on the methodological aspect, to show what it is possible to achieve, rather than to produce fully optimised results.
It is crucial for the performance of the unbiasing procedure to have good couplings between runs of the chain started at different states. There is a great body of literature on couplings which can be potentially exploited in order to on the one hand improve the results presented in the present paper and on the other hand extend the application of the unbiasing procedure to other algorithms.

Furthermore, as we demonstrated in Section \ref{sec:contr}, the tuning of the parameters appearing in the unbiasing procedure, namely the distribution $\bar{F}_i$ of the random truncation point $N$, the number of steps performed at each approximation level $a_i$ and the dimension of each approximation level $j_i$, has a huge impact on the performance. It is thus very important to develop an efficient algorithm that adapts the choice of these parameters and improves the simulation on the fly. This is particularly crucial
for the transdimensional framework, since the cost of producing samples in high dimensions rapidly increases and hence the best possible management of the available resources is crucial. 

One of the big advantages of the unbiasing methodology, is that it is very easily parallelisable. On the one hand we can use multiple cores to produce multiple copies of the unbiased estimator $Z$, while on the other hand the generation of the differences $\Delta_i$ is also readily parallelisable since we assume that they are mutually independent across different levels. Moreover, it is straightforward to manage heterogeneous computer architectures, by generating $\Delta_i$'s at low levels using slower CPU's and GPU's and reserving the faster processors for higher levels.

In the context of this paper, we provided a range of initial results on the performance of the unbiasing procedure against the ergodic average (see Section \ref{sec:contr}). It is clear from these results, which are not optimally tuned and do not make full use of parallel computing, that the unbiasing method is competitive. We are hence very much looking forward to further simulations and comparisons in problems of higher computational complexity which are closer to real-life applications.

%
%
%%%%%%%%%%%%%%%% PROOFS %%%%%%%%%%%%%%%%%%

\section{Proofs}

\label{sec:proofs} We now present the proofs of the results contained
in sections \ref{sec:Linear}-\ref{sec:trans}. 

\subsection{Proofs of the results in section \ref{sec:Linear}}
We first formulate and prove two lemmas which are necessary for the
proofs of Theorems \ref{thm:linear} and \ref{thm:linearhm} in section
\ref{sec:Linear}. Recall that $u^{i}|y^{i}$ is the truncation to
$j_{i}$ terms of the Karhunen-Loeve expansion of the Gaussian posterior
arising in the linear inverse problem setting with Gaussian priors.
Furthemore, $\tilde{{u}}^{i}|y^{i}$ is the modification of $u^{i}|y^{i}$,
in which we consider a draw from the prior in the copmlement of $\R^{j_{i}}.$ 
\begin{lem}
\label{lem1} It holds 
\begin{align}
\norm{u^{i}|y^{i}-u^{i-1}|y^{i-1}}_{2}^{2}=O(j_{i-1}^{1-2a}-j_{i}^{1-2a}).\label{eq:11}
\end{align}
In particular, 
\begin{enumerate}
\item[-] for $j_{i}=2^{i}$ we have $\norm{u^{i}|y^{i}-u^{i-1}|y^{i-1}}_{2}^{2}=\mathcal{O}(2^{i(1-2a)})$; 
\item[-]for $j_{i}=i^{q}$, for some $q\in\N$, we have $\norm{u^{i}|y^{i}-u^{i-1}|y^{i-1}}_{2}^{2}=\mathcal{O}(i^{q-1-2aq})$. 
\end{enumerate}
\end{lem}
\begin{proof}
By (\ref{eq:tudr}), we have 
\begin{align*}
\norm{u^{i}|y^{i}-u^{i-1}|y^{i-1}}_{2}^{2} & \leq2\sum_{\jj=j_{i-1}+1}^{j_{i}}\frac{\jj^{-4p}y_{\jj}^{2}}{(\jj^{2a}+\jj^{-4p})^{2}}+2\E\sum_{\jj=j_{i-1}+1}^{j_{i}}\frac{\zeta_{\jj}^{2}}{\jj^{2a}+\jj^{-4p}}\\
 & \leq c\sum_{\jj=j_{i-1}+1}^{j_{i}}\frac{\jj^{-4p}}{(\jj^{2a}+\jj^{-4p})^{2}}+c\sum_{\jj=j_{i-1}+1}^{j_{i}}\frac{1}{\jj^{2a}+\jj^{-4p}}\\
 & \leq c\sum_{\jj=j_{i-1}+1}^{j_{i}}\jj^{-4p-4a}+c\sum_{\jj=j_{i-1}+1}^{j_{i}}\jj^{-2a}\\
 & \leq c\sum_{\jj=j_{i-1}+1}^{j_{i}}\jj^{-2a}\\
 & =\mathcal{O}(j_{i-1}^{1-2a}-j_{i}^{1-2a}),
\end{align*}
where in the first inequality we used that the $\zeta_{\jj}$ are
centred, in the second inequality we used Assumption \ref{ass2} and
to get the order of the last sum we used comparison with an integral.
Here, $c$ is a positive constant which may be different from occurrence
to occurrence.

Suppose $j_{i}=2^{i}$. Then 
\[
j_{i-1}^{1-2a}-j_{i}^{1-2a}=2^{(i-1)(1-2a)}-2^{i(1-2a)}=2^{i(1-2a)}(2^{2a-1}-1)=\mathcal{O}(2^{i(1-2a)}).
\]
If $j_{i}=i^{q}$, then for $i\geq2$, there exists a constant $c>0$
which only depends on $q$ and $a$, such that 
\begin{align*}
j_{i-1}^{1-2a}-j_{i}^{1-2a} & =(i-1)^{q(1-2a)}-i^{q(1-2a)}\\
 & =i^{q(1-2a)}\left(\left(\frac{i-1}{i}\right)^{q(1-2a)}-1\right)\\
 & =i^{q(1-2a)}\left(\frac{1}{(1-\frac{1}{i})^{q(2a-1)}}-1\right)\\
 & \leq c\, i^{q(1-2a)}q(2a-1)\frac{1}{i}\\
 & =\mathcal{O}(i^{q-1-2aq}).
\end{align*}
To see this, we use the fact that $f(x)=\frac{1}{(1-x)^{r}}=1+rx+\mathcal{O}(x^{2})$
for $x$ small. \end{proof}
\begin{lem}
\label{lem2} It holds 
\begin{align}
\norm{\tilde{u}^{i}|y^{i}-\tilde{u}^{i-1}|y^{i-1}}_{2}^{2}=O(j_{i-1}^{1-4p-4a}-j_{i}^{1-4p-4a}).\label{eq:11hm}
\end{align}
In particular, 
\begin{enumerate}
\item[-] for $j_{i}=2^{i}$ we have $\norm{\tilde{u}^{i}|y^{i}-\tilde{u}^{i-1}|y^{i-1}}_{2}^{2}=\mathcal{O}(2^{i(1-4p-4a)})$; 
\item[-] for $j_{i}=i^{q}$, for some $q\in\N$, we have $\norm{\tilde{u}^{i}|y^{i}-\tilde{u}^{i-1}|y^{i-1}}_{2}^{2}=\mathcal{O}(i^{q-1-(4p+4a)q})$. 
\end{enumerate}
\end{lem}
\begin{proof}
We first observe that since for $x\leq1$ it holds that $(1+x)^{\frac{1}{2}}\geq1+(\sqrt{2}-1)x$,
we have that for $\jj\geq1$ 
\begin{align*}
1-\jj^{-a}(\jj^{2a}+\jj^{-4p})^{\frac{1}{2}} & =1-(1+\jj^{-4p-2a})^{\frac{1}{2}}\\
 & \leq1-1-(\sqrt{2}-1)\jj^{-4p-2a}\\
 & =(1-\sqrt{2})\jj^{-4p-2a}.
\end{align*}
Combining this with (\ref{eq:hmdr}), we have 
\begin{align*}
\norm{\tilde{u}^{i}|y^{i}-\tilde{u}^{i-1}|y^{i-1}}_{2}^{2} & \leq2\sum_{\jj=j_{i-1}+1}^{j_{i}}\frac{\jj^{-4p}y_{\jj}^{2}}{(\jj^{2a}+\jj^{-4p})^{2}}\\ &\quad+2\E\sum_{\jj=j_{i-1}+1}^{j_{i}}\frac{(1-\jj^{-a}(\jj^{2a}+\jj^{-2p})^{\frac{1}{2}})^{2}\zeta_{\jj}^{2}}{\jj^{2a}+\jj^{-4p}}\\
 & \leq c\sum_{\jj=j_{i-1}+1}^{j_{i}}\frac{\jj^{-4p}}{(\jj^{2a}+\jj^{-4p})^{2}}+c\sum_{\jj=j_{i-1}+1}^{j_{i}}\frac{\jj^{-4a-8p}}{\jj^{2a}+\jj^{-4p}}\\
 & \leq c\sum_{\jj=j_{i-1}+1}^{j_{i}}\jj^{-4p-4a}+c\sum_{\jj=j_{i-1}+1}^{j_{i}}\jj^{-8p-6a}\\
 & \leq c\sum_{\jj=j_{i-1}+1}^{j_{i}}\jj^{-4p-4a}\\
 & =\mathcal{O}(j_{i-1}^{1-4p-4a}-j_{i}^{1-4p-4a}),
\end{align*}
where in the first inequality we used that the $\zeta_{\jj}$ are
centred, in the second inequality we used Assumption \ref{ass2} and
to get the order of the last sum we used comparison with an integral.
Here $c$ is a positive constant, which may change from occurrence
to occurrence. The rest of the proof proceeds exactly as the proof
of Lemma \ref{lem1}. 
\end{proof}

\begin{proof}[Proof of Theorem \ref{thm:linear}]
We first use Proposition \ref{prop:rhee} which gives conditions securing
the unbiasedness and finite variance of $Z$ and then check that these
conditions are compatible with a finite expected computing time. By
the $s$-H\"older continuity of $f$, we have that 
\begin{equation}
\norm{\Delta_{i}}_{2}^{2}\leq\norm f_{s}^{2}\norm{u^{i}|y^{i}-u^{i-1}|y^{i-1}}_{2}^{2s}\leq\norm f_{s}^{2}(\norm{u^{i}|y^{i}-u^{i-1}|y^{i-1}}_{2}^{2})^{s},\label{eq:99}
\end{equation}
where for the last inequality we use Jensen's inequality for concave
functions if $s\in(0,1)$, otherwise if $s=1$ we have a trivial identity.
We will use this bound to verify the condition of Proposition \ref{prop:rhee}
for specific parameter choices. Regarding the expected computing time
$\E({\tau})$, under Assumption \ref{ass:lincost}, we have 
\begin{equation}
\E(\tau)=\E\sum_{i=0}^{N}t_{i}=\sum_{i=0}^{\infty}t_{i}\rp(N\geq i)=\sum_{i=0}^{\infty}j_{i}\rp(N\geq i).\label{eq:9}
\end{equation}
Below we use $c$ as a generic positive constant which may change
from occurrence to occurrence.

Suppose $j_{i}=2^{i}.$ Then by Lemma \ref{lem1} and the bound in
(\ref{eq:99}), we get that $\norm{\Delta_{i}}_{2}^{2}=\mathcal{{O}}(2^{is(1-2a)})$.
Letting $r=2^{\frac{{s(1-2a)}}{2}}$ and noticing that $r<1$ since
$a>\frac{1}{2}$, we get 
\[
\sum_{i\leq l}\frac{{\norm{\Delta_{i}}_{2}\norm{\Delta_{l}}_{2}}}{\rp(N\geq i)}\leq c\sum_{i=0}^{\infty}\frac{r^{i}}{\rp(N\geq i)}\sum_{l=i}^{\infty}r^{l}=c\frac{{1}}{1-r}\sum_{i=0}^{\infty}\frac{{r^{2i}}}{\rp(N\geq i)}.
\]
The right hand side is finite, and hence by Proposition \ref{prop:rhee}
the estimator $Z$ has finite variance, if we choose for example $\rp(N\geq i)\propto r^{(2-\epsilon)i},$
for $\epsilon>0$ which can be arbitrarily small. Under this choice,
we have 
\[
\E(\tau)=c\sum_{i=0}^{\infty}2^{i(1+(2-\epsilon)\frac{{s(1-2a)}}{2})}.
\]
In order to have that the right hand side is finite, we need to choose
$0<\epsilon<\frac{2+2s-4as}{s(1-2a)}$ which is possible since $a>\frac{1+s}{2s}$.
%this we just check by showing that the restiction on epsilon for 1+(2-epsilon)s(1-2a)2<0 is a valid

Alternatively, let $j_{i}\lesssim i^{q}$, $q>0$. Then again by Lemma
\ref{lem1} and (\ref{eq:99}), we have that $\norm{\Delta_{i}}_{2}^{2}=\mathcal{{O}}(i^{s(q-1-2aq)})$.
Letting $r=\frac{s(q-1-2aq)}{2}$, we get 
\[
\sum_{i\leq l}\frac{{\norm{\Delta_{i}}_{2}\norm{\Delta_{l}}_{2}}}{\rp(N\geq i)}\leq c\sum_{i\leq l}\frac{i^{r}l^{r}}{\rp(N\geq i)}\leq c\sum_{i=0}^{\infty}\frac{i^{r}}{\rp(N\geq i)}\sum_{l=i}^{\infty}j^{r}.
\]
For the second sum on the right hand side to be finite we need $r<-1,$
in which case comparison with an integral shows that the sum is of
order $i^{r+1}$. Then we have a bound of the form 
\[
\sum_{i\leq l}\frac{{\norm{\Delta_{i}}_{2}\norm{\Delta_{l}}_{2}}}{\rp(N\geq i)}\leq c\sum_{i\leq l}\frac{i^{2r+1}}{\rp(N\geq i)},
\]
where the right hand side is finite if we choose for example $\rp(N\geq i)\propto i^{2r+2+\epsilon}$
for $\epsilon>0$ which can be arbitrarily small. In order to have
$r<-1$, since $a>\frac{1}{2},$ we need to choose $q>\frac{s-2}{s(1-2a)}\coloneq q_{1}$.
Under this choice, we have 
\[
\E(\tau)=c\sum_{i=0}^{\infty}i^{q+2+\epsilon+s(q-1-2aq)}.
\]
For the last sum to be finite, we need to choose $0<\epsilon<s-3-q(1+s-2as)$
which is only possible if $q>\frac{s-3}{1+s-2as}\coloneq q_{2}$. The proof
is complete after noticing that $\max\{q_{1},q_{2}\}=q_{2}$. 
\end{proof}

\begin{proof}[Proof of Theorem \ref{thm:linearhm}]
The proof follows the lines of the proof of Theorem \ref{thm:linear}.
Since $f$ is linear, we have that 
\[
\norm{\Delta_{i}}_{2}^{2}\leq\norm f_{\mathcal{L}(\state,\R)}^{2}\norm{u^{i}|y^{i}-u^{i-1}|y^{i-1}}_{2}^{2},
\]
and we will use this bound to verify the condition of Proposition
\ref{prop:rhee} for specific parameter choices. Under Assumption
\ref{ass:lincost} the expected computing time $\E(\tau)$ is given
by (\ref{eq:9}). Below $c$ is a positive constant which may change
from occurrence to occurrence.

Let $j_{i}=2^{i}.$ Then by Lemma \ref{lem2}, we have that $\norm{\Delta_{i}}_{2}^{2}=\mathcal{{O}}(2^{i(1-4p-4a)})$.
Letting $r=2^{\frac{{1-4p-4a}}{2}}$ and noticing that $r<1$ since
$a>\frac{1}{2}$ and $p>0$, we get 
\[
\sum_{i\leq l}\frac{{\norm{\Delta_{i}}_{2}\norm{\Delta_{l}}_{2}}}{\rp(N\geq i)}\leq c\sum_{i=0}^{\infty}\frac{r^{i}}{\rp(N\geq i)}\sum_{l=i}^{\infty}r^{l}=c\frac{{1}}{1-r}\sum_{i=0}^{\infty}\frac{{r^{2i}}}{\rp(N\geq i)}.
\]
The right hand side is finite, and hence by Proposition \ref{prop:rhee}
the estimator $Z$ has finite variance, if we choose for example $\rp(N\geq i)\propto r^{(2-\epsilon)i},$
for $\epsilon>0$ which can be arbitrarily small. Under this choice,
we have 
\[
\E(\tau)=c\sum_{i=0}^{\infty}2^{i(1+(2-\epsilon)\frac{{1-4p-4a}}{2})}.
\]
In order to have that the right hand side is finite, we need to choose
$0<\epsilon<\frac{4-8p-4a}{1-4p-4a}$ which is possible since $a>\frac{1}{2}$.

Alternatively, let $j_{i}\lesssim i^{q},$ $q>0$. Then again by Lemma
\ref{lem2}, we have that $\norm{\Delta_{i}}_{2}^{2}=\mathcal{{O}}(i^{q-1-(4p+4a)q})$.
Letting $r=\frac{q-1-(4p+4a)q}{2}$, we get 
\[
\sum_{i\leq l}\frac{{\norm{\Delta_{i}}_{2}\norm{\Delta_{l}}_{2}}}{\rp(N\geq i)}\leq c\sum_{i\leq l}\frac{i^{r}l^{r}}{\rp(N\geq i)}\leq c\sum_{i=0}^{\infty}\frac{i^{r}}{\rp(N\geq i)}\sum_{l=i}^{\infty}l^{r}.
\]
For the second sum on the right hand side to be finite we need $r<-1,$
in which case comparison with an integral shows that the sum is of
order $i^{r+1}$. Then we have a bound of the form 
\[
\sum_{i\leq l}\frac{{\norm{\Delta_{i}}_{2}\norm{\Delta_{l}}_{2}}}{\rp(N\geq i)}\leq c\sum_{i\leq l}\frac{i^{2r+1}}{\rp(N\geq i)},
\]
where the right hand side is finite if we choose for example $\rp(N\geq i)\propto i^{2r+2+\epsilon}$
for $\epsilon>0$ which can be arbitrarily small. In order to have
$r<-1$, since $a>\frac{1}{2},$ we need to choose $q>\frac{-1}{1-4p-4a}\coloneq q_{3}$.
Under this choice, we have 
\[
\E(\tau)=c\sum_{i=0}^{\infty}i^{q+2+\epsilon+q-1-(4p+4a)q}.
\]
For the last sum to be finite, we need to choose $0<\epsilon<-q(2-4p-4a)-2$,
which is only possible if $q>\frac{-2}{2-4p-4a}\coloneq q_{4}$. The proof
is complete after noticing that $\max\{q_{3},q_{4}\}=q_{4}$.
\end{proof}

\subsection{Proofs of the results in section \ref{sec:wasserstein}}
\begin{proof}[Proof of Theorem \ref{thm:wasserstein}]
We first use Proposition \ref{prop:rhee} which gives conditions
securing unbiasedness and finite variance of $Z$ and then make sure
that these conditions are compatible with a finite expected computing
time. Let $\mathcal{F}_{k}=\sigma\left(\left\{ \ux{\ell}{i},\lx{\ell}{i}\,\vert\,\ell\leq k\right\} \right)$.
We bound

\begin{eqnarray}
\norm{\Delta_{i}}_{2}^{2} & \leq & \left\Vert f\right\Vert _{s}^{2}\E d^{2s}\left(\ux{0}{i},\lx{0}{i}\right)\nonumber \\
 & \leq & \left\Vert f\right\Vert _{s}^{2}\E\E\left(d^{2s}\left(\ux{0}{i},\lx{0}{i}\right)\vert\mathcal{F}_{-a_{i-1}}\right)\nonumber \\
 & \leq & \left\Vert f\right\Vert _{s}^{2}\E\left(K^{a_{i-1}}d^{2s}(\ux{-a_{i-1}}{i},x_{0})\right)\nonumber \\
 & \leq & c\left\Vert f\right\Vert _{s}^{2}r^{a_{i-1}}\E d^{2s}(\ux{-a_{i-1}}{i},x_{0})\nonumber \\
 & \leq & cr^{a_{i-1}},\label{eq:contr}
\end{eqnarray}
where the last step follows from (\ref{eq:contractionBoundedness})
and where we use $c$ as a positive constant which maybe different
from occurrence to occurrence.

By the considerations at the end of subsection \ref{ssec:implback},
it suffices to verify (\ref{eq:varest}) to get the unbiasedness and
finite variance of $Z$. Using (\ref{eq:contr}), we have 
\begin{eqnarray}
\sum_{i\le l}\frac{\norm{\Delta_{i}}_{2}\norm{\Delta_{l}}_{2}}{\rp(N\ge i)} & = & \sum_{i=0}^{\infty}\frac{\norm{\Delta_{i}}_{2}}{\rp(N\ge i)}\sum_{l=i}^{\infty}\norm{\Delta_{l}}_{2}\leq c\sum_{i=0}^{\infty}\frac{\norm{\Delta_{i}}_{2}}{\rp(N\ge i)}\frac{r^{\frac{1}{2}i}}{1-r^{\frac{1}{2}}}\nonumber \\
 & \leq & c\sum_{i=0}^{\infty}\frac{r^{i}}{\rp(N\ge i)}.\label{eq:wassersteinVarianceBound}
\end{eqnarray}
where we used that $r<1$ and $a_{i}\ge i$. It is hence sufficient
to choose the distribution of $N$ such that $\sum_{i}\frac{r^{i}}{\rp(N\ge i)}<\infty$
in order to have finite variance of the estimator $Z$; a valid choice
is for example $\rp(N\ge i)\propto r^{(1-\epsilon)i}$ for $\epsilon>0$
which can be arbitrarily small.

Regarding the expected computing time of $Z$, we have that it is
equal to $\sum_{i=0}^{\infty}t_{i}\rp(N\geq i)$, where $t_{i}$ is
the expected time to generate $\Delta_{i}$. By Assumption \ref{assu:wasserCost}
we have a mild condition on the growth of $a_{i}.$ For example $a_{i}\lesssim r^{(2\epsilon-1)i}$
works provided $\epsilon<\frac{1}{2}$ so that we have $a_{i}\geq i$
as required.
\end{proof}

\begin{proof}[Proof of Theorem \ref{thm:wassersteinsubg}]
The proof is very similar to the proof of Theorem \ref{thm:wasserstein},
where the estimate (\ref{eq:contr}) is replaced by 
\begin{equation}
\norm{\Delta_{i}}_{2}^{2}\leq ca_{i-1}^{-2r}.\label{eq:contrw}
\end{equation}
For example choose $a_{i}=i^{k},$. Then we have that 
\[
\sum_{i\le l}\frac{\norm{\Delta_{i}}_{2}\norm{\Delta_{l}}_{2}}{\rp(N\ge i)}\leq c\sum_{i=0}^{\infty}\frac{i^{-2rk+1}}{\rp(N\ge i)},
\]
where $c$ a positive constant which may change between different
occurrences. The right hand side is finite, if for example we choose
$\rp(N\geq i)\propto i^{-2rk+2+\epsilon}$ for $\epsilon>0$ which
can be arbitrarily small. Under this choice, the expected computing
time 
\[
\E(\tau)=\sum_{i=0}^{\infty}i^{k}\rp\left(N\ge i\right)=c\sum_{i=0}^{\infty}i^{k-2rk+2+\epsilon}.
\]
For the last sum to be finite, we need to choose $0<\epsilon<-3-(1-2r)k$
and this choice is possible only if $k>\frac{3}{2s-1}$. 
\end{proof}

\subsection{Proofs of the results in section \ref{sec:is}}
We first state and prove a crucial lemma which establishes that the $\Delta_i$ decay sufficiently quickly for the unbiasing programme to work. Then we provide the proof of Theorem \ref{thm:ExistenceIndep}.
\begin{lem}
\label{lem:ch5lem}Suppose Assumptions \ref{ass:indep} and \ref{ass:indepObservable}
are satisfied for some $\beta,\kappa>1$ respectively. Then for $a_{i}\sim\frac{\beta}{c}_{\star}\log(j_{i})$,
where $c_{\star}=-\log(1-\alpha_{\star})$, we have 
\[
\norm{\Delta_{i}}_{2}^{2}\lesssim j_{i-1}^{-(\beta\wedge\kappa)}.
\]
\end{lem}
\begin{proof}
In order to bound $E_{1}$, we need to be able to control the different
behaviour of the independence sampler in dimension $j_{i-1}$ and
$j_{i}$ if driven by the same underlying randomness $W$. For this
reason we introduce the random functions $\ub^{i},\lb^{i}$ below,
taking values in $\{1,2,3\}$,
\begin{eqnarray*}
\ub^{i}\left(x,W\right) & = & 1\cdot\1_{[0,\alpha_{\star}]}(U_{1}^{i})+2\cdot\1_{(\alpha_{\star},1]}(U_{1}^{i})\1_{[0,\frac{\alpha_{j_{i}}(x,\xi_{2}^{i})-\alpha_{\star}}{1-\alpha_{\star}}]}(U_{2}^{i})\\
 &  & +3\cdot\1_{(\alpha_{\star},1]}(U_{1}^{i})\1_{(\frac{\alpha_{j_{i}}(x,\xi_{2}^{i})-\alpha_{\star}}{1-\alpha_{\star}},1]}(U_{2}^{i}),\\
\lb^{i}(x,W) & = & 1\cdot\1_{[0,\alpha_{\star}]}(U_{1}^{i})+2\cdot\1_{(\alpha_{\star},1]}(U_{1}^{i})\1_{[0,\frac{\alpha_{j_{i-1}}(x,\Pi_{j_{i-1}}\xi_{2}^{i})-\alpha_{\star}}{1-\alpha_{\star}}]}(U_{2}^{i})+\\
 &  & 3\cdot\1_{(\alpha_{\star},1]}(U_{1}^{i})\1_{(\frac{\alpha_{j_{i-1}}(x,\Pi_{j_{i-1}}\xi_{2}^{i})-\alpha_{\star}}{1-\alpha_{\star}},1]}(U_{2}^{i}),
\end{eqnarray*}
such that $\ub^{i}\left(\ux{k-1}{i},W_{k}^{i}\right)$ and $\lb^{i}\left(\lx{k-1}{i},W_{k}^{i}\right)$
denotes the branch of the random functions $\varphi_{\ux{}{}}^{i}$
and $\varphi_{\lx{}{}}^{i}$ that was taken to go from $\ux{k-1}{i}$
to $\ux{k}{i}$ and from $\lx{k-1}{i}$ to $\lx{k}{i}$, respectively.
More precisely, 
\begin{enumerate}
\item if $\ub^{i}\left(\ux{k-1}{i},W_{k}^{i}\right)=1$ then $\ux{k}{i}=\xi_{1,k}^{i}$; 
\item if $\ub^{i}\left(\ux{k-1}{i},W_{k}^{i}\right)=2$ then $\ux{k}{i}=\xi_{2,k}^{i}$; 
\item if $\ub^{i}\left(\ux{k-1}{i},W_{k}^{i}\right)=3$ then $\ux{k}{i}=\ux{k-1}{i}$. 
\end{enumerate}
and the analogous statement for $\lb^{i}\left(\lx{k-1}{i},W_{k}^{i}\right)$
and $\lx{k}{i}$. For economy of notation we define $\ubk{k}^{i}\coloneq\ub^{i}\left(\ux{k-1}{i},W_{k}^{i}\right)$
and similarly for $\lbk{k}^{i}$. Note that $\ubk{k}^{i}=1$ if and
only if $\lbk{k}^{i}=1$ such that this leads to synchronisation because
in this case $\lx{k}{i}=\Pi_{j_{i-1}}\xi_{1,k}^{i}=\Pi_{j_{i-1}}\ux{k}{i}$.
Notice that this synchronisation property is preserved to $l>k$ as
long as $\lbk{\tilde{k}}^{i}=\ubk{\tilde{k}}^{i}$ for $\tilde{k}=k+1,\dots,l.$

This notation allows us to bound $E_{1}$ as follows:
\begin{eqnarray*}
E_{1} & \leq & \norm f_{\infty}^{2}\rp\left(\neg\left\{ \exists k\leq0\;\text{s.t.}\;\ubk{k}^{i}=0\text{ and }\forall l>k\:\ubk{l}^{i}=\lbk{l}^{i}\right\} \right),
\end{eqnarray*}
because $f(\Pi_{j_{i-1}}\ux{0}{i})-f(\lx{0}{i})=0$ on the event $\left\{ \exists k\leq0\,\ubk{k}^{i}=0\text{ and }\forall l>k\:\ubk{l}^{i}=\lbk{l}^{i}\right\} .$
In order to bound the above we introduce the filtration $\mathcal{F}_{k}=\sigma\left(\{W_{m}^{i}\}_{m\leq k}\right)$
and notice that $\tau=\inf\left\{ k:\ubk{k}^{i}=0\right\} $ is a
stopping time with respect to $\{\mathcal{F}_{k}\}$. We have 
\begin{eqnarray*}
E_{1} & \leq & \norm f_{\infty}^{2}\rp\left(\left\{ \tau>0\right\} \cup\left\{ \tau\leq0\text{ and }\exists0\ge l>\tau:\:\ubk{l}^{i}\ne\lbk{l}^{i}\right\} \right).\\
 & \leq & \norm f_{\infty}^{2}\left(\rp\left(\left\{ \tau>0\right\} \right)+\rp\left(\left\{ \tau\leq0\text{ and }\exists0\ge l>\tau:\:\ubk{l}^{i}\ne\lbk{l}^{i}\right\} \right)\right)\\
 & \leq & \norm f_{\infty}^{2}\left((1-\alpha_{\star})^{a_{i-1}}+\rp\left(\left\{ \tau\leq0\text{ and }\exists0\ge l>\tau:\:\ubk{l}^{i}\ne\lbk{l}^{i}\right\} \right)\right).
\end{eqnarray*}
Notice that for $\isEvt=\left\{ \tau\leq0\text{ and }\exists0\geq l>\tau:\ubk{l}^{i}\ne\lbk{l}^{i}\right\} $,
we have 
\begin{align*}
\rp\left(\isEvt\right)= & \E \left(\1_{\tau\leq0}\1_{\left\{ \exists0\ge l>\tau\:\ubk{l}^{i}\ne\lbk{l}^{i}\right\} }\right)\\
= & \E \left(\1_{\tau\leq0}\E \left(\1_{\left\{ \exists0\ge l>\tau\:\ubk{l}^{i}\ne\lbk{l}^{i}\right\} }\mid\tau\right)\right)\\
= & \E\left(\1_{\tau\leq0}\rp_{\left(\ux{\tau}{i},\lx{\tau}{i}\right)}\left(\exists-\tau\ge l>0:\ubk{l}^{i}\ne\lbk{l}^{i}\right)\right)\\
\leq & \rp\left(\tau\leq0\right)\sup_{\begin{array}{c}
\Pi_{j}\left(\ux{0}{i}\right)=\lx{0}{i}\end{array}}\rp_{\left(\ux{0}{i},\lx{0}{i}\right)}\left(\exists-\tau\ge l>0:\ubk{l}^{i}\ne\lbk{l}^{i}\right),
\end{align*}
where in the third identity we made use of the strong Markov property.
In order to bound the supremum on the right hand side above, we introduce
\[
\delta_{i}\coloneq\sup_{\begin{array}{c}
\Pi_{j}\left(x_{\ux{}{}}\right)=x_{\lx{}{}}\end{array}}\rp\left(b_{\ux{}{i}}\left(x_{\ux{}{}},W\right)\neq b_{\lx{}{i}}\left(x_{\lx{}{}},W\right)\right),
\]
where $x_{\ux{}{}}$ and $x_{\lx{}{}}$ live in $\mathcal{\state}_{j_{i}}$
and $\mathcal{\state}_{j_{i-1}},$ respectively. We have
\begin{align*}
\delta_{i}\leq & \sup_{\begin{array}{c}
\Pi_{j}\left(x_{\ux{}{}}\right)=x_{\lx{}{}}\end{array}}\E\left|\alpha_{j_{i}}(x_{\ux{}{}},\xi_{2}^{i})-\alpha_{j_{i-1}}(x_{\lx{}{}},\Pi_{j_{i-1}}\xi_{2}^{i})\right|\\
= & \sup_{\begin{array}{c}
\Pi_{j}\left(x_{\ux{}{}}\right)=x_{\lx{}{}}\end{array}}\E\left|1\wedge\exp\left(\frac{1}{2}\left\Vert y-G_{j_{i}}\left(x_{\ux{}{}}\right)\right\Vert _{\R^{d}}^{2}-\frac{1}{2}\left\Vert y-G_{j_{i}}\left(\xi_{2}^{i}\right)\right\Vert _{\R^{d}}^{2}\right)\right.\\
 & -\left.1\wedge\exp\left(\frac{1}{2}\left\Vert y-G_{j_{i-1}}\left(x_{\lx{}{}}\right)\right\Vert _{\R^{d}}^{2}-\frac{1}{2}\left\Vert y-G_{j_{i-1}}\left(\Pi_{j_{i-1}}\xi_{2}^{i}\right)\right\Vert _{\R^{d}}^{2}\right)\right|\\
\lesssim & \sup_{\begin{array}{c}
\Pi_{j}\left(x_{\ux{}{}}\right)=x_{\lx{}{}}\end{array}}\left|\left\Vert y-G_{j_{i}}\left(x_{\ux{}{}}\right)\right\Vert _{\R^{d}}^{2}-\left\Vert y-G_{j_{i-1}}\left(x_{\lx{}{}}\right)\right\Vert _{\R^{d}}^{2}\right|\\
 & +\E\left|\left\Vert y-G_{j_{i}}\left(\xi_{2}^{i}\right)\right\Vert _{\R^{d}}^{2}-\left\Vert y-G_{j_{i-1}}\left(\Pi_{j_{i-1}}\xi_{2}^{i}\right)\right\Vert _{\R^{d}}^{2}\right|.
\end{align*}
Using Assumption \ref{ass:indep}, we get 
\begin{align*}
 & \left|\left\Vert y-G_{j_{i}}\left(x_{\ux{}{}}\right)\right\Vert _{\R^{d}}^{2}-\left\Vert y-G_{j_{i-1}}\left(x_{\lx{}{}}\right)\right\Vert _{\R^{d}}^{2}\right|\\
 & =\left|\left\langle 2y-G_{j_{i}}\left(x_{\ux{}{}}\right)-G_{j_{i-1}}\left(x_{\lx{}{}}\right),G_{j_{i}}\left(x_{\ux{}{}}\right)-G_{j_{i-1}}\left(x_{\lx{}{}}\right)\right\rangle \right|\\
 & \leq\left\Vert 2y-G_{j_{i}}\left(x_{\ux{}{}}\right)-G_{j_{i-1}}\left(x_{\lx{}{}}\right)\right\Vert _{\R^{d}}\left\Vert G_{j_{i}}\left(x_{\ux{}{}}\right)-G_{j_{i-1}}\left(x_{\lx{}{}}\right)\right\Vert _{\R^{d}}\\
 & \lesssim j_{i-1}^{-\beta}.
\end{align*}
and similarly 
\[
\left|\left\Vert y-G_{j_{i}}\left(\xi_{2}^{i}\right)\right\Vert _{\R^{d}}^{2}-\left\Vert y-G_{j_{i-1}}\left(\Pi_{j_{i-1}}\xi_{2}^{i}\right)\right\Vert _{\R^{d}}^{2}\right|\lesssim j_{i-1}^{-\beta}.
\]
Combining, we obtain 
\[
\delta_{i}\lesssim j_{i-1}^{-\beta}.
\]

We next introduce the $\sigma$-algebra 
\[
\mathcal{S}_{l}=\sigma\left(\left\{ \ubk{k}^{l}=\lbk{k}^{i}\,\mid k=1\dots l\right\} \right)
\]
with the convention that $S_{0}=\{0,\Omega\}$ and let $\Pi_{j_{i-1}}\ux{0}{i}=\lx{0}{i}$
be arbitrary. Then we calculate 
\begin{align*}
 & \rp_{\left(\ux{0}{i},\lx{0}{i}\right)}\left(\exists-\tau\ge l>0:\ubk{l}^{i}\ne\lbk{l}^{i}\right)\\
= & 1-\rp_{\left(\ux{0}{i},\lx{0}{i}\right)}\left(\forall l<-\tau\:\ubk{l}^{i}=\lbk{l}^{i}\right)\\
= & 1-\prod_{l=1}^{-\tau}\rp_{\left(\ux{0}{i},\lx{0}{i}\right)}\left(\ubk{l}^{i}=\lbk{l}^{i}\vert S_{l-1}\right)\\
= & 1-\prod_{l=1}^{-\tau}\left(1-\rp_{\left(\ux{0}{i},\lx{0}{i}\right)}\left(\ubk{l}^{i}\neq\lbk{l}^{i}\vert S_{l-1}\right)\right)\\
= & 1-\prod_{l=1}^{-\tau}\left(1-\E\left(\rp_{\left(\ux{0}{i},\lx{0}{i}\right)}\left(\ubk{l}^{i}\neq\lbk{l}^{i}\vert\mathcal{F}_{l-1}\vee S_{l-1}\right)\vert S_{l-1}\right)\right),
\end{align*}
where the last step follows by the tower property of conditional expectation.
The Markov property and Bernulli's inequality yield 
\begin{align*}
 & \rp_{\left(\ux{0}{i},\lx{0}{i}\right)}\left(\exists-\tau\ge l>0:\ubk{l}^{i}\ne\lbk{l}^{i}\right)\\
 & =1-\prod_{l=1}^{-\tau}\left(1-\E\left(\rp_{\left(\ux{l-1}{i},\lx{l-1}{i}\right)}\left({\ubk{1}^{i}}\neq{\lbk{1}^{i}}\right)\vert S_{l-1}\right)\right)\\
 & \leq1-\prod_{l=1}^{-\tau}(1-\delta_{i})=1-(1-\delta_{i})^{-\tau}\leq1-(1-\delta_{i})^{a_{i-1}}\leq a_{i-1}\delta_{j_{i-1}j_{i}}.
\end{align*}
Hence we have 
\[
\rp\left(\isEvt\right)\leq a_{i-1}\delta_{i},
\]
and putting things together, we obtain that
\[
E_{1}\leq\norm f_{\infty}^{2}\left((1-\alpha_{\star})^{a_{i-1}}+a_{i-1}\delta_{i}\right).
\]
We thus have 
\[
\E_{1}\leq\norm f_{\infty}^{2}(\exp(-c_{\star}a_{i-1})+a_{i-1}j_{i-1}^{-\beta}),
\]
where $c_{\star}=-\log(1-\alpha_{\star})$. In order to optimise the
right hand side and since the first term decreases while the second
term increases with $a_{i}$, we need to balance the two terms by
choosing $a_{i}$ as an appropriate function of $j_{i}$. Indeed,
using \cite[Lemma 4.5]{ASZ14} we have that for $a_{i}\sim\frac{\beta}{c_{\star}}\log(j_{i})$
the two terms are asymptotically balanced as $i\to\infty$ so that
for this choice $E_{1}\lesssim j_{i-1}^{-\beta}$ (check that this
is true).

Now we treat the second term $E_{2},$ which by Assumption \ref{ass:indepObservable}
satisfies 
\begin{eqnarray*}
E_{2}=\E\left(f(\ux{a_{i}}{j_{i}})-f(\Pi_{j_{i-1}}\ux{a_{i}}{j_{i}})\right)^{2} & \lesssim & j_{i-1}^{-\kappa}
\end{eqnarray*}
Finally, combining the bounds for $E_{1}$ and $E_{2}$ yields 
\[
\norm{\Delta_{i}}_{2}^{2}\leq2\left(E_{1}+E_{2}\right)\lesssim j_{i-1}^{-(\beta\wedge\kappa)}.
\]
\end{proof}

\begin{proof}[Proof of Theorem \ref{thm:ExistenceIndep}]
In order for the unbiasing procedure to work we need to have both
finite computing time and finite variance of the estimator $Z$. By
the considerations at the end of subsection \ref{ssec:implback},
it suffices to verify (\ref{eq:varest}) to get the unbiasedness and
finite variance of $Z$. Below we use $c$ as a generic positive constant
which may be change between occurrences.

Using Lemma \ref{lem:ch5lem} and according to the stated choices
of the relevant parameters, we have 
\begin{align*}
\sum_{i\leq l}\frac{\norm{\Delta_{i}}_{2}\norm{\Delta_{l}}_{2}}{\rp(N\ge i)} & =\sum_{i=0}^{\infty}\frac{\norm{\Delta_{i}}_{2}}{\rp(N\geq i)}\sum_{l=i}^{\infty}\norm{\Delta_{l}}_{2}\\
 & \leq c\sum_{i=0}^{\infty}i^{-\frac{rq}{2}+t}\sum_{l=i}^{\infty}l^{-\frac{rq}{2}}\leq c\sum_{i=0}^{\infty}i^{1-rq+t},
\end{align*}
provided $q>\frac{2}{r}$ which holds since $q>\frac{3}{r-\theta}$.
The right hand side is finite provided $t<rq-2$.

Regarding the expected computing time of $Z$, by Assumption \ref{ass:cost}
we have 
\[
\E[\tau]=\sum_{i}^{\infty}t_{i}\rp(N\ge i)\lesssim\sum_{i}^{\infty}i^{\theta q-t}\log(i),
\]
which is finite provided $t>1+\theta q$.

Concatenating, we have that for the unbiased procedure to work we
need to choose $t\in(1+\theta q,rq-2)$, which is possible since $1+\theta q<rq-2$
under the assumption $q>\frac{3}{r-\theta}$. \end{proof}

\subsection{Proofs of the results in section \ref{sec:trans}}
In this section we present the proofs of Theorems \ref{thm:SuffCond} and \ref{thm:pCNSuffCondUnbounded}. The crucial step for both proofs is to derive bounds on $\E d(\ux{0}{i},\lx{0}{i})$ for the appropriate distance $d$, which in turn give  bounds on the decay of $\norm{\Delta_{i}}_{2}$; the method of generating $\ux{0}{i},\lx{0}{i}$ and $\Delta_i$ is summarised in Algorithm \ref{alg:transpCN}. The main idea used to obtain such bounds is explained in section \ref{sub:appSketch-Proof} under artificial conditions. The rigorous bounds used for Theorems \ref{thm:SuffCond} and \ref{thm:pCNSuffCondUnbounded} are contained in Lemma \ref{lem:appPCNdBound} and Lemma \ref{lem:appPCNdtildeBound}  in section \ref{sub:AppPcnNonImmediateDecay}, respectively. Obtaining these bounds is based on results known for coupling of the pCN algorithms on the same state space which are summarised in section \ref{sec:apppCN}. Finally, n section \ref{sec:apptranProofOfMain} we put things together and prove Theorems \ref{thm:SuffCond} and \ref{thm:pCNSuffCondUnbounded}.

\subsubsection{Main idea \label{sub:appSketch-Proof}}
In the following we show how to obtain bounds on the contraction of the transdimensional coupling $K_{j_{i-1}}^{j_i}$ defined in (\ref{eq:coupl1}), under the following artificial assumption on the fixed state space coupling $K_i$ defined in (\ref{eq:coupl2}):
\begin{equation}
\E_{W}d_\tau\left(\varphi_{\ux{}{}}^{i}\left(x_{1},W\right),\varphi_{\ux{}{}}^{i}\left(x_{2},W\right)\right)\leq rd_\tau(x_{1},x_{2}).\label{eq:transSimplifiedAssumptions}
\end{equation}
This assumption does not hold for the pCN algorithm, but allows
us to present the strategy of our proofs while avoiding technicalities and overloaded notation. The fact that $d_\tau$ satisfies the triangle inequality is crucial for our analysis. In particular, it allows us to introduce intermediate
steps $\ix{k}{i}=\varphi_{\ix{}{}}^{i}\left(\lx{k-1}{i},W_{k}^{i}\right)$
by performing a transition from a state of the lower level chain
$\lx{k}{i}$, according to the transition kernel $P_{j_{i}}$ of the
high level chain $\ux{k}{i}$. This enables us to use
(\ref{eq:transSimplifiedAssumptions}) to control the distance between $\ix{k}{i}$ and $\ux{k}{i}$, while at the same time $\ix{k}{i}$ is with high
probability close to $\lx{k}{i}=\varphi_{\lx{}{}}^{i}\left(\lx{k-1}{i},W_{k}^{i}\right)$, since they have the same starting point.
We show that this intuition is accurate below. 

The intermediate step $\ix{k}{i}$
is constructed as follows:
\begin{align*}
\hat{\ix{k}{i}} & =\varphi_{\pux{}{}}^{i}\left(\lx{k-1}{i},\xi_{k}^{i}\right)=\rho\lx{k-1}{i}+\left(1-\rho^{2}\right)^{\frac{1}{2}}\xi_{k}^{i},\\
\ix{k}{i} & =\varphi_{\ux{}{}}^{i}\left(\lx{k-1}{i},W_{k}^{i}\right)=\ind_{[0,\alpha(\lx{k-1}{i},{\pix{k}{i}})]}(U_{k}^{i})\hat{\ix{k}{i}}+\ind_{(\alpha\left(\lx{k-1}{i},{\pix{k}{i}}\right),1]}(U_{k}^{i})\lx{k-1}{i}.
\end{align*}
Using the triangle inequality, we get the bound
\begin{eqnarray}\label{eq:trbnd}
\E d_\tau(\ux{k}{i},\lx{k}{i}) & \leq & \E\left[d_\tau(\ux{k}{i},\ix{k}{i})+d_\tau(\ix{k}{i},\lx{k}{i})\right]\nonumber\\
 & = & \E\left[\E\left(d_\tau(\ux{k}{i},\ix{k}{i})+d_\tau(\ix{k}{i},\lx{k}{i})\vert\mathcal{F}_{k-1}\right)\right],
\end{eqnarray}
where $\mathcal{F}_{k}=\sigma\left(\left\{ \ux{l}{i},\lx{l}{i}\right\} \mid l\leq k\right)$.
We use (\ref{eq:transSimplifiedAssumptions}) together with the Markov property
in order to get 
\begin{align}\label{eq:trbnd1}
\E\left[d_\tau(\ux{k}{i},\ix{k}{i})\vert\mathcal{F}_{k-1}\right]  =  \E\left[\left(K_{i}d_\tau\right)(\ux{k-1}{i},\lx{k-1}{i})\right]\le\E \left[rd_\tau(\ux{k-1}{i},\lx{k-1}{i})\right].
\end{align}
Therefore it is left to consider $\E\left(d_\tau(\ix{0}{i},\lx{0}{i})\vert\mathcal{F}_{-1}\right)$.
Since $d_\tau\leq1,$ we have the bound 
\begin{eqnarray}
\E\left(d_\tau(\ix{k}{i},\lx{k}{i})\mid\mathcal{F}_{k-1}\right) & \leq & \E\left(0\cdot\1_{\text{both reject}}\mid\mathcal{F}_{k-1}\right)\label{eq:pCNcaseBounded}\nonumber\\
 &  & +\E\left(\frac{\left(1-\rho^{2}\right)^{\frac{1}{2}}\left\Vert \xi_{k}^{i}-\Pi_{j_{i-1}}\xi_{k}^{i}\right\Vert }{\tau}\cdot\1_{\text{both accept}}\mid\mathcal{F}_{k-1}\right)\nonumber \\
 &  & +\E\left(1\cdot\1_{\text{one accepts}}\mid\mathcal{F}_{k-1}\right),
\end{eqnarray}
where $\1_{\text{one accepts}}=\1_{\left\{ \ix{k}{i}=\hat{\ix{k}{i}}\text{ xor }\lx{k}{i}=\hat{\lx{k}{i}}\right\} }$.
The probability that only one of the chains accepts can be bounded using the Markov property
as follows: 
\begin{align}
 & \rp\left(\text{one accepts}\mid\mathcal{F}_{k-1}\right)\label{eq:transOnlyOneAccepts}\nonumber\\
 & =\E_{\xi_{k}^{i},U_{k}^{i}}\left(\ind_{\left[\alpha\left(\lx{k-1}{i},\pix{k}{i}\right),\alpha\left(\lx{k-1}{i},\plx{k}{i}\right)\right]}\left(U_{k}^{i}\right)+\ind_{\left[\alpha\left(\plx{k-1}{i},\lx{k}{i}\right),\alpha\left(\lx{k-1}{i},\pix{k}{i}\right)\right]}\left(U_{k}^{i}\right)\right)\nonumber\\
 & =\E_{\xi_{k}^{i}}\left|\alpha\left(\lx{k-1}{i},\hat{\ix{k}{i}}\right)-\alpha\left(\lx{k-1}{i},\hat{\lx{k}{i}}\right)\right|\nonumber \\
 & =\E_{\xi_{k}^{i}}\left|1\wedge\exp\left(g(\lx{k-1}{i})-g(\hat{\ix{k}{i}})\right)-1\wedge\exp\left(g(\lx{k-1}{i})-g(\hat{\lx{k}{i}})\right)\right|\nonumber \\
 & \leq C_{g}\E_{\xi_{k}^{i}}\left\Vert \hat{\ix{k}{i}}-\hat{\lx{k}{i}}\right\Vert \nonumber \\
 & \leq C_{g}\left(1-\rho^{2}\right)^{\frac{1}{2}}\E\left\Vert \xi_{k}^{i}-\Pi_{j_{i-1}}\xi_{k}^{i}\right\Vert, 
\end{align}
where $C_g$ depends on the Lipschitz constant of the log-change of measure $g$.
The second term on the right hand side of  \eqref{eq:pCNcaseBounded} is of similar
form, so that we get the overall bound 
\begin{eqnarray}
\E d_\tau(\lx{k}{i},\ix{k}{i}) & \leq & (\frac{1}{\tau}+C_{g})\left(1-\rho^{2}\right)^{\frac{1}{2}}\E\left\Vert \xi_{k}^{i}-\Pi_{j_{i-1}}\xi_{k}^{i}\right\Vert \nonumber \\
 & \leq & (\frac{1}{\tau}+C_{g})\left(1-\rho^{2}\right)^{\frac{1}{2}}K\sqrt{\sum_{k=j_{i-1}+1}^{j_{i}}\lambda_{k}}=:C_{j_{i-1},j_{i}},\label{eq:transdimAddConstant}
\end{eqnarray}
where we used Cauchy-Schwarz inequality in the last step. Repeated
use of the Markov property and the bounds \eqref{eq:trbnd}, \eqref{eq:trbnd1} and \eqref{eq:transdimAddConstant}, yields
that
\begin{align}
\E d_\tau(\ux{0}{i},\lx{0}{i}) & \leq\E\left[\E\left(d_\tau(\ix{0}{i},\lx{0}{i})+d_\tau(\ux{0}{i},\ix{0}{i})\vert\mathcal{F}_{-1}\right)\right]\nonumber\\
 & \leq\E\left[rd_\tau(\ux{-1}{i},\lx{-1}{i})+C_{j_{i-1},j_{i}}\right]\nonumber\\
 & \leq r\left(r\E d_\tau(\ux{-2}{i},\lx{-2}{i})+C_{j_{i-1},j_{i}}\right)+C_{j_{i-1},j_{i}}\nonumber\\
  &\quad\quad\vdots\nonumber\\
 & \leq r\left(\dots\left(r\E d_\tau(\ux{-a_{i-1}}{i},\lx{-a_{i-1}}{i})+C_{j_{i-1},j_{i}}\right)\dots\right)+C_{j_{i-1},j_{i}}\nonumber\\
 & \leq r^{a_{i-1}}\E d_\tau(\ux{-a_{i-1}}{i},\lx{-a_{i-1}}{i})+C_{j_{i-1},j_{i}}\frac{1-r^{a_{i-1}}}{1-r}\nonumber\\
 & \leq r^{a_{i-1}}+C_{j_{i-1},j_{i}}\frac{1-r^{a_{i-1}}}{1-r},\label{eq:transSimpleRecursion}
\end{align}
where in the last step we used that $d_\tau\leq1.$ Our strategy thus indeed gives a bound on the contraction of the transdimensional coupling $K_{j_{i-1}}^{j_i}$ under the artificial assumption \eqref{eq:transSimplifiedAssumptions} on the contraction of $K_i$. We use the same strategy to get the required contraction bounds in the more realistic settings considered in Lemmas \ref{lem:appPCNdBound} and \ref{lem:appPCNdtildeBound}.

\subsubsection{Overview of the coupling bounds}\label{sec:apppCN}
We next describe how the existing literature yields that the fixed state space coupling in  \eqref{eq:coupl2} leads to  contraction with respect to $d_\tau(x,y)=1\wedge\frac{\protect\norm{x-y}}{\tau}$ and $\tilde{d}(x,y)\coloneq\sqrt{d_{{\tau}}(x,y)\left(1+V(x)+V(y)\right)}$. For simplicity, below we assume that $j_i=i$. This particular coupling is called \emph{the basic coupling}, \cite{hairer2011spectral}. Recall that the contraction bound for a particular coupling is always an upper bound for the Wasserstein distance of the transition kernel, see
Remark \ref{rem:wassersteincontr}.2. In the following we summarise
the relevant results and make connections to geometric ergodicity.

Verifying that a particular coupling contracts is often
difficult, but \cite{weakHarris}, \cite{durmus2014new} and \cite{durmus2014geom}
give verifiable conditions which resemble the well-known conditions
for geometric and polynomial ergodicity. Geometric ergodicity is usually
established using the Harris Theorem by verifying {the existence of} a Lyapunov function,
also called geometric drift condition. That is, it suffices to show the existence of a function $V$, $0<\lambda<1$ and $b>0$ such that
\begin{equation}
PV\leq\lambda V+b\label{eq:AppPCNLyap}
\end{equation}
 and showing that an appropriate small set exists, see \cite[Section 3.4]{RR04}. The problem is
that the resulting error bounds on the ergodic average deteriorate
with dimension because it is difficult to find \textit{good} small sets.

This problem is alleviated when considering Wasserstein convergence.
In particular, the article \cite{weakHarris} establishes a weak Harris
Theorem. It shows exponential convergence with respect to the Wasserstein
distance based on $\tilde{d}(x,y)=\sqrt{d(x,y)(1+V(x)+V(y))}$, for
$d(x,y)\leq1$ a distance-like function. Below we use the letter $d$ to also denote the Wasserstein distance and hope that this does not cause confusion. The small set condition of
the Harris theorem is replaced by the requirements that:
\begin{enumerate}
\item a sub-level set $S$ of $V$ is $d$-small, that is, for all $x$ and
$y$ in $S$
\begin{equation}
d\left(P(x,\cdot),P(y,\cdot)\right)\leq s<1;\label{eq:appPCNsmall}
\end{equation}
\item the transition kernel $P$ is $d$-contracting, that is, there is a $0<c<1$
such that for $d(x,y)<1$
\begin{equation}
d\left(P(x,\cdot),P(y,\cdot)\right)\leq cd(x,y).\label{eq:apppCNcontr}
\end{equation}
\end{enumerate}
For a summary of the weak Harris theorem we refer the reader to Section
2.2.1 of \cite{hairer2011spectral}. Equations \eqref{eq:appPCNsmall} and \eqref{eq:apppCNcontr} are typically established using the fact that the Wasserstein distance can be bounded using a particular coupling. That is, it suffices to establish the existence of couplings $K^{(1)}$ and $K^{(2)}$ such that
\begin{align}
(K^{(1)} d)(x,y)\leq s<1,\label{eq:appPCNsmallcoup}\\
(K^{(2)} d)(x,y)\leq cd(x,y) \text{ for } d(x,y)<1.\label{eq:apppCNcontrcoup}
\end{align}
An inspection of the proof of the weak Harris theorem
in \cite{weakHarris}, shows that in fact the contraction property
is established for the coupling arising from 
\begin{itemize}
\item if $d(x,y)<1$ use coupling $K^{(2)}$; 
\item else if $x,y\in S$ use coupling $K^{(1)}$;
\item else use any coupling,
\end{itemize}
rather than directly for the Wasserstein distance which takes the infimum over all couplings. Thus, the same is true
for Theorems 2.14 and 2.17 of \cite{hairer2011spectral}, that derive
a non-explicit but dimension-independent contraction rate for the basic coupling $K_i$ of the pCN algorithm, for target measures which are changes of measure from a Gaussian distribution with log-density satisfying Assumption \ref{ass:lippCN}. More precisely, the proof shows that:\begin{itemize}
\item \eqref{eq:AppPCNLyap} is satisfied for $V(x)=\exp(\norm{x})$ and with $b$ and $\lambda$ which are dimension independent;
\item there exists a dimension-independent $\tau$, such that the basic coupling $K_i$ satisfies both \eqref{eq:appPCNsmallcoup} and \eqref{eq:apppCNcontrcoup} for the distance $d_\tau$, for $s$ and $c$ which are also dimension-independent.
\end{itemize} In particular, this shows that  for any $0<r<1,$ there
exists $n_{0}=n_0(r)\in\N$ such that 
\begin{equation}\label{eq:apppCNdtildebound}
\left((K_{i})^{n_{0}(r)}\tilde{d}\right)(x,y)\leq r\tilde{d}(x,y)\text{ for any }x,y\in\state_{i},i\in\N,
\end{equation}
for \[
\tilde{d}(x,y)\coloneq\sqrt{d_\tau(x,y)\left(1+V(x)+V(y)\right)}
\]
where $V(x)=\exp(\norm{x})$.

The work of \cite{weakHarris} has been extended  
\begin{itemize}
\item  in \cite{durmus2014new} to cover polynomial ergodicity using more complicated
drift conditions;
\item in \cite{durmus2014geom} to obtain more explicit bounds in the geometric
case.
\end{itemize}
The article \cite{durmus2014new} explicitly considers the pCN algorithm.
In particular Equation (68) in \cite{durmus2014new} establishes that
\begin{equation}\label{eq:appPCNimmediate}
K_{i}d_\tau  \leq  d_\tau.
\end{equation}
Combining  Proposition 12 of \cite{durmus2014new}, Lemma 3.2 of \cite{hairer2011spectral} and Theorem 1 of \cite{durmus2014geom}, we get that the basic coupling decays exponentially
\begin{equation}
\left(K_{i}\right)^{n}d_\tau(x,y)  \leq  Cr^{n}\left(V(x)+V(y)\right)\label{eq:transDimBound},
\end{equation}where $C$ is dimension independent and $V$ as above.

In the next subsection we will employ the above contraction results for the fixed state-space basic coupling $K_i$ of the pCN algorithm, to show the decay of the transdimensional coupling $K_{j_{i-1}}^{j_{i}}$ of the pCN algorithm, defined in \eqref{eq:coupl1}.

%%%%%%%%

\subsubsection{\label{sub:AppPcnNonImmediateDecay}Coupling bounds between $\ux{0}{i}$ and $\lx{0}{i}$  }\label{sub:ApppCNunboundedProof}

In the following we use the results reviewed in the previous subsection to obtain coupling bounds between $\ux{0}{i}$ and $\lx{0}{i}$  in terms of $d_\tau(x,y)=1\wedge\frac{\protect\norm{x-y}_{\state}}{\tau}$ (Lemma \ref{lem:appPCNdBound}) and  $\tilde{d}(x,y)\coloneq\sqrt{d(x,y)\left(1+V(x)+V(y)\right)}$ (Lemma \ref{lem:appPCNdtildeBound}). These bounds in turn imply bounds on the decay of $\norm{\Delta_i}_2$ which are crucial for the proofs of Theorem \ref{thm:SuffCond} and \ref{thm:pCNSuffCondUnbounded}.

\begin{lem}
\label{lem:appPCNdBound}Under Assumption \ref{ass:lippCN}, there exist $\tau, C>0$ and $r\in(0,1)$, such that 
\begin{equation}
\E d_\tau\left(\ux{0}{i},\lx{0}{i}\right)\leq Cr^{a_{i-1}}+a_{i-1}C_{j_{i-1},j_{i}},\label{eq:transDecauWConstantBound}
\end{equation}
with $C_{j_{i-1},j_{i}}\coloneq\sqrt{\sum_{k=j_{i-1}+1}^{j_{i}}\lambda_{k}}$. In particular if $f:\state\to\R$ is $s$-H\"older continuous with respect to $d_\tau$ for some $s\in[\frac12,1]$, for the choice $j_{i}\sim r^{-a_{i}\frac{2}{2\alpha-1}}$ we get the estimate \begin{equation}
\norm{\Delta_{i}}_{2}^{2}\lesssim r^{a_{i-1}}.\label{eq:apppCNdBound}
\end{equation}\end{lem}
\begin{proof}
From subsection \ref{sec:apppCN} we know that there are $\tau,$ $C$
and $r\in(0,1)$ independent of $i$, such that the fixed state space coupling $K_i$ satisfies
\begin{eqnarray}
\left(K_{i}\right)^{n}d_\tau(x,y) & \leq & Cr^{n}\left(V(x)+V(y)\right)\label{eq:transDimBound}\\
K_{i}d_\tau & \leq & d_\tau,\nonumber 
\end{eqnarray}
for any $n\in\N$ and for the Lyapunov function $V(x)=\exp(\norm{x})$.
This statement is much weaker than the artificial assumption (\ref{eq:transSimplifiedAssumptions}) in subsection \ref{sub:appSketch-Proof}
because of the multiplicative constant and the Lyapunov function. As a result
we cannot simply recurse as in (\ref{eq:transSimpleRecursion}). In the following, the constant $C$ may change from occurrence to occurrence, but $C$ it is always independent of $n$ and $i$.

We define recursively for $l\in\N$ and for a fixed $r\in\mathbb{Z}$, the $l$-step random functions
\begin{eqnarray*}
\varphi_{\lx{}{}}^{i,l}\left(x,\left\{ W_{s}\right\} _{s=1+r}^{l+r}\right) & \coloneq & \varphi_{\lx{}{}}^{i}\left(\varphi_{\lx{}{}}^{i,l-1}\left(x,\left\{ W_{s}\right\} _{s=r+1}^{l+r-1}\right),W_{l+r}\right)\\
\varphi_{\ux{}{}}^{i,l}\left(x,\left\{ W_{s}\right\} _{s=1+r}^{l+r}\right) & \coloneq & \varphi_{\ux{}{}}^{i}\left(\varphi_{\ux{}{}}^{i,l-1}\left(x,\left\{ W_{s}\right\} _{s=r+1}^{r+l-1}\right),W_{l+r}\right)
\end{eqnarray*}
with the convention that $\varphi^{i,0}_{\ux{}{}}=\varphi^{i,0}_{\lx{}{}}=Id$. Using the triangle inequality and following the strategy described at the beginning of subsection \ref{sub:appSketch-Proof}, we can bound
\begin{align}
\E d_\tau\left(\ux{0}{i},\lx{0}{i}\right)&= \E d_\tau\left(\varphi_{\ux{}{}}^{i,a_{i-1}}\left(\varphi_{\ux{}{}}^{i,a_{i}-a_{i-1}}\left(x_{0},\left\{ W_{s}\right\} _{s=-a_{i}+1}^{-a_{i-1}}\right),\left\{ W_{s}\right\} _{s=-a_{i-1}+1}^{0}\right),\right.\nonumber \\
 \quad&\quad  \left.\varphi_{\lx{}{}}^{i,a_{i-1}}\left(x_{0},\left\{ W_{s}\right\} _{s=-a_{i-1}+1}^{0}\right)\right)\nonumber\\
 & \leq  \E d_\tau\left(\varphi_{\ux{}{}}^{i,a_{i-1}}\left(\ux{-a_{i-1}}{i},\left\{ W_{s}\right\} _{s=-a_{i-1}+1}^{0}\right),\varphi_{\ux{}{}}^{i,a_{i-1}}\left(x_{0},\left\{ W_s\right\} _{s=-a_{i-1}+1}^{0}\right)\right)\nonumber\\
 &   +\E d_\tau\left(\varphi_{\ux{}{}}^{i,a_{i-1}}\left(x_{0},\left\{ W_{s}\right\} _{s=-a_{i-1}+1}^{0}\right),\varphi_{\lx{}{}}^{i,a_{i-1}}\left(x_{0},\left\{ W_{s}\right\} _{s=-a_{i-1}+1}^{0}\right)\right)\nonumber\\
 & =: R_1+R_2.\label{eq:TransDecayWConstantInit}
\end{align}
Using (\ref{eq:transDimBound}) and the Markov property, we have that 
\begin{align}\label{eq:spcnr1}
 R_1 &\leq  Cr^{a_{i-1}}\E\left(V(x_{0})+V(\ux{a_{i-1}}{i})\right)\leq  Cr^{a_{i-1}},
\end{align}
where the second inequality follows from the fact that $\sup_{j,n}P_{j}^{n}V(x_{0})<\infty$. This can be seen by  induction on \eqref{eq:AppPCNLyap} which as discussed in the previous subsection for the pCN algorithm is satisfied with dimension independent parameters, which gives
\[
P^{n}V=\lambda^{n}V+\frac{1-\lambda^{n}}{1-\lambda}b,
\]
implying
\[
\sup_{n}P^{n}V\leq V+\frac{1}{1-\lambda}b.
\]

By repeatedly using the triangle inequality in order to introduce intermediate steps which differ from each other in the evolution at only one time-step, we can estimate $R_2$ as below:
 \begin{align*}
R_2 &\leq \E\Bigg[d_\tau\Big(\varphi_{\ux{}{}}^{i,a_{i-1}-1}\big(\varphi_{\ux{}{}}^{i}(x_0, W_{-a_{i-1}+1}),\{ W_{s}\} _{s=-a_{i-1}+2}^{0}\big),\\
 &\quad\quad\quad\varphi_{\ux{}{}}^{i,a_{i-1}-1}\big(\varphi_{\lx{}{}}^{i}(x_0,W_{-a_{i-1}+1}),\{ W_{s}\} _{s=-a_{i-1}+2}^{0}\big)\Big)\\
&\quad+d_\tau\Big(\varphi_{\ux{}{}}^{i,a_{i-1}-1}\big(\varphi_{\lx{}{}}^{i}(x_0,W_{-a_{i-1}+1}),\{ W_{s}\} _{s=-a_{i-1}+2}^{0}\big),\\
 &\quad\quad\quad\varphi_{\lx{}{}}^{i,a_{i-1}-1}\big(\varphi_{\lx{}{}}^{i}(x_0,W_{a_{i-1}+1}),\{ W_{s}\} _{s=-a_{i-1}+2}^{0}\big)\Big)\Bigg]\\
 &\quad\quad\vdots\nonumber\\
  &\leq  \E\sum_{k=0}^{a_{i-1}-1}d_\tau\left(\varphi_{\ux{}{}}^{i,k}\left(\varphi_{\ux{}{}}^{i}\left(\varphi_{\lx{}{}}^{i,a_{i}-1-k}\left(x_{0},\left\{ W_{s}\right\} _{s=-a_{i-1}+1}^{-k-1}\right),W_{-k} \right),\left\{ W_{s}\right\} _{s=-k+1}^{0}\right),\right.\\
   &\quad\quad\quad \left.\varphi_{\ux{}{}}^{i,k}\left(\varphi_{\lx{}{}}^{i}\left(\varphi_{\lx{}{}}^{i,a_{i}-1-k}\left(x_{0},\left\{ W_{s}\right\} _{s=-a_{i-1}+1}^{-k-1}\right),W_{-k} \right),\left\{ W_{s}\right\} _{s=-k+1}^{0}\right)\right).
\end{align*}
Since $K_i d_\tau\leq d_\tau$ by  (\ref{eq:transDimBound})
we hence have
\begin{align*}
 R_2 &\leq \; \E\sum_{k=0}^{a_{i-1}-1}d_\tau\left(\varphi_{\ux{}{}}^{i}\left(\varphi_{\lx{}{}}^{i,a_{i}-1-k}\left(x_{0},\left\{ W_{s}\right\} _{s=-a_{i-1}+1}^{-k-1}\right), W_{-k} \right),\right.\\
   & \left.\varphi_{\lx{}{}}^{i}\left(\varphi_{\lx{}{}}^{i,a_{i}-1-k}\left(x_{0},\left\{ W_{s}\right\} _{s=-a_{i-1}+1}^{-k-1}\right),W_{-k} \right)\right).
\end{align*}
Similarly to (\ref{eq:transdimAddConstant}), we see that for
each $k$ the summands are bounded by $C_{j_{i-1},j_{i}}$, hence we get that
\begin{eqnarray}\label{eq:spcnr2}
 R_2 \leq  a_{i-1}C_{j_{i-1},j_{i}}
\end{eqnarray}
with $C_{j_{i-1},j_{i}}\coloneq\sqrt{\sum_{k=j_{i-1}+1}^{j_{i}}\lambda_{k}}.$
Combining \eqref{eq:spcnr1} and \eqref{eq:spcnr2} with \eqref{eq:TransDecayWConstantInit}, we get the 
desired bound \eqref{eq:transDecauWConstantBound}.

To (roughly) optimise the right hand side of \eqref{eq:transDecauWConstantBound}, we require that the two terms have bounds of the same order.
Since $\lambda_{\ell}\lesssim \ell^{-2\alpha}$,
we have that
\[
C_{j_{i-1},j_{i}}=\Big(\sum_{k=j_{i-1}+1}^{j_{i}}\lambda_{k}\Big)^{\frac{1}{2}}\lesssim j_{i-1}^{\frac{1-2a}{2}}.
\] We hence require
 $a_{i-1}j_{i-1}^{\frac{1-2a}{2}}= r^{a_{i-1}}$ and using \cite[Lemma 4.5]{ASZ14} we get that the choice $j_i\sim r^{a_{i}\frac{2}{1-2a}}$ as $i\to\infty$ yields that
 \[
\E d_\tau\left(\ux{0}{i},\lx{0}{i}\right)\lesssim r^{a_{i-1}}.
\]
The estimate \eqref{eq:apppCNdBound} then follows from \eqref{eq:pcnDeltaBound}.
\end{proof}

\begin{lem}
\label{lem:appPCNdtildeBound}Under Assumption \ref{ass:lippCN}, there exist $\tau$ and $r\in(0,1)$, such that 
\begin{align}
 \E\tilde{d}\left(\ux{0}{i},\lx{0}{i}\right)\lesssim r^{a_{i-1}}+C_{j_{i-1},j_{i}}^{\frac{1}{2}},
\end{align}
with $C_{j_{i-1},j_{i}}\coloneq\sqrt{\sum_{k=j_{i-1}+1}^{j_{i}}\lambda_{k}}$. In particular if $f:\state\to\R$ is $\frac12$-H\"older continuous with respect to $\tilde{d}$, for the choice $j_{i}\sim Cr^{-a_{i}\frac{4}{2\alpha-1}}$ we have the estimate \begin{equation}
\norm{\Delta_{i}}_{2}^{2}\lesssim r^{a_{i-1}}.\label{eq:apppCNdtildeBound}
\end{equation}\end{lem}
\begin{proof}
Combining \eqref{eq:appPCNimmediate} with \eqref{eq:AppPCNLyap}, which as discussed in the previous subsection are both satisfied for the pCN algorithm with dimension independent constants $b>0, 0<\lambda<1$ and $\tau>0$ and for $V(x)=\exp(\norm{x})$, 
 we get the following bound 
\begin{align}
\left((K_i)^{n}\tilde{d}\right)(x,y) & \leq\left((K_i)^{n}d_\tau(x,y)\right)^{\frac{1}{2}}\left(1+\left((K_i)^{n}V\right)(x)+\left((K_i)^{n}V\right)(y)\right)^{\frac{1}{2}}\nonumber \\
 & \leq\left(d_\tau(x,y)\right)^{\frac{1}{2}}\left(1+V(x)+V(y)+\frac{2b}{1-\lambda}\right)^{\frac{1}{2}}\nonumber \\
 & \leq\sqrt{\frac{2b}{1-\lambda}+1}\cdot\tilde{d}(x,y)\coloneq\tilde{b}\tilde{d}(x,y).\label{eq:ApppCNanystep}
\end{align}
In the first inequality we used Cauchy-Schwarz and in the second we used that \eqref{eq:AppPCNLyap} implies  \begin{equation}\label{eq:itlyap}
P^nV\leq \lambda^nV+\frac{b}{1-\lambda}.
\end{equation} A major problem
is that $\tilde{d}$ usually does not satisfy the triangle inequality, so that the method described in Subsection \ref{sub:appSketch-Proof} cannot be applied directly.
However, this can be circumvented using the technique of section 4.1.1
of \cite{hairer2011spectral}. More precisely, we define 
\[
\hat{d}(x,y)\coloneq\sqrt{\inf_{n,x=z_{1},\dots,y=z_{n}}\sum_{j=1}^{n-1}d_{0}\left(z_{j},z_{j+1}\right)}
\]
with $d_{0}=d_\tau\left(1+V(x)+V(y)\right)$; $\hat{d}$  satisfies the triangle inequality
by construction. Following the proof of Lemma 4.1.1 in \cite{hairer2011spectral},
it is possible to show that there exists a constant $C_L\leq1$ such that \[
C_{L}\tilde{d}\leq\hat{d}\leq\tilde{d}.
\]
Using \eqref{eq:apppCNdtildebound}, we choose $n_0=n_{0}(\frac{C_{L}}{2})$ such that \begin{equation}
\left((K_{j})^{n_{0}}\tilde{d}\right)(x,y)\leq\frac{C_{L}}{2}\tilde{d}(x,y)\text{ for any }j.\label{eq:pcndtildercoupling}
\end{equation}
Using the triangle inequality for $\hat{d}$ and $\tilde{d}\leq\frac1{C_L}\hat{d}$, we get
\begin{align*}
\E\tilde{d}\left(\ux{0}{i},\lx{0}{i}\right)\leq & \frac{1}{C_{L}}\E\hat{d}\left(\varphi_{\ux{}{}}^{i,n_{0}}\left(\ux{-n_{0}}{i},\left\{ W_{s}\right\} _{s=-n_{0}+1}^{0}\right),\varphi_{\ux{}{}}^{i,n_{0}}\left(\lx{-n_{0}}{i},\left\{ W_{s}\right\} _{s=-n_{0}+1}^{0}\right)\right)\\
 & +\frac{1}{C_{L}}\E\hat{d}\left(\varphi_{\lx{}{}}^{i,n_{0}}\left(\lx{-n_{0}}{i},\left\{ W_{s}\right\} _{s=-n_{0}+1}^{0}\right),\varphi_{\ux{}{}}^{i,n_{0}}\left(\lx{-n_{0}}{i},\left\{ W_{s}\right\} _{s=-n_{0}+1}^{0}\right)\right)\\=:&R_{1}+R_{2},
\end{align*}
where the $l$-step random functions $\varphi_{\ux{}{}}^{i,l}
,\varphi_{\lx{}{}}^{i,l}$ are defined as in the proof of Lemma \ref{lem:appPCNdBound}. Using \eqref{eq:pcndtildercoupling}, we get that 
\[
R_{1}\leq\frac{1}{2}\tilde{d}\left(\ux{-n_{0}}{i},\lx{-n_{0}}{i}\right).
\]
Using the triangle inequality for $\hat{d}$, $\hat{d}\leq\tilde{d}$ and 
\eqref{eq:ApppCNanystep}, we have 
\begin{align}
R_2 &\leq \frac{1}{C_L}\E\sum_{k=0}^{n_{0}-1}\hat{d}\left\{ \varphi_{\ux{}{}}^{i,k}\left(\varphi_{\ux{}{}}^{i}\left(\varphi_{\lx{}{}}^{i,n_{0}-1-k}\left(\lx{-n_{0}}{i},\left\{ W_{s}\right\} _{s=-n_{0}+1}^{-k-1}\right),W_{-k} \right),\left\{ W_{s}\right\} _{s=-k+1}^{0}\right),\right.\nonumber\\
 & \left.\varphi_{\ux{}{}}^{i,k}\left(\varphi_{\lx{}{}}^{i}\left(\varphi_{\lx{}{}}^{i,n_{0}-1-k}\left(\lx{-n_{0}}{i},\left\{ W_{s}\right\} _{s=-n_{0}+1}^{-k-1}\right),W_{-k} \right),\left\{ W_{s}\right\} _{s=-k+1}^{0}\right)\right\} \nonumber\\
 & \le\frac{\tilde{b}}{C_L}\E\sum_{k=0}^{n_{0}-1}\tilde{d}\left\{ \varphi_{\ux{}{}}^{i}\left(\varphi_{\lx{}{}}^{i,n_{0}-1-k}\left(\lx{-n_{0}}{i},\left\{ W_{s}\right\} _{s=-n_{0}+1}^{-k-1}\right),W_{-k} \right),\right.\nonumber\\
 & \left.\varphi_{\lx{}{}}^{i}\left(\varphi_{\lx{}{}}^{i,n_{0}-1-k}\left(\lx{-n_{0}}{i},\left\{ W_{s}\right\} _{s=-n_{0}+1}^{-k-1}\right),W_{-k} \right)\right\}\label{eq:r2bbb}
\end{align}
 The next step is to derive a bound
on the one-step difference between $\varphi_{\lx{}{}}^{i}\left(x,W\right)$
and $\varphi_{\ux{}{}}^{i}\left(x,W\right)$. We obtain the following
bound using the Cauchy-Schwarz inequality,  \eqref{eq:transdimAddConstant}
and  \eqref{eq:AppPCNLyap} 
\begin{align*}
\E\tilde{d}\left(\varphi_{\lx{}{}}^{i}\left(x,W\right),\varphi_{\ux{}{}}^{i}\left(x,W\right)\right)\leq & \left(\E d_\tau\left(\varphi_{\lx{}{}}^{i}\left(x,W\right),\varphi_{\ux{}{}}^{i}\left(x,W\right)\right)\right)^{\frac{1}{2}}\\
 & \cdot\left(1+\E V(\varphi_{\lx{}{}}\left(x,W\right))+\E V(\varphi_{\ux{}{}}\left(x,W\right)\right)^{\frac{1}{2}}\\
\lesssim & C_{j_{i-1},j_{i}}^{\frac{1}{2}}(1+2\lambda V(x)+2b)^{\frac{1}{2}}.
\end{align*}
Notice that the application of Cauchy-Schwarz here leads to $C_{j_{i-1},j_{i}}^{\frac{1}{2}}$
instead of $C_{j_{i-1},j_{i}}$ in section \ref{sub:AppPcnNonImmediateDecay}.
This is the reason for the stronger condition $a>2\theta+\frac{1}{2}$
in Theorem \ref{thm:pCNSuffCondUnbounded} compared to $a>\theta+\frac{1}{2}$
in Theorem \ref{thm:SuffCond}. Using this bound on the right hand side of \eqref{eq:r2bbb}, yields\begin{align*}
R_{2} & \leq\frac{\tilde{b}}{C_L}\E\sum_{k=0}^{n_{0}-1}C_{j_{i-1},j_{i}}^{\frac{1}{2}}(1+2\lambda V(\lx{-k-1}{i})+2b)^{\frac{1}{2}}.
\end{align*}
Using the Cauchy-Schwarz
inequality together with \eqref{eq:itlyap} we thus get that 
\begin{align*}
R_2&\leq \frac{\tilde{b}}{C_L}C_{j_{i-1},j_{i}}^{\frac{1}{2}}\sum_{k=0}^{n_0-1}\left(\E1+2\lambda\lambda^{a_{i-1}-k-1}V(x_{0})+2\lambda\frac{2b}{1-\lambda}+2b\right)^\frac12\\
 &\leq M\frac{\tilde{b}}{C_L}C_{j_{i-1},j_{i}}^{\frac{1}{2}}n_{0}\left(1+2V(x_{0})+\frac{8b}{1-\lambda}\right)^{\frac{1}{2}},
\end{align*}
where the constant $M$ only depends on $\lambda$ and $n_0$ and we used that $\lambda<1$ and $a_i$ is increasing. We abuse notation and write $M=M\frac{\tilde{b}}{C_L}n_0$.
Combining the bounds for $R_1$ and $R_2$, we obtain that 
\[
\E\tilde{d}\left(\ux{0}{i},\lx{0}{i}\right)\leq\frac{1}{2}\E\tilde{d}\left(\ux{-n_{0}}{i},\lx{-n_{0}}{i}\right)+MC_{j_{i-1},j_{i}}^{\frac{1}{2}}\left(1+2V(x_{0})+\frac{8b}{1-\lambda}\right)^{\frac{1}{2}}.
\]
Finally, using the Markov property we can iterate the above bound $k=\left\lfloor \frac{a_{i-1}}{n_{0}}\right\rfloor $times
to obtain 
\begin{align*}
  \E\tilde{d}\left(\ux{0}{i},\lx{0}{i}\right)
 & \leq\left(\frac{1}{2}\right)^{k}\E\tilde{d}\left(\ux{-kn_{0}}{i},\lx{-kn_{0}}{i}\right)+2MC_{j_{i-1},j_{i}}^{\frac{1}{2}}\left(1+2V(x_{0})+\frac{8b}{1-\lambda}\right)^{\frac{1}{2}}\\
 &  \leq\left(\frac{1}{2}\right)^{k}\E\tilde{b}\tilde{d}\left(\ux{-a_{i-1}}{i},x_0 \right)
+2MC_{j_{i-1},j_{i}}^{\frac{1}{2}}\left(1+2V(x_{0})+\frac{8b}{1-\lambda}\right)^{\frac{1}{2}}\\   & \leq\left(\frac{1}{2}\right)^{k}\tilde{b}\E\sqrt{1+V(x_0)+V(\ux{-a_{i-1}}{i})}+CC_{j_{i-1},j_{i}}^{\frac{1}{2}}\\
 & \leq\left(\frac{1}{2}\right)^{k}\tilde{b}\sqrt{1+2V(x_0)+\frac{b}{1-\lambda}}+CC_{j_{i-1},j_{i}}^{\frac{1}{2}}\lesssim r^{a_{i-1}}+C_{j_{i-1},j_{i}}^{\frac{1}{2}}
\end{align*}
where to get the first inequality we summed-up a geometric series, in the second inequality we used \eqref{eq:ApppCNanystep} and where $r=\left(\frac{1}{2}\right)^{n_{0}(\frac{C_{L}}{2})^{-1}}$. 

The rest of the proof is very similar to the last part of the proof of Lemma \ref{lem:appPCNdBound} (using \eqref{eq:delunb} instead of \eqref{eq:pcnDeltaBound}), and is hence omitted.
\end{proof}

\subsubsection{Proofs of Theorems \ref{thm:SuffCond} and \ref{thm:pCNSuffCondUnbounded}}\label{sec:apptranProofOfMain}
In order to prove Theorems \ref{thm:SuffCond} and \ref{thm:pCNSuffCondUnbounded}
we need to first use Proposition \ref{prop:rhee} which gives conditions
securing unbiasedness and finite variance of $Z$ and then make sure
that these conditions are compatible with a finite expected computing
time. 

\begin{proof}[Proof of Theorem \ref{thm:SuffCond}]
By the considerations at the end of subsection \ref{ssec:implback},
in order to get the unbiasedness and
finite variance of $Z$ it suffices to verify (\ref{eq:varest}). Using (\ref{eq:apppCNdBound}), we have that for the stated choices of $a_i$ and $j_i$, it holds
\begin{align}
\sum_{i\le l}\frac{\norm{\Delta_{i}}_{2}\norm{\Delta_{l}}_{2}}{\rp(N\ge i)} & =  \sum_{i=0}^{\infty}\frac{\norm{\Delta_{i}}_{2}}{\rp(N\ge i)}\sum_{l=i}^{\infty}\norm{\Delta_{l}}_{2}\lesssim\sum_{i=0}^{\infty}\frac{\norm{\Delta_{i}}_{2}}{\rp(N\ge i)}\frac{r^{\frac{m}{2}i}}{1-r^{\frac{1}{2}}}\nonumber \\
 & \lesssim\sum_{i=0}^{\infty}\frac{r^{mi}}{\rp(N\ge i)},\label{eq:wassersteinVarianceBound}
\end{align}
where $r<1$ is defined in Lemma \ref{lem:appPCNdBound}. It is hence sufficient
to choose the distribution of $N$ such that $\sum_{i}\frac{r^{mi}}{\rp(N\ge i)}<\infty$. A valid choice
is for example $\rp(N\ge i)\propto r^{(m-\epsilon)i}$ for $\epsilon>0$
which can be arbitrarily small.

Regarding the expected computing time of $Z$, we have that it is
equal to $\sum_{i=0}^{\infty}t_{i}\rp(N\geq i)$, where $t_{i}$ is
the expected time to generate $\Delta_{i}$. By Assumption \ref{ass:pcncost}
we have $t_i\lesssim a_i j_i^\theta$, hence 
\[\sum_{i=0}^{\infty}t_{i}\rp(N\geq i)\lesssim \sum_{i=0}^\infty i r^{(\frac{2\theta m}{1-2a}+m-\epsilon)i}.\] To get that the right hand side is finite, we need to have $\epsilon<\frac{2\theta m}{1-2a}+m$ and such a choice is possible since $a>\theta+\frac12$.
\end{proof}
\begin{proof}[Proof of Theorem \ref{thm:pCNSuffCondUnbounded}]
The proof is almost identical to the proof of Theorem \ref{thm:SuffCond}, and is hence omitted.
\end{proof}

\begin{rem}
In Theorem \ref{thm:ExistenceIndep} in section \ref{sec:is}, we give an example of parameter choices for which the unbiasing procedure works, which is such that $a_{i}$ grow logarithmically and
$j_{i}$ polynomially in $i$. A simple calculation shows that we could
have made the same choices here and would have ended up with the same
condition on the regularity, $\alpha$, of the reference measure $\mu_0$. The present choice implies
that the random variable $N$ has moments of all orders. On the other
hand the dimensionality $j_{N}$ increases exponentially in $N$.
Thus, the comparison of both approaches depends on the concrete choices. 
\end{rem}
%%%%%%%%%%%%%%% TECHNICAL RESULTS %%%%%%%%%%

\section{Appendix}\label{sec:appendix}

\subsection{Generalisation of Proposition \ref{prop:rhee}\label{sub:GeneralisedProp}}
In this section we state and prove the generalisation of Proposition \ref{prop:rhee} to the setting of Hilbert space-valued random variables.
\begin{prop}
\label{prop:genrhee} Suppose that $\left(\Delta_{i}:i\ge0\right)$ is a sequence of random
variables with values in a Hilbert space $H$. Let $N$ be an integer-valued random variable which is independent of the $\Delta_i$'s and satisfies $\rp(N\ge i)>0$ for all $i\ge0$. Define $V={L}^{2}\left(\Omega,H,\rp\right)$
with norm 
\[
\left\Vert v\right\Vert _{V}=\left(\int_{\Omega}\left\Vert v(\omega)\right\Vert _{H}^{2}d\rp(\omega)
\right)^\frac12\]
and assume that 
\begin{equation}
\sum_{i\leq j}\frac{\norm{\Delta_{i}}_{V}\norm{\Delta_{j}}_{V}}{\rp\left(N\ge i\right)}<\infty.\label{eq:vectorValuesRhee}
\end{equation} 
Then $Y_{n}\coloneq\sum_{i=0}^{n}\Delta_{i}$ converges in $V$
to a limit $Y\coloneq\sum_{i=0}^{\infty}\Delta_{i}$ as $n\rightarrow\infty$.
Let $\alpha=\E Y$$\left(=\lim_{n\rightarrow\infty}\E Y_{n}\right)\in H$
and suppose that for all $i$, ${\tilde \Delta _i}$  is a copy of $\Delta_i$ such that $\{\tilde{\Delta} _i\}$ are mutually independent. Then $\tilde{Z}\coloneq\sum_{i=0}^{N}\frac{\tilde{\Delta}_{i}}{\rp\left(N\ge i\right)}$
is an unbiased estimator for $\alpha\in H$ with finite second moment
\[
\E[\tilde{Z}\otimes\tilde{Z}]=\sum_{i=0}^{\infty}\frac{\tilde{\nu}_{i}}{\rp\left(N\ge i\right)}\in H\otimes H,
\]
where $\tilde{\nu}_{i}=\E[\Delta_{i}\otimes\Delta_{i}]+2\E[\Delta_{i}]\otimes\E[\left(Y-Y_{i}\right)].$\end{prop}
\begin{proof}
The proof follows the same arguments as the proof of Proposition 6 in \cite{RheePHD}, but makes
use of the Fubini theorem for Bochner integrals, stated below. For more information
about Bochner integrals and the proof of the Fubini theorem
we refer the reader to Appendix E of \cite{DudleyUniLimits}. 
\begin{thm}\label{thm:fubini}
Let $(X,\mathcal{A},\mu)$ and $\left(Y,\mathcal{B},\nu\right)$ be
$\sigma$-finite measure spaces. Let $f$ be a measurable function
from $X\times Y$ into a Banach space $S$ such that 
\[
\int\int\left\Vert f(x,y)\right\Vert d\mu(x)d\nu(x)<\infty.
\]
Then for $\mu$-almost every $x$, $f(x,\cdot)$ is Bochner integrable from
$X$ into $S$, and 
\[
\int fd(\mu\otimes\nu)=\int\int f(x,y)d\mu(x)d\nu(y)=\int\int f(x,y)d\nu(y)d\mu(x).
\]
\end{thm}
Notice that \eqref{eq:vectorValuesRhee} implies that $\sum_{i=0}^\infty\norm{\Delta_i}_V<\infty$, which in turn implies that the $Y_{n}$'s
form a Cauchy sequence in $V$, since for $n<m$ 
\[
\left\Vert Y_{n}-Y_{m}\right\Vert _{V}\leq\sum_{k=n+1}^{m}\left\Vert \Delta_{k}\right\Vert _{V}\leq\sum_{k=n+1}^{\infty}\left\Vert \Delta_{k}\right\Vert _{V}\rightarrow0\text{ as }n\rightarrow\infty.
\] Furthermore, using Cauchy-Schwarz inequality  we see that $\sum_{i=0}^\infty\norm{\Delta_i}_V<\infty$ implies that $\sum_{i=0}^\infty\E\norm{\Delta_i}_H<\infty$.
Using Jensen's inequality for Bochner integrals (see \cite{DudleyUniLimits}) and the Fubini-Tonelli theorem for positive random variables, we then get {\begin{align*}
\left\Vert \E \sum_{i=0}^{\infty}\Delta_{i}-\E \sum_{i=0}^{n}\Delta_{i}\right\Vert _{H} & \leq\E \left\Vert \sum_{i=n+1}^{\infty}\Delta_{i}\right\Vert _{H}
  \leq\E \sum_{i=n+1}^{\infty}\left\Vert \Delta_{i}\right\Vert _{H}\\
 & =\sum_{i=n+1}^{\infty}\E \left\Vert \Delta_{i}\right\Vert _{H}\rightarrow0,\text{ as }n\rightarrow\infty,
\end{align*}} which implies that $\E Y_{n}\rightarrow\E Y$ in $H$
as $n\rightarrow\infty$. 

In order to show that $Z$ is unbiased, we note that 
\begin{align*}
\E \sum_{i=0}^{\infty}\frac{\tilde{\Delta}_{i}\1_{N\ge i}}{\rp\left(N\ge i\right)} & =\sum_{i=0}^{\infty}\frac{\E [\tilde{\Delta}_{i}\1_{N\ge i}]}{\rp\left(N\ge i\right)}=\sum_{i=0}^{\infty}\frac{\E [\tilde{\Delta}_{i}]\rp\left(N\ge i\right)}{\rp\left(N\ge i\right)}\\
 & =\sum_{i=0}^{\infty}\E \tilde{\Delta}_{i}=\E \sum_{i=0}^{\infty}\Delta_{i}.
\end{align*}
{The use of Fubini theorem (Theorem \ref{thm:fubini} above) in the first and last equality is justified because $\sum_{i=0}^\infty\E\norm{\Delta_i}_H<\infty$.} %\[
%\sum_{i=0}^{\infty}\frac{\E \left\Vert \tilde{\Delta}_{i}\1_{N\ge i}\right\Vert _{H}}{\rp\left(N\ge i\right)}<\infty
%\]
%by \eqref{eq:vectorValuesRhee}. 

We next show that the second moment
$\E [\tilde{Z}\otimes\tilde{Z}]$ is well-defined in {$H\otimes H$. We note that for $i\leq j$, by independence we have that
\begin{align*}
\E\left| \frac{\left\langle \tilde{\Delta}_{i},\tilde{\Delta}_{j}\right\rangle _{H}\1_{i,j\leq N}}{\rp\left(N\ge i\right)\rp\left(N\ge j\right)}\right| & \leq\frac{\left\Vert \Delta_{i}\right\Vert _{V}\left\Vert \Delta_{j}\right\Vert _{V}}{\rp\left(N\ge i\right)},
\end{align*}
which together with \eqref{eq:vectorValuesRhee} implies that \[\sum_{i,j}\E\left| \frac{\left\langle \tilde{\Delta}_{i},\tilde{\Delta}_{j}\right\rangle _{H}\1_{i,j\leq N}}{\rp\left(N\ge i\right)\rp\left(N\ge j\right)}\right|<\infty.\] This allows the application of the Fubini theorem (Theorem \ref{thm:fubini} above) to get by \eqref{eq:vectorValuesRhee} again that 
\begin{align*}
\E \left\Vert Z\otimes Z\right\Vert _{H\otimes H} & =\E \sqrt{\left\langle Z,Z\right\rangle _{H}\left\langle Z,Z\right\rangle _{H}}=\E \left\Vert Z\right\Vert _{H}^{2}\\
 & =\E \sum_{i,j}\frac{\left\langle \tilde{\Delta}_{i},\tilde{\Delta}_{j}\right\rangle _{H}\1_{i,j\leq N}}{\rp\left(N\ge i\right)\rp\left(N\ge j\right)} =\sum_{i,j}\E\frac{\left\langle \tilde{\Delta}_{i},\tilde{\Delta}_{j}\right\rangle _{H}\1_{i,j\leq N}}{\rp\left(N\ge i\right)\rp\left(N\ge j\right)}\\
& \leq\sum_{i,j}\frac{\left\Vert \Delta_{i}\right\Vert _{V}\left\Vert \Delta_{j}\right\Vert _{V}}{\rp\left(N\ge i\right)} \leq2\sum_{i\leq j}\frac{\left\Vert \Delta_{i}\right\Vert _{V}\left\Vert \Delta_{j}\right\Vert _{V}}{\rp\left(N\ge i\right)}<\infty.
\end{align*}
Finally, to get the expression for the $\tilde{\nu_i}$'s, note that the last bound enables us to use Fubini in the following calculation
\begin{align*}
\E [\tilde{Z}\otimes\tilde{Z}] & =\sum_{i,j}\frac{\E [(\tilde{\Delta}_{i}\otimes\tilde{\Delta}_{j})\1_{i,j\leq N}]}{\rp\left(N\ge i\right)\rp\left(N\ge j\right)}\\
 & =\sum_{i=0}^\infty\left(\E[ \tilde{\Delta}_{i}\otimes\tilde{\Delta}_{i}]+2\sum_{j=i+1}^{\infty}\E [\tilde{\Delta}_{i}\otimes\tilde{\Delta}_{j}]\right)/\rp\left(N\ge i\right)\\
 & =\sum_{i=0}^\infty\left(\E [\tilde{\Delta}_{i}\otimes\tilde{\Delta}_{i}]+2\E[ \tilde{\Delta}_{i}\otimes\sum_{j=i+1}^{\infty}\tilde{\Delta}_{i}]\right)/\rp\left(N\ge i\right)\\
 & =\sum_{i=0}^\infty\left(\E [\tilde{\Delta}_{i}\otimes\tilde{\Delta}_{i}]+2\E [\tilde{\Delta}_{i}]\otimes(\E Y-\E Y_{i})\right)/\rp\left(N\ge i\right),
\end{align*}
where we used the independence of the $\tilde{\Delta}_{i}$'s in the
last step. }
\end{proof}

\subsection{An example of an elliptic inverse problem}

\label{ssec:elipdisc} One simple example that fits into the framework
of section \ref{sec:is}, is the inverse problem for the diffusion
coefficient. We consider the 1-dimensional case, and in particular
the ordinary differential equation 
\begin{eqnarray*}
-\left(up^{\prime}\right)^{\prime} & = & h,\quad\text{in}\quad(0,1),\\
p(0)=p(1) & = & 0,
\end{eqnarray*}
where $u\in\state\coloneq\{u\in C(0,1):{\rm ess\, inf}_{s\in(0,1)}u(s)>0\}$
and where, for simplicity, $h\in C^{\infty}(0,1)$ has an antiderivative
which is available analytically. The observation operator $G:\state\to\R^{d}$
takes the form 
\[
G(u)=\left(p(u)(x_{1}),\dots,p(u)(x_{d})\right)
\]
where $p$ is the solution operator mapping $u$ to the solution of
the above ODE. The following discussion illustrates the choices that
have to be made to follow the programme outlined in section \ref{sec:is}.
We sketch how to derive $\beta$ in Assumption \ref{ass:indep}, $\kappa$
in Assumption \ref{ass:indepObservable}, and $\theta$ in Assumption
\ref{ass:cost}, for this example.

Using separation of variables we find that 
\begin{eqnarray}
p(a)(x) & = & {-}\int_{0}^{x}\frac{H(s)+C_{u}}{u(s)}ds,\label{eq:toyEllipticExplicit}\\
C_{u} & = & -\frac{\int_{0}^{1}\frac{H(s)}{u(s)}ds}{\int_{0}^{1}\frac{1}{u(s)}ds},\nonumber 
\end{eqnarray}
where $H$ is the anti-derivative of $h$.

We specify a uniform prior on $u$, by considering the expansion of
$u$ in the Fourier basis $\left\{ e_{i}\right\} $ for $L_{\text{Dirichlet}}^{2}\left(0,1\right),$
and similarly to (\ref{eq:is_prior}) in section \ref{sec:is}, letting
\begin{equation}
u(s)=m_{0}+\sum_{i=0}^{\infty}u_{i}e_{i}(s),\text{ with }u_{i}\overset{\text{i.i.d.}}{\sim}{\rm U}[-u_{i}^{\star},u_{i}^{\star}],\label{eq:isexampleprior}
\end{equation}
where $u_{i}^{\star}\lesssim i^{-\gamma}$ for some $\gamma>3$. Here,
$m_{0}$ is a positive constant assumed to be bigger than $\sum_{i=1}^{\infty}u_{i}^{\star}$.
We consider also the $j$-dimensional truncation of the uniform prior
\begin{equation}
\Pi_{j}u=m_{0}+\sum_{i=0}^{j}u_{i}e_{i}.\label{eq:tisexampleprior}
\end{equation}
According to the considerations in \cite[section 2.2]{DS13}, we have
that $u\in C^{2}(0,1)$ and that $u$ and $\Pi_{j}u$ are bounded
above and below by two positive constants $u_{min}$ and $u_{max}$,
almost surely in $(0,1)$.

We define the approximation $\tilde{G}_{j}$ of the observation operator
$G$, 
\[
\tilde{G}_{j}\left(u\right)=\left(p(\Pi_{j}u)(x_{1}),\dots,p(\Pi_{j}u)(x_{d})\right),
\]
based on the exact solution to the ODE as given in (\ref{eq:toyEllipticExplicit})
but using a truncation of the diffusion coefficient up to $j$ terms.
We approximate $\tilde{G}_{j}$ in turn by $G_{j}^{N}$ using the
trapezoidal rule with $N$ points to approximate the integral in the
solution formula (\ref{eq:toyEllipticExplicit}). In order to establish
Assumption \ref{ass:indep} we use the triangle inequality to obtain
\begin{equation}
\left\Vert G_{j}^{N}(u)-G\left(u\right)\right\Vert _{\R^{d}}\leq\left\Vert \tilde{G}_{j}\left(u\right)-G\left(u\right)\right\Vert _{\R^{d}}+\left\Vert G_{j}^{N}\left(u\right)-\tilde{G}_{j}\left(u\right)\right\Vert _{\R^{d}}.\label{eq:exerror}
\end{equation}
The first term can be bounded using the following lemma which we prove
at the end of this section.
\begin{lem}
\label{ellipticlem} Let $u$ as in \eqref{eq:isexampleprior}. Then
the following (deterministic) bound holds 
\[
\left\Vert \tilde{G}_{j}\left(u\right)-G\left(u\right)\right\Vert _{\R^{d}}\leq K\sqrt{\sum_{i=j+1}^{\infty}\left(u_{i}^{\star}\right)^{2}}.
\]
\end{lem}
The error in the second term in (\ref{eq:exerror}) arises due to
the use of the trapezoidal rule. The error of the trapezoidal rule
applied to $\int_{0}^{1}gdx$ is bounded by $\left\Vert g\right\Vert _{C^{2}}N^{-2}$,
\cite{Atkinson1989IntroNum}. Since $\gamma>3$ and $H\in C^{\infty}(0,1),$
it holds that all appearing integrands in the solution formula (\ref{eq:toyEllipticExplicit})
have a uniformly bounded $C^{2}$ norm, hence the second term in (\ref{eq:exerror})
is of order $N^{-2}$. We can thus combine the two bounds to obtain
\begin{eqnarray*}
\left\Vert G_{j}^{N}(u)-G\left(u\right)\right\Vert _{\R^{d}} & \lesssim & \frac{1}{N^{2}}+\sqrt{\sum_{i=j+1}^{\infty}\left(u_{i}^{\star}\right)^{2}}\\
 & \lesssim & \frac{1}{N^{2}}+j^{-\gamma+\frac{1}{2}}.
\end{eqnarray*}
The choice $N_{j}=j^{\frac{\gamma}{2}-\frac{1}{4}}$ yields that 
\[
\left\Vert G_{j}^{N_{j}}\left(u\right)-G\left(u\right)\right\Vert _{\R^{d}}\leq j^{\frac{1}{2}-\gamma}.
\]
We now set $G_{j}\coloneq G_{j}^{N_{j}}$ and see that this choice of approximation
in (\ref{eq:obsop}), satisfies Assumption \ref{ass:indep} with $\beta=\gamma-\frac{1}{2}>1,$
since $\gamma>3$.

Moreover, note that we can identify $u$ from (\ref{eq:isexampleprior})
with the sequence of its coefficients $\{u_{i}\}\in\state$, and the
truncation $\Pi_{j}u$ with the corresponding truncated sequence of
coefficients $\{u_{i}\}_{i\leq j}\in\state_{j}$, where $\state,\state_{j}$
are as in section \ref{sec:is}. Then if we want to estimate posterior
expectations of a function $f:\state\to\R$ which is $s$-H\"older continuous
for some $s\in(0,1)$, we have that 
\[
\sup_{x\in\state}\left|f(\Pi_{j}x)-f(\Pi_{\tilde{j}}x)\right|\lesssim|\Pi_{j}x-\Pi_{\tilde{j}}x|^{s}\leq(j\wedge\tilde{j})^{s(1-\gamma)},
\]
so that Assumption \ref{ass:indepObservable} is satisfied with $\kappa=2s(\gamma-1)$,
where $\kappa>1$ provided that $\gamma>\frac{1+2s}{2s}$.

Finally, the computational time $s_{i}$ to perform one step of the
chain can be bounded as follows. It is easy to check that the complexity
of evaluating $\varphi_{\ux{}{}}^{i}(x,W)$ is dominated by evaluating
the acceptance probability $\alpha_{j_{i}}$ because all other computations
have complexity of order $j_{i}.$ The evaluation of $\alpha_{j_{i}}$
boils down to evaluating $G_{j_{i}}$ which can be implemented using
the fast Fourier transform to evaluate $\Pi_{j_{i}}u$ at $j_{i}\vee N_{j_{i}}$
points at a computing time proportional to 
\[
N_{j_{i}}\log N_{j_{i}}\vee j_{i}\log j_{i}\lesssim j_{i}^{1\vee\left(\frac{\gamma}{2}-\frac{1}{4}\right)}\log j_{i}=j_{i}^{\left(\frac{\gamma}{2}-\frac{1}{4}\right)}\log j_{i},
\]
since $\gamma>3$. The trapezoidal rule has computational complexity
of order $N_{j_{i}}$ and is therefore negligible. Thus Assumption
\ref{ass:cost} is satisfied with $\theta=\frac{\gamma}{2}-\frac{1}{4}>\frac{5}{4}$,
since $\gamma>3$. Letting $r=\beta\wedge\kappa$, we get that in
order to have that $\theta<r$ as required by Assumption \ref{ass:cost},
and assuming that $s>\frac{1}{4}$, we need to have $\gamma>\frac{1-8s}{2-8s}$.
Hence for $s>\frac{1}{4}$, provided $\gamma>3\vee\frac{1-8s}{2-8s},$
we have that the Assumptions \ref{ass:indep}, \ref{ass:indepObservable},
and \ref{ass:cost} hold simultaneously.
\begin{proof}[Proof of Lemma \ref{ellipticlem}]
By definition of the $\tilde{G}_{j}$ we see that 
\[
\left\Vert \tilde{G}_{j}\left(u\right)-G\left(u\right)\right\Vert _{\R^{d}}\leq K\left|p(u)-p(\Pi_{j}u)\right|_{\infty}.
\]
Here and below, $K$ denotes a positive constant that may change between
occurrences. %In the following, we prove for $u,\tilde{u}\in C([0,1])$%with $\infty>u_{max}>u,\tilde{u}> u_{{min}}>0$, that $\left|p(u)-p(\tilde{u})\right|_{\infty}\leq K\int_{0}^{1}\left|u-\tilde{u}\right|ds\leq K\sqrt{\int_{0}^{1}\left|u-\tilde{u}\right|^{2}ds}$.%The result of the lemma follows immediately for the choice $\tilde{u}=\Pi_{j}u$.Let
$u$ and $\tilde{u}=\Pi_{j}u$ as defined in (\ref{eq:isexampleprior})
and (\ref{eq:tisexampleprior}), respectively, and recall that our
assumptions imply that $u,\tilde{u}\in C([0,1])$, and that $u$ and
$\tilde{u}$ are bounded above and below by the positive constants
$u_{min},u_{max}$, almost everywhere in $(0,1)$. We have 
\begin{align*}
\left|p(u)(x)-p(\tilde{u})(x)\right|= & \left|\int_{0}^{x}\frac{H(s)+C_{u}}{u(s)}ds-\int_{0}^{x}\frac{H(s)+C_{\tilde{u}}}{\tilde{u}(s)}ds\right|\\
\leq & \int_{0}^{x}\left|\frac{H(s)}{u(s)}-\frac{H(s)}{\tilde{u}(s)}\right|ds+\int_{0}^{x}\left|\frac{C_{u}-C_{\tilde{u}}}{u(s)}\right|ds\\
 & +\left|\int_{0}^{x}\frac{C_{\tilde{u}}}{u(s)}-\frac{C_{\tilde{u}}}{\tilde{u}(s)}\right|ds.
\end{align*}
The first term on the right hand side can be bounded as follows 
\begin{eqnarray*}
\left|\int_{0}^{x}\frac{H(s)}{u(s)}ds-\int_{0}^{x}\frac{H(s)}{\tilde{u}(s)}ds\right| & \leq & \int_{0}^{x}\frac{H(s)}{u_{\text{min}}^{2}}\left|u-\tilde{u}\right|ds\\
 & \leq & K\int_{0}^{x}\left|u-\tilde{u}\right|ds\\
 & \leq & K\sqrt{\int_{0}^{1}\left|u-\tilde{u}\right|^{2}ds}.
\end{eqnarray*}
The bound on the third term is similar. For the second term we note
\[
\left|C_{u}-C_{\tilde{u}}\right|\leq\left|\frac{\int_{0}^{1}\frac{H(s)}{u(s)}ds}{\int_{0}^{1}\frac{1}{u(s)}ds}-\frac{\int_{0}^{1}\frac{H(s)}{\tilde{u}(s)}ds}{\int_{0}^{1}\frac{1}{u(s)}ds}\right|+\left|\frac{\int_{0}^{1}\frac{H(s)}{\tilde{u}(s)}ds}{\int_{0}^{1}\frac{1}{u(s)}ds}-\frac{\int_{0}^{1}\frac{H(s)}{\tilde{u}(s)}ds}{\int_{0}^{1}\frac{1}{\tilde{u}(s)}ds}\right|
\]
The first term in the equation can be bounded similar to the considerations
above. After applying the mean value theorem to the function $\frac{1}{x}$,
noting that the integrals $\int_{0}^{1}\frac{1}{u}$ and $\int_{0}^{1}\frac{1}{\tilde{u}}$
are positive and finite since $u_{min}\leq u,\tilde{u}\leq u_{\text{max}}$,
the second term is also possible to be bounded using the same technique.
We hence have that 
\[
\norm{\tilde{G}_{j}\left(u\right)-G\left(u\right)}_{\R^{d}}\leq K\sqrt{\int_{0}^{1}|u-\Pi_{j}u|^{2}ds},
\]
and the claim follows using the orthonormality of the basis $\{e_{i}\}$
and the fact that $u_{i}\in[-u_{i}^{\star},u_{i}^{\star}]$. 
\end{proof}

\bibliographystyle{plain}
\bibliography{biblio}
 
\end{document}